\documentclass[12pt]{amsart}

\usepackage[foot]{amsaddr}
\usepackage[english]{babel}
\usepackage{amssymb}
\usepackage{mathrsfs}
\usepackage{mathtools}
\usepackage{bbm}
\usepackage{braket}
\usepackage{graphicx}
\usepackage{tikz}
\usetikzlibrary{positioning,arrows.meta}
\usepackage{xcolor}
\usepackage{enumitem}
\usepackage{url}
\usepackage{csquotes}
\usepackage{hyperref}
\usepackage{fullpage}
\usepackage{subcaption}

\hypersetup{
  colorlinks   = true,  %Colours links instead of ugly boxes
  urlcolor     = blue,  %Colour for external hyperlinks
  linkcolor    = blue,  %Colour of internal links
  citecolor    = red    %Colour of citations
}

\newcommand{\iu}{\mathrm{i}}

\DeclareMathOperator{\supp}{supp}

\newcommand{\KL}{D_\mathtt{KL}}
\newcommand{\fdiv}{D_f}
\newcommand{\dinf}{D_\infty}
\newcommand{\Um}{D_\mathtt{U}}
\newcommand{\BS}{D_\mathtt{BS}}
\newcommand{\BSunr}{D^{\mathtt{UNR}}_{\mathtt{BS}}}
\newcommand{\CB}{\mathtt{CB}}

\newcommand{\di}{{\rm d}}
\newcommand{\Der}{{\rm D}}
\newcommand{\R}{\mathbb{R}}

\newcommand{\calF}{\mathcal{F}}
\newcommand{\calP}{\mathcal{P}}

\newcommand{\scrB}{\mathscr{B}}
\newcommand{\scrP}{\mathscr{P}}
\newcommand{\scrL}{\mathscr{L}}
\newcommand{\scrS}{\mathscr{S}}
\newcommand{\scrU}{\mathscr{U}}

\newcommand{\sfK}{\mathsf{K}}

\newcommand{\frH}{\mathfrak{h}}
\newcommand{\frX}{\mathfrak{X}}

\newcommand{\tr}{\mathrm{tr}}

\mathtoolsset{showonlyrefs}
\numberwithin{equation}{section}

\DeclareMathOperator*{\esssup}{ess\,sup}

\DeclareMathOperator{\Tr}{Tr}
\DeclareMathOperator{\rank}{rank}
\DeclareMathOperator{\spec}{spec}

\newtheorem{theorem}{Theorem}[section]
\newtheorem*{theorem*}{Theorem}
\newtheorem{proposition}[theorem]{Proposition}
\newtheorem{lemma}[theorem]{Lemma}
\newtheorem{corollary}[theorem]{Corollary}

\theoremstyle{definition}
\newtheorem{definition}[theorem]{Definition}
\newtheorem{property}{Property}

\theoremstyle{remark}
\newtheorem{remark}[theorem]{Remark}

\newcommand{\tue}{Eindhoven University of Technology, P.\ O.\ Box 513,
  5600 MB Eindhoven, The Netherlands}
\newcommand{\addrmath}{Department of Mathematics and Computer Science, \tue}
\newcommand{\addrphys}{Department of Applied Physics and Science Education, \tue}

\begin{document}

\title{Strengthening the quantum data processing inequality for maximal $f$-divergences via unraveling}

\author{R. J. P. T. de Keijzer}
\address[R. J. P. T. de Keijzer]{Unaffiliated, Eindhoven, The Netherlands}

\author{M. Ruibal Ortigueira}
\address[M. Ruibal Ortigueira]{Department of Statistics, University of Oxford, Oxford, UK}

\author{O. Tse}
\email{o.t.c.tse@tue.nl}
\address[O. Tse, L. Y. Visser]{\addrmath}
% \thanks{Corresponding author: \texttt{o.t.c.tse@tue.nl}}

\author{L. Y. Visser}
%\address[L. Y. Visser]{\addrmath}

\author{S. J. J. M. F. Kokkelmans}
\address[S. J. J. M. F. Kokkelmans]{\addrphys}


\date{}

\begin{abstract}
Maximal $f$-divergences, among them the Belavkin--Staszewski relative entropy, are non-increasing under completely positive trace-preserving maps, which is the quantum data processing inequality (DPI). We quantify this decrease for Lindblad evolutions with Hermitian jump operators by lifting the problem to a stochastic unraveling, in which the two states become measures on pure states evolving under a classical Fokker--Planck equation. The lift rests on the identity that the infimum of the classical $f$-divergence over ensembles realizing two states is equal to their maximal $f$-divergence. This infimum is attained on a common basis of pure states, which is admissible initial data for the unraveling.

The result is a strengthened DPI, with the loss bounded below by the time integral of a relative $f$-Fisher information. Given a two-sided bound on the unraveling's transition kernel, we derive two exponential decay estimates. For bracket-generating jump operators, the rate reaches twice the spectral gap of the symmetrized generator after an explicit delay. Without bracket generation, the kernel bound alone gives any rate below twice the spectral gap of the full generator.
\end{abstract}
\maketitle
%\tableofcontents

\section{Introduction}

Relative entropy is an important notion in both classical probability theory and quantum mechanics. It has many applications in measure theory, large deviation theory, and information theory \cite{relativeentropybackground}. This asymmetric functional takes two objects from the same space and outputs a non-negative number indicating their dissimilarity. In the classical case, the inputs are probability distributions, while in the quantum case, they are states. This quantity is useful because applying the same map to both arguments---whether a measurement, a noisy channel, or simply allowing the system to evolve for a while---can only decrease it. This is the \emph{data processing inequality} (DPI), which is the focus of this work. For classical distributions, the relative entropy is also known as the \emph{Kullback--Leibler} (KL) \emph{divergence} $\KL$ \cite{kullback}. Using an axiomatic approach in which the entropy measure has to satisfy several properties such as monotonicity and positivity, Alfr\'{e}d R\'{e}nyi succeeded in deriving the most general form of classical entropy, the R\'{e}nyi relative entropy or R\'{e}nyi divergence \cite{renyientropy2}.

The \emph{Umegaki relative entropy} $\Um$, introduced in \cite{Umegaki1962ConditionalEI}, is the earliest notion of relative entropy for non-commutative algebras, and the most studied form of quantum relative entropy because of its operational interpretation in quantum hypothesis testing via quantum versions of Stein's lemma \cite{quantumstein2,steinslemma1} and Sanov's theorem \cite{quantumsanov1}. Since then, several well-established notions of quantum relative entropy have been introduced, satisfying properties such as additivity and unitary invariance \cite{renyientropy}. Among these is the quantum Petz--R\'{e}nyi entropy, introduced by Petz \cite{PETZ198657} in the effort to establish the most general quantum relative entropy, which produces the Umegaki relative entropy as an asymptotic limit.

A different generalization of R\'{e}nyi entropy is the maximal (or geometric) R\'{e}nyi entropy, introduced later in \cite{Matsumoto2018AF-Divergenceb}. In that same limit, the maximal R\'{e}nyi entropy converges to an entropy that is of special interest in this work, the \emph{Belavkin--Staszewski} (BS) relative entropy $\BS$ \cite{quantumchannelBS1}. The BS relative entropy was first introduced generally in \cite{belavkinstaszewski} and in operator form in \cite{belavkinoperator}. Recent interest in BS entropy has increased because of its application in quantum channel capacities \cite{strengtheneddpi,quantumchannelBS1}. For simultaneously diagonalizable states, all these various quantum relative entropies reduce to the classical KL divergence of the diagonal entries. The Umegaki and the Belavkin--Staszewski entropies both satisfy a quantum variant of the data processing inequality \cite{strengtheneddpi,Hayden2004}, as does every maximal $f$-divergence \cite{Matsumoto2018AF-Divergenceb}---this last family is the one to which the results in this manuscript apply.

Here, we study channels that arise in a specific physical setting, namely noisy quantum systems affected by classical noise. This noise and the resulting loss of coherence can greatly reduce the accuracy of current quantum computers and simulators \cite{fvqoc,ssetheory1}. In an ideal closed quantum system, the state $|\psi\rangle$ changes over time according to a Hamiltonian $H$ through the Schr\"{o}dinger equation, while remaining pure. In practice, however, quantum systems interact with their surroundings. These interactions cause decoherence, so the system must be described as an open system. Interactions between the system and its environment through system-bath couplings drive the system into a mixed state, represented by a density matrix $\rho$, whose dynamics are often described by the Lindblad equation \cite{lindblad1}
\begin{equation}
\label{eq:lindblad}\tag{\textsf{LiE}}
    \partial_t\rho =\mathcal{L}(\rho)=-\iu[H,\rho]+\sum_j \gamma_j^2S_j \rho S_j^\dagger-\frac{1}{2} \gamma_j^2\{S_j^\dagger S_j,\rho\},\quad \rho(0)=\rho_0, 
\end{equation}
where $S_j$ are the Lindblad jump operators, and $\gamma_j\geq0$ are loss rates. Throughout this work, we assume the jump operators are Hermitian, $S_j=S_j^\dagger$. 

Interest in unravelings of the Lindblad equation has increased in recent years, especially for physically inspired instances \cite{unraveltome,metounravel1,physicalunravel2,physicalunravel3}. These unravelings are useful for simulation and provide a deeper understanding of system-bath couplings. A specific example is the stochastic Schr\"{o}dinger equation (SSE) \cite{sse1} (see also \cite{fvqoc,ssetheory1,physicalunravel3,colorednoisepaper}), introduced as a physically inspired unraveling that describes classical control noise. The SSE is given by  
\begin{equation}
\label{eq:sse}\tag{\textsf{SSE}}
    \di |\psi\rangle=-\iu H|\psi\rangle\,\di t-\frac{1}{2}\sum_j \gamma_j^2 S_j^\dagger S_j|\psi\rangle\,\di t+\iu\sum_j \gamma_j S_j |\psi\rangle\, \di W_{j},
\end{equation}
where $W_j$ are independent Brownian motions. For arbitrary jump operators, the density operator $\mathbb{E}[|\psi\rangle\!\langle\psi|]$ solves \eqref{eq:lindblad} and $\di\|\psi\|^2=\iu\sum_j\gamma_j\braket{\psi|(S_j-S_j^\dagger)|\psi}\di W_j$ is a martingale, so the norm is preserved in expectation. Clearly, if $S_j$ is Hermitian, then the integrand vanishes. Moreover, the term $-\tfrac12\gamma_j^2S_j^2\ket{\psi}$ in \eqref{eq:sse} is precisely the It\^{o}--Stratonovich conversion term. Hence, the Stratonovich form of \eqref{eq:sse} reads
\begin{equation}\label{eq:sseStrat}\tag{\textsf{SSE}$'$}
  \di\ket{\psi} = -\iu H\ket{\psi}\,\di t + \iu\sum\nolimits_j\gamma_jS_j\ket{\psi}\circ\di W_j.
\end{equation}
Since Stratonovich differentials obey the ordinary chain rule, we directly deduce $\di\|\psi\|^2=0$, i.e., the norm is conserved along every trajectory. The dissipation in \eqref{eq:lindblad} is then an effect of averaging rather than of a system-bath coupling. In this sense, restricting to Hermitian $S_j$ is both a modeling and a technical choice. Note that all unravelings of this form induce a classical probability distribution $\mu=\text{Law}(P)$ on the space of pure states \cite{colorednoisepaper}.

\subsubsection*{Main results}

Every maximal $f$-divergence $\fdiv^{\max}$ is non-increasing under completely positive trace-preserving (CPTP) maps. The question we address is by how much, for the channels $\Phi_t=e^{t\mathcal{L}}$ generated by \eqref{eq:lindblad} with Hermitian jump operators. We prove three results, each for any normalized operator convex $f$ (cf.\ Section~\ref{sec:qrelent}) and every pair of faithful initial states.
\begin{enumerate}[label=(\arabic*)]
  \item \emph{A strengthened data processing inequality} (Theorem~\ref{thm:contraction}). The loss $\fdiv^{\max}(\rho_0\|\sigma_0)-\fdiv^{\max}(\rho_t\|\sigma_t)$ is bounded below by the integral over $[0,t]$ of a non-negative dissipation rate given by the relative $f$-Fisher information of the unraveling.
  \item \emph{Decay under bracket generation} (Theorem~\ref{thm:decay}, Corollary~\ref{cor:lsidecay}). If the jump operators alone generate a transitive action on the pure states, and the transition kernel of the unraveling satisfies a two-sided bound, then $\fdiv^{\max}(\rho_t\|\sigma_t)$ decays exponentially, at a rate that reaches twice the spectral gap $\lambda_{\mathtt{sym}}$ of the symmetrized generator after an explicit time delay.
  \item \emph{Decay without a functional inequality} (Theorem~\ref{thm:sharpchi}). The kernel bound alone yields exponential decay at rate $2\lambda$ for any $\lambda$ below the spectral gap of the full generator. The Hamiltonian enters this gap, so it may provide a rate that the jump operators alone cannot. 
\end{enumerate}

\subsubsection*{The approach}

Estimates of this kind are usually established by working directly with the quantum semigroup, where the necessary functional inequalities are hard to obtain. We instead lift the problem to the unraveling, as summarized in Figure~\ref{fig:lift}. The pair of states $(\rho_t,\sigma_t)$ is replaced by a pair of measures $(\mu_t,\nu_t)$ on pure states, and the Lindblad flow is replaced by the Fokker--Planck equation \eqref{eq:fokker-planck} associated with \eqref{eq:sse}. The estimate is proved for the measures and then transferred back to the states, so everything analytic in this route is classical. Two ingredients make the lift possible.

\begin{figure}[h]
\centering
\begin{tikzpicture}[>={Stealth[length=2.4mm]}, thick]
  \node (q0) at (0,2.5)   {$(\rho_0,\sigma_0)$};
  \node (qt) at (10,2.5)  {$(\rho_t,\sigma_t)$};
  \node (c0) at (0,0)     {$(\mu_0,\nu_0)$};
  \node (ct) at (10,0)    {$(\mu_t,\nu_t)$};

  \draw[->] (q0) -- node[above=1pt]{\footnotesize Lindblad flow \eqref{eq:lindblad}} (qt);
  \draw[->] (c0) -- node[below=1pt]{\footnotesize Fokker--Planck \eqref{eq:fokker-planck}} (ct);
  \draw[<-] (c0) -- node[left=2pt]{$\CB$} (q0);
  \draw[->] (ct) -- node[right=2pt]{$\Lambda$} (qt);

  \node[anchor=west, align=left] at (0.6,1.25)
    {\footnotesize Theorem~\ref{thm:maintheorem}};
  \node[anchor=east, align=right] at (9.4,1.25)
    {\footnotesize Definition~\ref{def:unr}};
\end{tikzpicture}
\caption{\emph{The lift.} A pair of states $(\rho_0,\sigma_0)$ is represented by a pair of measures $(\mu_0,\nu_0)=(\mu_{\CB},\nu_{\CB})$ on the pure states, supported on a common basis, propagated by a classical Fokker--Planck equation \eqref{eq:fokker-planck}, and mapped back to states by the barycentre map $\Lambda$ (cf.\ Definition~\ref{def:realization}). The left vertical arrow is an equality provided by Theorem~\ref{thm:maintheorem} on a common basis of the pair $(\rho_0,\sigma_0)$. The right one is only an inequality, with the propagated measures realizing $(\rho_t,\sigma_t)$ without requiring minimization. Reading the square from top left to top right via the bottom horizontal line therefore bounds $\fdiv^{\max}(\rho_t\|\sigma_t)$ by a quantity computed entirely classically.}
\label{fig:lift}
\end{figure}

The first is a bridge from measures back to states, constructed in Section~\ref{sec:results_marcos}. We define the unravel $f$-divergence $\fdiv^{\mathtt{UNR}}(\rho\|\sigma)$ as the infimum of the classical $f$-divergence over all pairs of probability measures on $\scrP(\frH)$ that realize $\rho$ and $\sigma$, and show that it equals $\fdiv^{\max}(\rho\|\sigma)$ (Theorem~\ref{thm:maintheorem}). For $f=f_\mathtt{KL}$, this reads $\BSunr=\BS$. Any pair of measures realizing $(\rho_t,\sigma_t)$ therefore bounds $\fdiv^{\max}(\rho_t\|\sigma_t)$ from above. The infimum is attained on a \emph{common basis} of $(\rho,\sigma)$, a possibly non-orthogonal basis of $\frH$ in which both states are diagonal, and a single minimizing pair serves every $f$ at once.

The second ingredient is admissibility, which Section~\ref{sec:dissipation} puts to work. A common basis consists of pure states, so the minimizing pair is admissible initial data for \eqref{eq:sse}. Propagating it along the Fokker--Planck equation, the classical data processing inequality with its dissipation term gives result (1). This becomes a rate either through a Poincar\'{e} or logarithmic Sobolev inequality for the unraveling, which gives result (2), or through the $L^2$ decay of the Fokker--Planck semigroup, which the kernel bound provides and which gives result (3).

\subsubsection*{Relation to previous work}

The literature already provides a variational description of $\fdiv^{\max}$. Matsumoto \cite{Matsumoto2018AF-Divergenceb} obtains it as an infimum of the classical $f$-divergence $\fdiv(p\|q)$ over the \emph{reverse tests} of $(\rho,\sigma)$, the triples $(\Gamma,\{p,q\})$ in which $p$ and $q$ are positive measures on a finite set and $\Gamma$ is a trace-preserving positive map with $\Gamma(p)=\rho$ and $\Gamma(q)=\sigma$. The minimizing common basis recovers, in a different form, the reverse test constructed in \cite{explicitmax}. What is new in Theorem~\ref{thm:maintheorem} is that the infimum is attained over all Borel measures on the continuum $\scrP(\frH)$, which is where the laws generated by \eqref{eq:sse} live, and that the minimizer is admissible for the unraveling.

Strengthenings of data processing have been obtained before, of a different type. Carlen and Vershynina \cite{carlenvershynina2018,carlenvershynina2020} bound the loss of a standard $f$-divergence from below by the failure of a recovery map to reconstruct the input, and Bluhm and Capel \cite{strengtheneddpi} prove the corresponding statement for $\BS$ and for a class of maximal $f$-divergences. Result (1) is also a lower bound on the loss, but its remainder is of a different nature: a dissipation integral rather than a recovery defect. It is confined to the flow generated by \eqref{eq:lindblad}, and since it is additive in time, it iterates, which converts it into a rate. We give up intrinsicness: our remainder is a functional of the unraveling rather than of $\rho_t$ and $\sigma_t$.

Decay rates for this same class of Lindbladians have also been obtained without unraveling. Gao, Junge, LaRacuente, and Li \cite{gaojungelaracuenteli} prove that the sub-Laplacian of a H\"{o}rmander system on a compact manifold has a positive complete modified logarithmic Sobolev constant and transfer this result to the quantum setting. The scalar statement is older: \L{}ugiewicz and Zegarli\'{n}ski \cite[Thm.~3.1]{lugiewiczzegarlinski} establish a tight logarithmic Sobolev inequality for any H\"{o}rmander system on a compact connected manifold without boundary, under a dimension hypothesis excluding $n=2$; the contribution of \cite{gaojungelaracuenteli} is the matrix-valued version. With the jump algebra $\mathrm{Lie}_\R\{\iu S_j\}$ and the defining representation, this covers the dissipative part of \eqref{eq:lindblad}, namely, systems with $H=0$. The two assertions differ. Their result concerns the relaxation of a single state to the fixed point, measured in Umegaki relative entropy, with a constant obtained by transference. Ours compares two evolving states, applies to every maximal $f$-divergence, and allows an arbitrary Hamiltonian, which the symmetry required there excludes. By result (3), the Hamiltonian may even supply the rate. In neither case is the spectral-gap constant explicit, but our delays are given in terms of the transition kernel of the unraveling. The approaches are opposite in spirit: their work treats the quantum semigroup directly, using tools from operator space theory. In the same spirit, Junge, LaRacuente, and Rouz\'{e} \cite{jungelaracuenterouze2023} generalize the Holley--Stroock perturbation argument to finite-dimensional quantum Markov semigroups, allowing inequalities proved for self-adjoint processes to be applied to merely state-preserving ones; the classical form of that argument is discussed in Remark~\ref{rem:holleystroock}.

\subsubsection*{Outline}

Section~\ref{sec:background} collects the necessary definitions and classical facts. Section~\ref{sec:results_marcos} constructs the bridge: it introduces $\fdiv^{\mathtt{UNR}}$, constructs the common basis, and proves Theorem~\ref{thm:maintheorem}. Section~\ref{sec:dissipation} contains the dynamical results: the strengthened data processing inequality in Section~\ref{sec:contraction}, the functional inequalities and the decay under bracket generation in Sections~\ref{sec:lsi} and~\ref{sec:decay}, the decay without a functional inequality in Section~\ref{sec:sharprate}, and two qubit examples in Section~\ref{sec:examples}. Section~\ref{sec:discussion} ends with open questions. Appendix~\ref{app:commonbasisproofs} gives an explicit construction of the common basis and the remaining proofs of Section~\ref{sec:results_marcos}. Appendix~\ref{sec:largedev} motivates Definition~\ref{def:unr} by showing that the infimum over ensembles arises as the contraction of a large deviation rate function along $\Lambda$, with $\BS$ the least exponential cost over ensembles realizing $\sigma$. Appendix~\ref{app:lindbladflow} derives the dissipation identity behind Theorem~\ref{thm:contraction}.

\section{Background and preliminary results}
\label{sec:background}
This section introduces several concepts in classical probability theory and quantum information needed to prove the results in Section~\ref{sec:results_marcos}. Throughout, we consider a finite-dimensional Hilbert space $\frH$ of dimension $n$ (for $m$ qubits $n=2^m$). We denote the space of bounded linear operators on $\frH$ by $\scrB(\frH)$, and the set of density operators, henceforth called \emph{states}, by
\[
    \scrS(\frH) := \bigl\{ \rho\in \scrB(\frH): \rho^\dagger=\rho,\; \rho \succeq 0,\;\tr[\rho]=1\bigr\}.
\]
We equip $\scrS(\frH)$ with the trace distance $d_\mathtt{TR}$, given by \[d_\mathtt{TR}(\rho,\sigma)=\frac{1}{2}\|\rho-\sigma\|_1=\frac{1}{2}\Tr[\sqrt{(\rho-\sigma)^\dagger(\rho-\sigma)}],\] making $(\scrS(\frH),d_\mathtt{TR})$ a metric space. We further denote $\scrS_+(\frH)$ as the set of positive definite density operators ($\rho\succ 0$), henceforth called \emph{faithful states}. States that are rank-$1$ projection operators ($\rho^2=\rho$) are called \emph{pure states}. We denote the set of pure states by $\scrP(\frH)$, i.e., 
\[
     \scrP(\frH):=\bigl\{P\in\scrS(\frH) : P^2=P\bigr\}.
\]

Since $\tr[P]=1$ and $P^2=P$, every $P\in\scrP(\frH)$ is an orthogonal rank-1 projection. Pure states are denoted throughout by $P,Q\in\scrP(\frH)$, and we write $|\psi\rangle\in\frH$ for a unit vector \emph{representing} $P$, that is, $P=|\psi\rangle\!\langle\psi|$. A representative exists for every $P$ and is unique up to a phase, since $|\psi\rangle$ and $e^{\iu\theta}|\psi\rangle$ determine the same $P$. Accordingly, $\scrP(\frH)$ is the quotient of the unit sphere of $\frH$ by the global phase, and the Fubini--Study distance below is a metric on this quotient. Objects depending on the representative choice, not just on $P$, such as the dual bases of Section~\ref{sec:bssseeq}, are always introduced with an explicit choice.

\medskip
The following compactness result for the set of pure states is proven in \cite{manifolds1997introduction}.

\begin{proposition}[Compactness of $\scrP(\frH)$]
\label{prop:fubinistudy}
    Let $\frH$ be a finite-dimensional Hilbert space with inner product $\braket{\cdot|\cdot}$. Then $\scrP(\frH)$ is compact under the Fubini--Study metric 
\begin{equation}
    d_{\mathtt{FS}}(P,Q) = \arccos\sqrt{\tr[PQ]},\qquad P,Q\in\scrP(\frH).
\end{equation}
In terms of unit vectors $|\psi\rangle$ and $|\varphi\rangle$ representing $P$ and $Q$ respectively, one has that $\tr[PQ]=|\langle\psi|\varphi\rangle|^2$, which gives $d_\mathtt{FS}(P,Q)=\arccos|\langle\psi|\varphi\rangle|$.
\end{proposition}
Henceforth, we equip $\scrP(\frH)$ with the Fubini--Study metric $d_{\mathtt{FS}}$, thus making $(\scrP(\frH),d_{\mathtt{FS}})$ a compact metric space.
More than a metric space, $\scrP(\frH)$ is a Riemannian manifold, and we will need its volume measure. It is a smooth compact manifold of real dimension $2(n-1)$, where $n:=\dim\frH$, diffeomorphic to the complex projective space $\mathbb{CP}^{n-1}$, on which the group $\scrU(\frH)$ of unitaries on $\frH$ acts smoothly and transitively by conjugation, $P\mapsto UPU^\dagger$. Transitivity holds because any two unit vectors of $\frH$ are carried into one another by a unitary, hence so are any two rank-one projections where the stabilizer of $P=|\psi\rangle\!\langle\psi|$ consists of the unitaries preserving the line $\mathbb{C}\!\ket{\psi}$. Thus, $\scrP(\frH)$ is the quotient of $\scrU(\frH)$ by the stabilizer of any one pure state. 

\begin{definition}[Fubini--Study metric and measure]\label{def:fsmetric}
The \emph{Fubini--Study metric} $\mathbf{g}_\mathtt{FS}$ is the unique $\scrU(\frH)$-invariant Riemannian metric on $\scrP(\frH)$, up to a positive factor, and normalized so that its geodesic distance is the $d_\mathtt{FS}$ of Proposition~\ref{prop:fubinistudy}. The \emph{Fubini--Study measure} $\nu_\mathtt{FS}\in\calP(\scrP(\frH))$ is the corresponding normalized Riemannian volume.
\end{definition}

Since $\mathbf{g}_\mathtt{FS}$ is $\scrU(\frH)$-invariant, so is its Riemannian volume, and since the action is transitive, no other such measure exists. Thus, $\nu_\mathtt{FS}$ is the only $\scrU(\frH)$-invariant probability measure on $\scrP(\frH)$, which coincides with the push-forward of the Haar measure on $\scrU(\frH)$ under $U\mapsto UP_0U^\dagger$, for any fixed $P_0\in\scrP(\frH)$, which is invariant as well.

\begin{definition}[Killing field]\label{def:killing}
A vector field $X$ on $\scrP(\frH)$ is a \emph{Killing field} if its flow acts by isometries of $\mathbf{g}_\mathtt{FS}$, and \emph{$\nu_\mathtt{FS}$-divergence free} if $\int XF\,\di\nu_\mathtt{FS}=0$ for every smooth $F$.
\end{definition}

A Killing field is $\nu_\mathtt{FS}$-divergence free and its flow preserves the Riemannian volume. We refer to \cite{bengtssonzyczkowski} for the Fubini--Study metric and measure, and to \cite{haarmeasure} for the Haar measure.

\subsection{Classical relative entropies}
We denote by $\calP(\frX)$ the set of probability measures on a measurable space $(\frX,\calF)$. We adapt the following definitions from \cite{Polyanskiy2025InformationLearning}. The results of this work are stated for $f$-divergences, of which the Kullback--Leibler divergence is the best-known example.
\begin{definition}[$f$-divergence]
\label{def:fdiv}
Let $f:(0,\infty)\rightarrow \mathbb{R}$ be a convex function with $f(1)=0$ and $f(0):=\lim_{x\downarrow 0}f(x)$, the limit existing in $(-\infty,+\infty]$ because a convex function on $(0,+\infty)$ lies above its tangent line at $1$. Let $\mu,\nu$ be two probability measures on a measurable space $(\frX, \calF)$. Then the $f$-divergence of $\mu$ with respect to $\nu$ is defined as
\begin{equation}
    \fdiv(\mu \| \nu):= \begin{cases}
        \displaystyle\int_\frX f\biggl(\frac{\di \mu}{\di \nu}\biggr)  \nu(\di x)  & \text{if } \mu \ll \nu, \\
        +\infty & \text{else.}
        \end{cases}
\end{equation}
\end{definition}
The KL divergence $\KL$ is obtained for $f_\mathtt{KL}(r):=r\log(r)-r + 1$, and the $\chi^2$-divergence is obtained for $f_\mathtt{2}(r)=(r-1)^2$. These are the two cases carried in Section~\ref{sec:dissipation}. All $f$-divergences enjoy the (classical) Data Processing Inequality (DPI). To understand the DPI, we introduce Markov kernels and composition measures \cite{Cinlar2011ProbabilityStochastics}.
\begin{definition}[Markov kernel]
\label{def:markovkernel}
    Let $(\frX, \calF)$ be a measurable space. A Markov kernel $\sfK : \frX {\times} \calF \to [0,1]$ on $(\frX, \calF)$ is a map satisfying the following properties:
    \begin{enumerate}[label=(\arabic*)]
        \item For every $x \in \frX$, $\sfK(x,\cdot)\in\calP(\frX)$;
        \item For any measurable set $A \in \calF$, the map $ x \mapsto \sfK(x,A)$ is $\calF$-measurable.
    \end{enumerate}
\end{definition}
\begin{definition}[Composition measure]
    Let $\mu\in\calP(\frX)$ and $\sfK$ be a Markov kernel on a measurable space $(\frX, \calF)$. The composition measure $\sfK{\circ}\mu\in\calP(\frX)$ is defined as
    \begin{equation}
        \sfK{\circ}\mu(A) := \int_\frX \sfK(x,A)\mu(\di x),\qquad A\in\calF.
    \end{equation}
\end{definition}

\begin{proposition}[DPI for $f$-divergences]\label{prop:finitedpi}
    Let $\sfK$ be a Markov kernel on a measurable space $(\frX, \calF)$, and let $f$ be as in Definition~\ref{def:fdiv}. Then,
    \begin{equation}
        \fdiv(\sfK{\circ}\mu\|\sfK{\circ}\nu) \leq \fdiv(\mu\| \nu)\qquad\text{for every $\mu,\nu\in\calP(\frX)$.}
    \end{equation}
\end{proposition}
Intuitively, this implies that two probability measures become harder to distinguish when the same operation is applied to them.

\subsection{Quantum relative entropies}\label{sec:qrelent}
As mentioned in the introduction, several inequivalent extensions of the classical $f$-divergence to quantum states exist. This work focuses on the maximal $f$-divergence in Definition~\ref{def:fmax}, which is the largest and contains $\BS$. 

The classical objects in the previous subsection relate to quantum states via a single map, in which a quantum state is expressed as an average of pure states. Each such decomposition is a probability measure on $\scrP(\frH)$ whose barycentre is the state (cf.\ Definition~\ref{def:realization}). When two quantum states admit decompositions over one and the same set of pure states, that set is called a \emph{common basis} (cf.\ Definition~\ref{def:commonbasis}). 

\begin{definition}[Realization of states]
\label{def:realization}
    Let $\rho \in \scrS(\frH)$ be a state. A measure $\mu\in\calP(\scrP(\frH))$ is said to \textit{realize} $\rho$ if
    \[
        \mathcal{P}(\mathscr{P}(\frH))\ni \mu\mapsto  \Lambda(\mu):= \int_{\scrP(\frH)} P \,\mu(\di P) = \rho \in \scrS(\frH).
    \]
\end{definition}

\begin{definition}[Common basis]
\label{def:commonbasis}
    Let $\rho\in\scrS(\frH)$ and let $\sigma\in\scrS_+(\frH)$ be faithful. A (possibly non-orthogonal) basis $\{\psi_j\}_{j=1}^n$ of $\frH$ is said to be a \emph{common basis} for the pair $(\rho,\sigma)$ if there are coefficients $\rho_j\in[0,1]$ and $\sigma_j\in(0,1]$ for $j \in \{1,...,n\}$ such that, writing $P_j:=|\psi_j\rangle\langle\psi_j|\in\scrP(\frH)$ for the associated pure states,
    \begin{equation}
        \rho = \sum_{j=1}^n \rho_j P_j, \quad \sum_{j=1}^n \rho_j = 1,\qquad \sigma = \sum_{j=1}^n \sigma_j P_j, \quad \sum_{j=1}^n \sigma_j = 1.
    \end{equation}
We denote the set of all common bases for a state $\rho\in\scrS(\frH)$ and a faithful $\sigma \in \scrS_+(\frH)$ by $\CB(\rho, \sigma)$, shortened to $\CB$.
\end{definition}

We remark that on a common basis for $(\rho,\sigma)$, one can associate measures $\mu_{\CB} := \sum_j \rho_j\delta_{P_j}$ and $\nu_{\CB} := \sum_j \sigma_j\delta_{P_j}$, that realize $\rho$ and $\sigma$, such that
\[
    \fdiv(\mu_{\CB}\|\nu_{\CB})=\sum_{j=1}^n f\biggl(\frac{\rho_j}{\sigma_j}\biggr)\sigma_j.
\]

\begin{definition}[Quantum relative entropy]
    Let $\rho,\sigma$ be states in $\scrS(\frH)$. Two versions of the quantum relative entropy \cite{Umegaki1962ConditionalEI,belavkinstaszewski} are given by
    \begin{equation}
    \begin{aligned}
        \Um(\rho\|\sigma) &:=\begin{cases}
        \Tr[\rho \log (\rho)-\rho\log(\sigma)] &\text{if } \supp(\rho) \subseteq \supp(\sigma),\\
        +\infty & \text{else}.\end{cases} \,\,\,&\text{(Umegaki)}\\
        \BS(\rho\|\sigma) &:=\begin{cases}
        \Tr[\rho \log (\sqrt{\rho} \sigma^{-1} \sqrt{\rho})] &\,\,\,\,\,\text{if } \supp(\rho) \subseteq \supp(\sigma),\\
        +\infty & \,\,\,\,\,\text{else}.\end{cases}\,\,\,\,&\text{(Belavkin--Staszewski)}
    \end{aligned}
    \end{equation}
    Note the resemblance between the quantum relative entropy and the KL divergence. In fact, when $\rho$ and $\sigma$ commute, $\Um=\BS$, and they are exactly the KL divergence $\KL$ between the eigenvalues in the shared eigenbasis. In \cite{Hiai1991TheProbability}, it is proven that $\Um\leq \BS.$ 
\end{definition}

Because quantum information is non-commutative, a given operator-convex $f$ can define different quantum $f$-divergences since the classical integrand $f(\di\mu/\di\nu)$ has no unique non-commutative counterpart. The two entropies above are what $f_\mathtt{KL}$ gives for two such choices. In fact, the Belavkin--Staszewski entropy is part of a larger family of \emph{maximal (quantum) $f$-divergences}, first introduced in \cite{PetzRuskaiMaximal} and named by Matsumoto in \cite{Matsumoto2018AF-Divergenceb}.
\begin{definition}[Maximal $f$-divergence]
\label{def:fmax}
    Let $f:(0,+\infty) \to \mathbb{R}$ be an operator-convex function with $f(1)=0$. Let $\rho,\sigma\in\scrS_+(\frH)$. Then the maximal $f$-divergence is defined as
    \begin{equation}
       \fdiv^{\max}(\rho \| \sigma) := \Tr\left[\sigma f\Bigl(\sigma^{-1/2} \rho\sigma^{-1/2}\Bigr) \right].
    \end{equation}
    For non-invertible $\rho, \sigma \in \scrS(\frH)$, define
    \begin{equation}
        \fdiv^{\max}(\rho \| \sigma) := \lim_{\epsilon\to 0} \fdiv^{\max}(\rho + \epsilon I\| \sigma + \epsilon I).
    \end{equation}
    When $\sigma$ is faithful and $\rho$ is not, this limit is again given by functional
    calculus, evaluated with the convention $f(0)$ of Definition~\ref{def:fdiv}, which may take the value
    $+\infty$.
\end{definition}
For $f_\mathtt{KL}$, one recovers $\BS$, and for $f_\mathtt{2}$ the maximal $\chi^2$-divergence $D^{\max}_{\chi^2}(\rho\|\sigma)=\Tr[\rho\sigma^{-1}\rho]-1$, a member of the family of quantum $\chi^2$-divergences studied in \cite{temmechi2}. 
Moreover, the reverse-test description mentioned in the introduction extends to maximal $f$-divergences. Matsumoto \cite{Matsumoto2018AF-Divergenceb} defines $\fdiv^{\max}$ as the infimum of the classical $f$-divergence $\fdiv(p\|q)$ over the reverse tests of the pair $(\rho,\sigma)$. The definition used in Section~\ref{sec:results_marcos} is of this type, but the measures need not be finitely supported.

\subsubsection*{Normalization.} Note that adding an affine term $c(x-1)$ to $f$ changes none of the functionals considered here. For a maximal $f$-divergence this is $\Tr[\sigma(\sigma^{-1/2}\rho\,\sigma^{-1/2}-\mathbbm{1})]=\tr[\rho]-\tr[\sigma]=0$, while for a classical one it is $\int(\di\mu/\di\nu-1)\,\di\nu=0$. Therefore, we may assume without loss of generality that
\begin{equation}\label{eq:normalization}
  f(1)=0\qquad\text{and}\qquad f'(1)=0 ,
\end{equation}
and we call an operator convex $f$ satisfying \eqref{eq:normalization} a \emph{normalized operator convex} function.

Normalized operator convex functions admit a classical integral representation, going back to L\"{o}wner's theory of operator monotone functions \cite{ReviewHiaiMosonyi}. We recall it here as follows.

\begin{proposition}[Integral representation]\label{prop:opconv}
Let $f:(0,\infty)\to\R$ be normalized, operator convex, and not affine. Then there are $b\geq0$ and a positive measure $\tau_f$ on
$[0,\infty)$ such that
\begin{equation}\label{eq:opconvrep}
  f(r)=(r-1)^2\,g_f(r),\qquad g_f(r):=b +\int_{[0,\infty)}\frac{1}{r+s}\,\tau_f(\di s) .
\end{equation}
In particular, $g_f$ is positive, continuous and non-increasing on $(0,\infty)$ with $g_f(1)=\tfrac12f''(1)$, and the extreme rays of the cone are $f=f_\mathtt{2}$ and $f_x(r)=(r-1)^2/(r+x)$ for $x\in[0,\infty)$.
\end{proposition}
\begin{proof}
Positivity and monotonicity of $g_f$ are immediate from \eqref{eq:opconvrep}, since each $r\mapsto 1/(r+s)$ is positive and non-increasing. Also, the positivity is strict because $f$ is not affine, so $b>0$ or $\tau_f\neq0$. Continuity follows from the representation, by dominated convergence.
\end{proof}

\begin{remark}[The Belavkin--Staszewski entropy]\label{ex:opconv}
For $f_\mathtt{KL}(r)=r\log r-r+1$, one has $b=0$ and $\di\tau_f(s)=s\,\di s/(1+s)^2$, so that
\[
  g_f(r)=\int_0^\infty\frac{s\,\di s}{(r+s)(1+s)^2}=\frac{r\log r-r+1}{(r-1)^2} = \frac{f_\mathtt{KL}(r)}{f_\mathtt{2}(r)},
\]
decreasing from $1$ at $r=0$ to $0$ at infinity, with $g_f(1)=\tfrac12$. Since $b=0$ and $\tau_f\neq0$, the Belavkin--Staszewski entropy is not an extreme ray of the cone but a mixture over the family $f_x$ with the density $s/(1+s)^2$ vanishing at both ends and peaking at $s=1$.
\end{remark}

Every $\fdiv^{\max}$ averages $f$ over the spectrum of $\sigma^{-1/2}\rho\,\sigma^{-1/2}$, but no average records how far that spectrum reaches. To capture this, we define the following quantity.

\begin{definition}[Max-relative entropy]\label{def:dmax}
Let $\rho,\sigma\in\scrS_+(\frH)$. The \emph{max-relative entropy} \cite{datta2009} is defined by
\begin{equation}\label{eq:dmax}
  \dinf^{\max}(\rho\|\sigma) := \log\min\bigl\{\lambda>0 \;:\; \rho\preceq\lambda\sigma\bigr\}
  = \log\bigl\|\sigma^{-1/2}\rho\,\sigma^{-1/2}\bigr\|_\infty ,
\end{equation}
where $\|\cdot\|_\infty$ is the usual operator norm. The two expressions agree because $\rho\preceq\lambda\sigma$ holds if and only if $\sigma^{-1/2}\rho\,\sigma^{-1/2}\preceq\lambda\mathbbm{1}$. The same prescription defines the max-relative entropy of a pair of probability measures $\mu,\nu$ on a measurable space:
\begin{equation}\label{eq:dmaxclassical}
  \dinf(\mu\|\nu) := \log\min\bigl\{\lambda>0 \,:\, \mu\leq\lambda\nu\bigr\}
  = \begin{cases}
     \displaystyle \log \Bigl(\nu\text{-}\esssup \frac{\di\mu}{\di\nu}\Bigr) & \text{if $\mu\ll\nu$},\\
     +\infty & \text{otherwise}.
  \end{cases}
\end{equation}
This is the commutative case of \eqref{eq:dmax}. For commuting $\rho$ and $\sigma$ with eigenvalues $p_i$ and $q_i$ in a shared eigenbasis, both expressions reduce to $\log\max_i (p_i/q_i)$.
\end{definition}

\begin{proposition}[Comparison with $\chi^2$]\label{prop:sandwich}
Let $f$ and $g_f$ be as in Proposition~\ref{prop:opconv}, and write
\[
  \theta_{\max}:=e^{\dinf^{\max}(\rho\|\sigma)},\qquad
  \theta_{\min}:=e^{-\dinf^{\max}(\sigma\|\rho)}
\]
for the largest and smallest eigenvalues of $\sigma^{-1/2}\rho\,\sigma^{-1/2}$. Then for all
$\rho,\sigma\in\scrS_+(\frH)$,
\begin{equation}\label{eq:sandwich}
  g_f(\theta_{\max})\,D^{\max}_{\chi^2}(\rho\|\sigma) \leq \fdiv^{\max}(\rho\|\sigma)
   \leq g_f(\theta_{\min})\,D^{\max}_{\chi^2}(\rho\|\sigma).
\end{equation}
\end{proposition}
\begin{proof}
Write $\eta:=\sigma^{-1/2}\rho\,\sigma^{-1/2}$, so that $\fdiv^{\max}(\rho\|\sigma)=\Tr[\sigma f(\eta)]$ by Definition~\ref{def:fmax}. That $\theta_{\max}$ and $\theta_{\min}$ are the extreme eigenvalues of $\eta$ follows from Definition~\ref{def:dmax} together with the fact that $\rho^{-1/2}\sigma\rho^{-1/2}$ has the reciprocal eigenvalues.

By \eqref{eq:opconvrep} and the functional calculus of the single operator $\eta$, $f(\eta)=(\eta-\mathbbm{1})^2g_f(\eta)$, and all three factors commute. Since $g_f$ is non-increasing and the spectrum of $\eta$ lies in $[\theta_{\min},\theta_{\max}]$,
\[
  g_f(\theta_{\max})\mathbbm{1} \preceq g_f(\eta) \preceq g_f(\theta_{\min})\mathbbm{1},
\]
and hence $g_f(\theta_{\max})(\eta-\mathbbm{1})^2\preceq f(\eta)\preceq g_f(\theta_{\min})(\eta-\mathbbm{1})^2$. Taking $\Tr[\sigma\,\cdot\,]$, which preserves the operator order because $\sigma\succeq0$, and using $\Tr[\sigma(\eta-\mathbbm{1})^2]=D^{\max}_{\chi^2}(\rho\|\sigma)$, gives
\eqref{eq:sandwich}.
\end{proof}

In \cite{strengtheneddpi}, the BS-entropy is shown to enjoy the quantum counterpart of DPI. More generally, every maximal $f$-divergence does as well \cite{Matsumoto2018AF-Divergenceb}.

\begin{proposition}[Quantum DPI]
\label{prop:quantumdpi}
For any completely positive and trace-preserving (CPTP) map $\Phi:\scrS(\frH)\to\scrS(\frH')$ between finite-dimensional Hilbert spaces,
any operator-convex $f$ with $f(1)=0$, and all $\rho,\sigma\in\scrS(\frH)$,
\begin{equation}\label{eq:qdpi}
  \fdiv^{\max}(\Phi(\rho)\|\Phi(\sigma))\leq \fdiv^{\max}(\rho\|\sigma),
  \qquad
  \dinf^{\max}(\Phi(\rho)\|\Phi(\sigma))\leq \dinf^{\max}(\rho\|\sigma) .
\end{equation}
\end{proposition}

The first inequality is found in \cite{Matsumoto2018AF-Divergenceb}. The second holds for any positive trace-preserving $\Phi$. Indeed, if $\rho\preceq\lambda\sigma$ then $\lambda\Phi(\sigma)-\Phi(\rho)=\Phi(\lambda\sigma-\rho)\succeq0$, so $\Phi(\rho)\preceq\lambda\Phi(\sigma)$, and taking the least such $\lambda$ in Definition~\ref{def:dmax} gives the claim.

\section{The bridge: unravel $f$-divergences}
\label{sec:results_marcos}

This section lays the groundwork for the dynamical results in Section~\ref{sec:dissipation}. To transfer a decay estimate for measures on pure states to the states themselves, we need two ingredients: an $f$-divergence for ensembles that dominates a maximal $f$-divergence, and a minimizer that is itself an admissible initial condition for the unraveling. Theorem~\ref{thm:maintheorem} provides both. 

Throughout this section $\sigma\in\scrS_+(\frH)$ is faithful, while $\rho\in\scrS(\frH)$ is arbitrary, and we write
\begin{equation}\label{eq:spectral-geometric}
  \eta:=\sigma^{-1/2}\rho\,\sigma^{-1/2}=\sum_{i=1}^{n}\theta_i\ket{w_i}\bra{w_i},
  \qquad \{w_i\}_{i=1}^n\ \text{orthonormal},\ \theta_i\geq0,
\end{equation}
for its spectral decomposition, where the $\theta_i$ are strictly positive precisely when $\rho$ is faithful. Three operators carry these same eigenvalues, since they are similar to one another,
\begin{equation}\label{eq:threeforms}
  \rho\,\sigma^{-1}=\sigma^{1/2}\eta\,\sigma^{-1/2},
  \qquad \eta=\sigma^{-1/2}\rho\,\sigma^{-1/2},
  \qquad \sigma^{-1}\rho=\sigma^{-1/2}\eta\,\sigma^{1/2},
\end{equation}
and they are diagonalized respectively by the common basis $\ket{\psi_i}\propto\sigma^{1/2}\ket{w_i}$ of \eqref{eq:cbspectral}, by the orthonormal $\ket{w_i}$, and by the dual basis $\ket{\psi_i^\perp}\propto\sigma^{-1/2}\ket{w_i}$ of Proposition~\ref{prop:spectral}. Each of the three is used below, and which one appears is decided by the basis at hand: Section~\ref{sec:bssseeq} works with $\rho\sigma^{-1}$, whose eigenvectors are the common basis itself, and Appendix~\ref{app:commonbasisproofs} with $\sigma^{-1}\rho$, whose eigenvectors are what the construction there builds the basis from. This decomposition underlies the construction below and reappears in Section~\ref{sec:dissipation}.

Stochastic unravelings of master equations, resulting in stochastic Schr\"{o}dinger equations \eqref{eq:sse}, provide an interpretation of the state as a representation of an ensemble of pure states. Taking an infimum over all ensembles realizing the pair gives the largest quantity dominated by every realizing pair, which is exactly what lets an estimate proved for measures descend to the states. We define

\begin{definition}[Unravel $f$-divergence]\label{def:unr}
    Let $\rho,\sigma\in\scrS(\frH)$ be states, and let $f$ be as in Definition~\ref{def:fdiv}. Our proposed \emph{unravel $f$-divergence} ($\mathtt{UNR}$) is given by
    \begin{equation}
    \label{eq:sseentropydef}
    \fdiv^{\mathtt{UNR}}(\rho \|\sigma):=
    \inf\,\Bigl\{\fdiv(\mu \| \nu) \;:\; \rho = \Lambda(\mu),\; \sigma = \Lambda(\nu),\;\; \mu, \nu \in \calP(\scrP(\frH))\Bigr\},
    \end{equation}
    as the infimum over the $f$-divergences between measures $\mu$ and $\nu$ on $\scrP(\frH)$ that realize $\rho$ and $\sigma$, respectively. For $f=f_\mathtt{KL}$ we write $\BSunr:=\fdiv^{\mathtt{UNR}}$ and call it the \emph{unravel relative entropy}.
\end{definition}

Note that the infimum is taken over all Borel probability measures on $(\scrP(\frH),d_\mathtt{FS})$, and not only over finitely supported ones. This is intentional, as the laws $\mathrm{Law}(P_t)$ produced by \eqref{eq:sse} are diffuse. A formulation restricted to finite ensembles would therefore not apply to the objects an unraveling generates. Theorem~\ref{thm:maintheorem} states that the infimum over this larger class is still attained, and that it is attained on $n$ atoms.

Moreover, this definition is not arbitrary. The contraction principle produces an infimum over the fibers of $\Lambda$ when a large deviation rate function is pushed forward along $\Lambda$, and Appendix~\ref{sec:largedev} makes this precise for $f=f_\mathtt{KL}$, indicating that $\BSunr(\rho\|\sigma)$ is the least exponential cost, over ensembles $\nu$ realizing $\sigma$, of the empirical measure of samples from $\nu$ producing $\rho$, and the cost is attained on the common basis.

\medskip
The following theorem is a tighter version of the `minimal reverse test' characterization of \emph{maximal $f$-divergences} in the spirit of \cite{Matsumoto2018AF-Divergenceb,Matsumoto2005ReverseET}. To this end, we seek probability measures $\mu,\nu$ on pure states that realize a state $\rho$ and a faithful $\sigma$, respectively, and achieve the infimum in \eqref{eq:sseentropydef}. That is, we prove the following statement.
\begin{theorem}[$f$-divergence equivalence]
\label{thm:maintheorem}
Let $\rho\in\scrS(\frH)$ and let $\sigma\in\scrS_+(\frH)$ be faithful. Then there exist probability measures $\mu_{\CB},\nu_{\CB}$ with atoms on a common basis $\{\psi_j\}_{j=1}^n$, that realize $\rho,\sigma$, respectively, such that for every operator-convex $f$ with $f(1)=0$,
\begin{align}
\fdiv^{\max}(\rho\|\sigma)\overset{(\bigstar)}{=}\fdiv(\mu_{\CB}\| \nu_{\CB}) \overset{(\blacktriangle)}{=} \fdiv^{\mathtt{UNR}}(\rho\|\sigma).
\end{align}
In particular, for $f=f_\mathtt{KL}$,
$\BS(\rho\|\sigma)=\KL(\mu_{\CB}\|\nu_{\CB})=\BSunr(\rho\|\sigma)$.
\end{theorem}

We stress that the common basis $\{\psi_j\}_{j=1}^n$ is chosen independently of $f$, so a \emph{single} pair $(\mu_\CB,\nu_\CB)$ attains the infimum simultaneously for every operator-convex $f$. The minimizing reverse test is universal in this sense.

\begin{remark}\label{rem:whereopconv}
The two halves of Theorem~\ref{thm:maintheorem} make different demands on $f$. The proof of $(\bigstar)$ in Section~\ref{sec:bssseeq} evaluates $f$ by functional calculus on the spectrum of $\rho\sigma^{-1}$, so it requires only that $f$ be defined at the ratios $\theta_i=\rho_i/\sigma_i$, i.e., operator convexity is not needed. It is needed only for $(\blacktriangle)$, where the quantum DPI of Proposition~\ref{prop:quantumdpi} enters. Thus, the identity
\begin{equation}\label{eq:starallf}
  \Tr\bigl[\sigma f(\sigma^{-1/2}\rho\,\sigma^{-1/2})\bigr]
  = \sum_{i=1}^n \sigma_i\,f(\theta_i)
  = \fdiv(\mu_\CB\|\nu_\CB),
\end{equation}
holds for every $f$, and the common basis computes the whole family at once.
\end{remark}

Our unraveling-inspired interpretation of mixed states as statistical ensembles of pure states allows a constructive proof of this theorem, which provides a direct connection between the maximal $f$-divergence of two mixed states and the classical $f$-divergence of two statistical ensembles of pure states that realize them. Before turning to the proof, we highlight what in Theorem~\ref{thm:maintheorem} is new against the existing reverse-test literature.

\begin{remark}[Relation to previous work]\label{rem:delimitation} Proposition~\ref{prop:spectral} below shows that $(\mu_\CB,\nu_\CB)$ coincides with the distributions given in \cite{explicitmax}, and that the common basis is the support of the reverse test implicitly given there. The existence statement of Theorem~\ref{thm:commonbasis} is therefore not new. What the construction of Section~\ref{sec:basisconstruction} adds is the formulation in terms of the basis rather than the spectrum, which makes Lemmas~\ref{lem:cbstructure}$(i)$ and \ref{lem:cbstructure}$(ii)$ and the dual basis of Lemma~\ref{lem:basis} available. Moreover, Theorem~\ref{thm:maintheorem}$(\blacktriangle)$ shows that the infimum is attained over all of $\calP(\scrP(\frH))$, where the minimizer consists of pure states, and is thus an admissible initial condition for \eqref{eq:sse}. Section~\ref{sec:dissipation} uses this last point, which has no counterpart in a spectral formulation.
\end{remark}

\begin{remark}[Measures on a continuum]\label{rem:continuum} The infima defining $D_f^{\max}$ in \cite{explicitmax,Matsumoto2018AF-Divergenceb} are taken over finitely supported $p,q$, whereas Definition~\ref{def:unr} ranges over $\calP(\scrP(\frH))$. This is a strictly larger class, and it is where the laws generated by \eqref{eq:sse} live. Theorem~\ref{thm:maintheorem}$(\blacktriangle)$ closes this gap using the discretization of Lemma~\ref{lem:discretization} and the classical DPI.
\end{remark}

In Section~\ref{sec:basisconstruction}, we detail a method of constructing a common basis. Using a basis from $\CB(\rho,\sigma)$ (see Definition~\ref{def:commonbasis}), it is possible to determine $\mu_{\CB},\nu_{\CB}$ from any pair $\rho,\sigma$. This allows for some exploratory simulations shown in Figure~\ref{fig:densities}, where we randomly sample states in $\scrS_+(\frH)$ from the Hilbert--Schmidt ensemble, i.e., as partial traces of Haar-random pure states on $\frH\otimes\frH$ \cite{haarmeasure} and calculate the quantities $D^{\max}_{\chi^2}(\rho\|\sigma)$, $D_{\chi^2}(\mu_{\CB}\|\nu_{\CB})$, $\BS(\rho\|\sigma), \KL(\mu_{\CB}\| \nu_{\CB})$, and $\Um(\rho\|\sigma)$. These results indicate equivalence between $D_{\chi^2}$ and $\KL$ on the common basis and the respective maximal divergences $D^{\max}_{\chi^2}$ and $\BS$, but a clear difference with the distribution of $\Um$. Moreover, every random realization of $D_{\chi^2}$ and $\KL$ on the common basis and $D^{\max}_{\chi^2}$ and $\BS$ respectively was found to correspond up to numerical accuracy, providing a numerical confirmation of the equality given by $(\bigstar)$.

\begin{figure}[hb]
  \centering
  \begin{subfigure}[t]{0.49\linewidth}
    \centering
    \includegraphics[width=\linewidth]{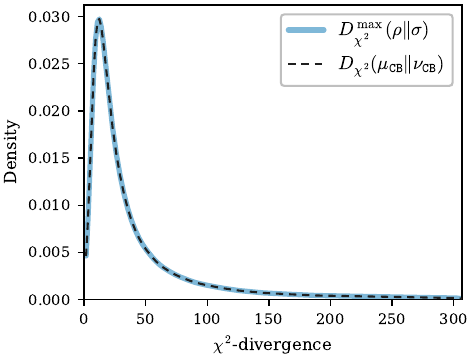}
    \subcaption{}\label{fig:densities-chi2}
  \end{subfigure}
  \hfill
  \begin{subfigure}[t]{0.49\linewidth}
    \centering
    \includegraphics[width=\linewidth]{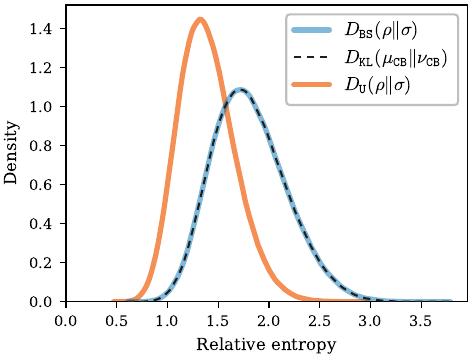}
    \subcaption{}\label{fig:densities-relent}
  \end{subfigure}
  \caption{Approximate probability density functions constructed from $10^6$ random pairs of states $(\rho,\sigma)$ sampled from the Hilbert--Schmidt ensemble. (\subref{fig:densities-chi2}) The maximal $\chi^2$-divergence $D^{\max}_{\chi^2}(\rho\|\sigma)$ against the classical $\chi^2$-divergence $D_{\chi^2}(\mu_{\mathtt{CB}}\|\nu_{\mathtt{CB}})$ on the common basis. (\subref{fig:densities-relent}) The BS-entropy $D_{\mathtt{BS}}(\rho\|\sigma)$, the classical KL-divergence $D_{\mathtt{KL}}(\mu_{\mathtt{CB}}\|\nu_{\mathtt{CB}})$ on the common basis, and the Umegaki entropy $D_{\mathtt{U}}(\rho\|\sigma)$. In both figures, the first two are indistinguishable (i.e., $(\bigstar)$); the third is a different quantum extension of $f_\mathtt{KL}$ and clearly different. 
  The $\chi^2$-divergence is numerically unstable on pairs close to violating absolute continuity, so the $5\%$ of pairs with the largest $D_{\chi^2}$ are discarded, and both panels are drawn from this one subsample.
  }
  \label{fig:densities}
\end{figure}

\subsubsection*{Proof outline of Theorem~\ref{thm:maintheorem}}
First, Section~\ref{sec:basisconstruction} shows that any state and any faithful state admit a common basis $\{\psi_i\}_{i=1}^n$, and expressing $(\rho,\sigma)$ in it gives the measures $(\mu_{\CB},\nu_{\CB})$. Section~\ref{sec:bssseeq} then uses dual bases to prove that $\fdiv(\mu_{\CB}\|\nu_{\CB})$ equals the maximal $f$-divergence $\fdiv^{\max}(\rho\|\sigma)$ for every operator-convex $f$, which is $(\bigstar)$. Finally, Section~\ref{sec:klmin} uses a discretization argument and the classical DPI to show that $\fdiv(\mu_{\CB}\|\nu_{\CB})$ attains the infimum defining $\fdiv^{\mathtt{UNR}}(\rho\|\sigma)$, which is $(\blacktriangle)$. Taking $f=f_\mathtt{KL}$ then yields the result for $\BS$.

\subsection{Common basis construction}
\label{sec:basisconstruction}

One of the main results needed to prove Theorem~\ref{thm:maintheorem} is the construction of a common basis (Definition~\ref{def:commonbasis}) for the pair of states $(\rho,\sigma)$ in an $n$-dimensional Hilbert space $\frH$. The intuition behind this comes from the two-dimensional case (qubit), where the states lie in the interior of the Bloch ball. In this case, the measures $\mu$ and $\nu$ would have the lowest KL divergence (i.e., be least discernible) if they were supported on the same, and on as few, pure states as possible. This happens when the two states are expressed as the weighted sum of the two pure states at the intersection of the Bloch sphere and the line connecting $\rho$ and $\sigma$. These two states then form a common basis (see Figure~\ref{fig:commonbasis}). Theorem~\ref{thm:commonbasis} extends the existence of a common basis to higher dimensions.

\begin{figure}[h]
    \centering
    \includegraphics[width=0.4\linewidth]{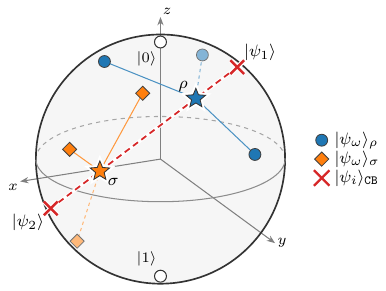}
    \caption{Common basis construction on a 1-qubit state. Ensembles of pure states $\lvert\psi_\omega\rangle_\rho$ (circles) and $\lvert\psi_\omega\rangle_\sigma$ (diamonds) realize the states $\rho$ and $\sigma$ (stars). A common basis $\lvert\psi_i\rangle_{\mathtt{CB}}$ (red crosses) is constructed by finding the intersection of the Bloch sphere and the line connecting $\rho$ and $\sigma$ (red dashed).}
    \label{fig:commonbasis}
\end{figure}

\begin{theorem}[Existence of common basis]
\label{thm:commonbasis}
    There exists a common basis for any $\rho \in \scrS(\frH)$ and any faithful $\sigma \in \scrS_+(\frH)$, i.e., there is a (possibly non-orthogonal) basis $\{\psi_i\}_{i=1}^n\in\CB(\rho,\sigma)$ and coefficients $\rho_i, \sigma_i \in[0,1]$ for $i \in \{1,...,n\}$ such that
    \begin{equation}
        \rho = \sum_{i=1}^n \rho_i |\psi_i\rangle\langle\psi_i|, \quad \sum_{i=1}^n \rho_i = 1,\qquad
        \sigma = \sum_{i=1}^n \sigma_i |\psi_i\rangle\langle\psi_i|, \quad \sum_{i=1}^n \sigma_i = 1.
    \end{equation}
\end{theorem}

\begin{proof}
With $\eta$ and $\{\theta_i,w_i\}$ as in \eqref{eq:spectral-geometric}, put
\begin{equation}\label{eq:cbspectral}
  \ket{\psi_i}:=\frac{\sigma^{1/2}\ket{w_i}}{\|\sigma^{1/2}\ket{w_i}\|},\qquad
  \sigma_i:=\braket{w_i|\sigma|w_i},\qquad \rho_i:=\theta_i\sigma_i .
\end{equation}
Then $\{\psi_i\}$ is a basis, since $\sigma^{1/2}$ is invertible, and
\[
  \sum\nolimits_i\sigma_i\ket{\psi_i}\bra{\psi_i}
  =\sigma^{1/2}\Bigl(\sum\nolimits_i\ket{w_i}\bra{w_i}\Bigr)\sigma^{1/2}=\sigma,
  \qquad
  \sum\nolimits_i\rho_i\ket{\psi_i}\bra{\psi_i}=\sigma^{1/2}\eta\,\sigma^{1/2}=\rho ,
\]
the coefficients summing to $\tr\,\sigma=1$ and $\tr\,\rho=1$ respectively.
\end{proof}

In fact, every common basis arises this way. If $\rho=\sum_i\rho_i\ket{\psi_i}\bra{\psi_i}$ and $\sigma=\sum_i\sigma_i\ket{\psi_i}\bra{\psi_i}$, then the family of vectors $v_i:=\sigma^{-1/2}\ket{\psi_i}$ satisfies $\sum_i\sigma_i\ket{v_i}\bra{v_i}=\mathbbm{1}$, so the $n$ vectors $\sqrt{\sigma_i}\ket{v_i}$ are orthonormal, and $\sum_i(\rho_i/\sigma_i)\,\sigma_i\ket{v_i}\bra{v_i}$ is a spectral decomposition of $\eta=\sigma^{-1/2}\rho\,\sigma^{-1/2}$. Common bases of $(\rho,\sigma)$ therefore correspond exactly to orthonormal eigenbases of $\eta$, with $\theta_i=\rho_i/\sigma_i$ the eigenvalues. A more explicit construction is carried out in Appendix~\ref{app:commonbasisproofs}, where we also check that it yields the same basis.

\medskip
The following lemma shows the structure of this basis.

\begin{lemma}[Structure of the common basis]\label{lem:cbstructure}
    Let $\rho\in\scrS(\frH)$ and let $\sigma\in\scrS_+(\frH)$ be faithful.
    \begin{enumerate}[label=(\roman*)]
        \item Any probability measure realizing $\rho$ has support on at least $\rank\rho$ linearly independent pure states, and $\mu_\CB$ attains this bound, so a common basis is a minimal realization. In particular, the bound is $n$ when $\rho$ is faithful.
        \item If the $\theta_i$ of \eqref{eq:spectral-geometric} are pairwise distinct, the common basis is unique up to relabelling and phases, i.e., for any two common bases $\{\psi_i\}_{i=1}^n$ and $\{\phi_i\}_{i=1}^n$ of $(\rho,\sigma)$, there are scalars $\alpha_i\in\mathbb{C}$ with $|\alpha_i|=1$ and a permutation $\pi$ of $\{1,\dots,n\}$ such that $\ket{\psi_i}=\alpha_i\ket{\phi_{\pi(i)}}$. Since $0$ may then occur among the $\theta_i$ at most once, this covers a singular $\rho$ only when $\rank\rho=n-1$.
    \end{enumerate}
\end{lemma}

Both parts are proved in Appendix~\ref{app:commonbasisproofs}.
    Uniqueness fails when $\eta=\sigma^{-1/2}\rho\,\sigma^{-1/2}$ has $m<n$ distinct eigenvalues, since permuting the
    vectors before applying the scheme of Lemma~\ref{lem:gramschmidt} returns a different set. In any case, a common basis of $n$ linearly independent pure states has been constructed and shown to be as small as possible and, under specific circumstances, unique up to permutation.

\subsection{Relation to the spectral construction}
\label{sec:relation-spectral}

The common basis of \eqref{eq:cbspectral} is the reverse test implicit in \cite{explicitmax}. We record the correspondence explicitly, both to place the construction in the literature and because the formulas below are used in Section~\ref{sec:dissipation}.

\begin{proposition}[Dual basis]\label{prop:spectral}
  Let $\rho\in\scrS(\frH)$, let $\sigma\in\scrS_+(\frH)$ be faithful, and let $\{\theta_i,w_i\}$ be as in \eqref{eq:spectral-geometric}.
  The common basis $\{\psi_i\}$ of \eqref{eq:cbspectral} has dual basis
  \[
    \ket{\psi_i^\perp}=\bigl\|\sigma^{1/2}\ket{w_i}\bigr\|\,\sigma^{-1/2}\ket{w_i}.
  \]
  This dual basis agrees with the basis produced by the construction of Appendix~\ref{app:commonbasisproofs} up to phases and, in the degenerate case, up to the choice
  of orthonormal frame within each eigenspace of $\eta=\sigma^{-1/2}\rho\,\sigma^{-1/2}$.
\end{proposition}
\begin{proof}
  Duality is immediate from $\braket{w_i|\sigma^{-1/2}\sigma^{1/2}|w_j}=\delta_{ij}$, the two
  normalizing factors cancelling. The agreement
  with the construction of Appendix~\ref{app:commonbasisproofs} is checked there.
\end{proof}

In particular, $\rho_i/\sigma_i=\theta_i$, so that Theorem~\ref{thm:maintheorem} recovers the closed form \[\fdiv^{\max}(\rho\|\sigma)=\sum\nolimits_i \sigma_i\,f(\theta_i)=\tr[\sigma f(\sigma^{-1/2}\rho\,\sigma^{-1/2})].\] For $f=f_\mathtt{KL}$ this is $\BS(\rho\|\sigma)=\sum_i\sigma_i\,\theta_i\log\theta_i$.

\begin{remark}\label{rem:dmaxbasis}
The same identification computes the max-relative entropy of Definition~\ref{def:dmax}. Since $\theta_i$'s are the eigenvalues of $\eta=\sigma^{-1/2}\rho\,\sigma^{-1/2}$, we find by definition that $\dinf^{\max}(\rho\|\sigma)=\log\max_i\theta_i=\dinf(\mu_\CB\|\nu_\CB)$, where the last equality holds because $\di\mu_\CB/\di\nu_\CB$ takes the value $\theta_i$ on the atom $P_i$. The minimizing pair therefore computes $\dinf^{\max}$ as well, exactly as it computes $\fdiv^{\max}$ in $(\bigstar)$. Likewise $\dinf^{\max}(\sigma\|\rho)=-\log\min_i\theta_i$, since the eigenvalues of $\eta^{-1}$ are the $\theta_i^{-1}$'s. In particular, $\theta_{\max}:=\max_i\theta_i=e^{\dinf^{\max}(\rho\|\sigma)}$ and $\theta_{\min}:=\min_i\theta_i=e^{-\dinf^{\max}(\sigma\|\rho)}$.

It also follows that the two max-relative entropies control every maximal $f$-divergence simultaneously. Indeed, since $\sum_i \sigma_i=1$, for any convex $f$,
\begin{equation}\label{eq:fmaxbydmax}
  \fdiv^{\max}(\rho\|\sigma)=\sum\nolimits_i \sigma_i\,f(\theta_i) \leq \max_{\theta\in[\theta_{\min},\theta_{\max}]}f(\theta)
  = \max\bigl\{f(\theta_{\min}),\,f(\theta_{\max})\bigr\},
\end{equation}
where the last equality holds because a convex function on an interval attains its maximum at an endpoint. 

For $f$ of the form $f(x)=x\,g(x)$ with $g$ monotone or convex, the bound sharpens, because the tilted weights $\sigma_i\theta_i$ also sum to $1$, giving
\[
  \fdiv^{\max}(\rho\|\sigma)=\sum\nolimits_i (\sigma_i\theta_i)\,g(\theta_i)
   \leq \max\bigl\{g(\theta_{\min}),\,g(\theta_{\max})\bigr\}.
\]
Hence, taking $f(x)=x\log(x)$, with $g$ increasing, gives
$\BS(\rho\|\sigma)\leq\log\theta_{\max}=\dinf^{\max}(\rho\|\sigma)$ and thus
$\Um\leq\BS\leq \dinf^{\max}$. 

For $f_\mathtt{2}(x)=(x-1)^2$, the same constraints give an identity rather than a bound. Indeed, since $\sum_i \sigma_i\theta_i=1$, we have that
\[
  D^{\max}_{\chi^2}(\rho\|\sigma)=\sum\nolimits_i \sigma_i(\theta_i-1)^2
  =\sum\nolimits_i \sigma_i\theta_i^2-\Bigl(\sum\nolimits_i \sigma_i\theta_i\Bigr)^{\!2}
  =\operatorname{Var}_\sigma(\theta),
\]
which is the variance of the ratios under the weights $(\sigma_i)_i$. A sharp bound on the variance in terms of the range and the mean is given in \cite{bhatiadavis}, which states that if $\theta\in[\theta_{\min},\theta_{\max}]$ has mean $\mu$, then $\mathrm{Var}(\theta)\leq(\theta_{\max}-\mu)(\mu-\theta_{\min})$. With $\mu=1$, we obtain
\begin{equation}\label{eq:chibydmax}
  D^{\max}_{\chi^2}(\rho\|\sigma) \leq (\theta_{\max}-1)(1-\theta_{\min})
  =\bigl(e^{\dinf^{\max}(\rho\|\sigma)}-1\bigr)\bigl(1-e^{-\dinf^{\max}(\sigma\|\rho)}\bigr),
\end{equation}
with equality when the $\theta_i$'s take only the two extreme values.
\end{remark}

\subsection{$\fdiv^{\max}$/$\fdiv^{\mathtt{UNR}}$ equivalence}
\label{sec:bssseeq}
Before proving $(\bigstar)$ of Theorem~\ref{thm:maintheorem}, we record an auxiliary lemma, proved in Appendix~\ref{app:inverse}. From here on, we use the convention $\braket{\psi^\perp_i | \psi_j} = \delta_{ij}$. Hence, the dual basis states are not necessarily normalized, i.e., $\braket{\psi^\perp_i | \psi_i ^\perp} \neq 1$.

\begin{lemma}[Inverse state]
\label{lem:inverse}
    Let $\sigma\in \scrS_+(\frH)$ be a faithful state. Suppose $\sigma$ can be written as a convex combination of $\{|\psi_i\rangle\langle\psi_i|\}_{i=1}^n$ with coefficients $\sigma_i\in(0,1]$, where $\{\psi_i\}_{i=1}^n$ is a basis of $\frH$. Let $\{\psi_i^\perp\}_{i=1}^n$ be the corresponding dual basis of $\{\psi_i\}_{i=1}^n$. Then,
    \begin{align}
        \sigma^{-1} = \sum_{i=1}^n \frac{1}{\sigma_i}|\psi_i ^\perp\rangle\langle\psi_i^\perp|.
    \end{align}
\end{lemma}

Now $(\bigstar)$ of Theorem \ref{thm:maintheorem} can be proven.

\begin{proof}[Proof of $(\bigstar)$ in Theorem~\ref{thm:maintheorem}]
     Let $\rho \in \scrS(\frH)$ and let $\sigma \in \scrS_+(\frH)$ be faithful. Further, let $\{\psi_j\}_{j=1}^n\in \CB(\rho, \sigma)$ be a common basis and let $\{\psi_j^\perp\}_{j=1}^n$ be its corresponding dual basis, such that $\braket{\psi_i | \psi_j ^\perp}=\delta_{ij}$. Then there are coefficients $\rho_i, \sigma_i$ such that
   \begin{align}
       \rho = \sum_{i=1}^n \rho_i \ket{\psi_i}\bra{\psi_i}, \quad \sigma = \sum_{i=1}^n \sigma_i \ket{\psi_i}\bra{\psi_i}.
   \end{align}
    Let $\mu_{\CB}$, $\nu_{\CB}$ be the measures having support on the common basis that realize $\rho$, $\sigma$, then
    \begin{equation}
    \fdiv(\mu_{\CB}\|\nu_{\CB})=\sum_{i=1}^n\sigma_i\, f\!\left(\frac{\rho_i}{\sigma_i}\right).
    \end{equation}
    By \eqref{eq:threeforms} the operator $\rho\sigma^{-1}$ is similar to the Hermitian $\eta$, hence diagonalizable with the same real non-negative eigenvalues. By functional calculus $f(\rho\sigma^{-1}) = \sigma^{1/2}f(\sigma^{-1/2}\rho\,\sigma^{-1/2})\sigma^{-1/2}$, so that by the cyclic property of the trace
    \begin{equation}\label{eq:fmax-cyclic}
        \fdiv^{\max}(\rho\|\sigma) = \Tr\bigl[\sigma f(\sigma^{-1/2}\rho\,\sigma^{-1/2})\bigr]
        = \Tr\bigl[\sigma^{1/2} f(\sigma^{-1/2}\rho\,\sigma^{-1/2})\sigma^{1/2}\bigr]
        = \Tr\bigl[f(\rho\sigma^{-1})\,\sigma\bigr].
    \end{equation}
    Using Lemma \ref{lem:inverse}, one deduces
    \begin{align}
        \rho \sigma^{-1} = \sum_{i=1}^n \sum_{j=1}^n \frac{\rho_i}{\sigma_j} \ket{\psi_i} \braket{\psi_i|\psi_j ^\perp} \bra{\psi_j^\perp} = \sum_{i=1}^n \frac{\rho_i}{\sigma_i} \ket{\psi_i}\bra{\psi_i^\perp} .
    \end{align}
    From the representation above, it follows that the eigenvectors of $\rho \sigma^{-1}$ are $\ket{\psi_i}$ with eigenvalues $\rho_i/\sigma_i$ and dual eigenvectors $\ket{\psi_i^\perp}$. Hence, one can directly compute (if $\rho$ is singular and $f(0)=+\infty$, both sides below equal $+\infty$ under the convention of Definition~\ref{def:fmax})
    \begin{align}
        f(\rho \sigma^{-1}) =
         \sum_{i=1}^n f\!\left(\frac{\rho_i}{\sigma_i}\right) \ket{\psi_i}\bra{\psi_i^\perp}.
    \end{align}
    Finally, substituting into \eqref{eq:fmax-cyclic} and using $\braket{\psi_i^\perp|\psi_j}=\delta_{ij}$,
    \begin{align*}
        \fdiv^{\max}(\rho\|\sigma)=&\Tr\left[ \left(\sum_{i=1}^n f\!\left(\frac{\rho_i}{\sigma_i}\right)\ket{\psi_i}\bra{\psi_i^\perp}  \right) \left(\sum_{j=1}^n \sigma_j \ket{\psi_j}\bra{\psi_j} \right) \right] \\
        =& \Tr\left[\sum_{i=1}^n \sigma_i\, f\!\left(\frac{\rho_i}{\sigma_i} \right) \ket{\psi_i} \bra{\psi_i}\right] = \sum_{i=1}^n \sigma_i\, f\!\left(\frac{\rho_i}{\sigma_i} \right) = \fdiv(\mu_{\CB}\| \nu_{\CB}),
    \end{align*}
which proves $(\bigstar)$ of Theorem~\ref{thm:maintheorem}. Taking $f=f_\mathtt{KL}$ gives $\BS(\rho\|\sigma)=\KL(\mu_\CB\|\nu_\CB)$.
\end{proof}

\subsection{Minimization of the $f$-divergence}
\label{sec:klmin}
In this section, we prove ($\blacktriangle$) of Theorem~\ref{thm:maintheorem}, which is the statement
    \begin{equation}
            \fdiv(\mu_{\CB}\| \nu_{\CB}) = \inf_{\mu, \nu \in \calP(\scrP(\frH))} \Bigl\{\fdiv(\mu \| \nu) \;:\; \rho = \Lambda(\mu),\; \sigma = \Lambda(\nu) \Bigr\}.
    \end{equation}

For this, we require the following, proven in Appendix~\ref{app:commonbasisproofs}.
\begin{lemma}[Discretization]
\label{lem:discretization}
    Let $\mu$, $\nu\in \mathcal{P}(\scrP(\frH))$ be measures realizing $\rho\in\scrS(\frH)$ and a faithful $\sigma\in\scrS_+(\frH)$, respectively, and let $f$ be operator convex with $f(1)=0$. Then there exists a sequence of discrete measures $\mu_k$, $\nu_k$ realizing $\rho_k$, $\sigma_k$, respectively, independently of $f$, such that
    \begin{enumerate}[label=(\roman*)]
        \item $\rho_k\rightarrow\rho$ and $\sigma_k\rightarrow\sigma$ in $\scrS(\frH)$, i.e., $d_\mathtt{TR}(\rho_k,\rho)\to 0$ and $d_\mathtt{TR}(\sigma_k,\sigma)\to 0$;
        \item $\fdiv(\mu\|\nu)\geq \limsup_{k\rightarrow\infty} \fdiv(\mu_k\|\nu_k)$;
        \item $\fdiv(\mu_k\|\nu_k)\geq \fdiv^{\max}(\rho_k\|\sigma_k)$.
    \end{enumerate}
\end{lemma}

Using these lemmas, it is possible to prove ($\blacktriangle$) of Theorem~\ref{thm:maintheorem}.
\begin{proof}[Proof of $(\blacktriangle)$ in Theorem~\ref{thm:maintheorem}]
Let $\mu$ and $\nu$ be two measures realizing $\rho$ and $\sigma$ respectively. Taking sequences of discrete measures $\mu_k,\nu_k$ and states $\rho_k,\sigma_k$ given by Lemma~\ref{lem:discretization},
\begin{equation}
    \fdiv(\mu \|\nu)\geq \limsup_{k\to \infty}\fdiv(\mu_k \|\nu_k)  \geq \liminf_{k\to\infty}\fdiv^{\max}(\rho_k \| \sigma _k) \geq \fdiv^{\max}(\rho \| \sigma) = \fdiv(\mu_{\CB}\| \nu_{\CB}),
\end{equation}
where the first inequality follows from $(ii)$, the second from $(iii)$, the third from $(i)$ combined with the lower semi-continuity of $\fdiv^{\max}$ with respect to the trace distance \cite{Matsumoto2018AF-Divergenceb}, and the last equality from $(\bigstar)$. Since $\mu_\CB$ and $\nu_\CB$ realize $\rho$ and $\sigma$, they are themselves admissible in the infimum, which is therefore attained at them. 

Taking $f=f_\mathtt{KL}$ gives $\KL(\mu_\CB\|\nu_\CB)=\BSunr(\rho\|\sigma)$.
\end{proof}

\section{Improving data processing along the Lindblad flow}
\label{sec:dissipation}

This section contains the main results of this work. The quantum DPI of Proposition~\ref{prop:quantumdpi} asserts that $\fdiv^{\max}(\rho_t\|\sigma_t)\leq\fdiv^{\max}(\rho_0\|\sigma_0)$ along the Lindblad flow, and the following is an attempt to quantify this in accordance with the approach outlined in Figure~\ref{fig:lift}. Theorem~\ref{thm:contraction} strengthens the data processing inequality with an explicit dissipation term, while the remaining sections turn that term into a rate. Section~\ref{sec:lsi} establishes a $\Phi$-entropy inequality for the unraveling, which, for every operator-convex $f$, follows from the spectral gap. Section~\ref{sec:decay} combines it with a dissipation identity to obtain exponential decay. Section~\ref{sec:sharprate} provides a second estimate that does not require a functional inequality. Section~\ref{sec:examples} examines two examples to assess the sharpness of the decay estimates.

\subsection{A strengthened data processing inequality}
\label{sec:contraction}
The unraveling of the Lindblad equation \eqref{eq:lindblad} with the stochastic Schr\"odinger equation \eqref{eq:sse} allows one to obtain a refined contraction estimate. Recall that the $W_j$ driving \eqref{eq:sse} are independent Brownian motions with quadratic variation $[W_j]_t=t$. The corresponding Fokker--Planck equation \cite{risken1996fokker} for the law $\mu_t = \text{Law}(P_t)$ of the process then reads
\begin{equation}\label{eq:fokker-planck}\tag{\textsf{FP}}
    \partial_t \mu_t = \scrL^\dagger\mu_t,
\end{equation}
where $\scrL^\dagger$ is the adjoint of the infinitesimal generator
\[
    \scrL F(P) = \Der F(P)\bigl[\mathcal{L}(P)\bigr] + \frac{1}{2}\sum\nolimits_j \gamma_j^2\, \Der^2F(P)\bigl[\iu[S_j,P],\iu[S_j,P]\bigr].
\]
Here, $\Der F$ and $\Der^2F$ are the first- and second-order Fr\'echet derivatives of the function $F$ on $\scrP(\frH)$, respectively, and $\mathcal{L}$ is the Lindblad generator of \eqref{eq:lindblad}. The form of the drift is It\^{o}'s formula applied to the rank-one projection $P$. Indeed,  expanding $\di P$ from \eqref{eq:sse} gives 
\[
	\di P = \mathcal{L}(P)\,\di t + \iu\sum\nolimits_j\gamma_j [S_j,P]\,\di W_j,
\]
so that the \emph{It\^{o}} drift is the Lindbladian itself and the noise acts along the directions $\iu[S_j,P]$. 

\medskip
Associated with $\scrL$ is the \emph{carr\'{e} du champ} operator
\begin{equation}\label{eq:carre}
\begin{aligned}
  \Gamma(u,v)(P) &:= \tfrac{1}{2}\bigl(\scrL(uv)-u\,\scrL v-v\,\scrL u\bigr)(P) \\
  &= \tfrac{1}{2}\sum\nolimits_j \gamma_j^2\,\Der u(P)\bigl[\iu[S_j,P]\bigr]\,\Der v(P)\bigl[\iu[S_j,P]\bigr],
\end{aligned}
\end{equation}
in which the Hamiltonian $H$ does not appear, as it enters $\scrL$ only through the first-order term and cancels in the symmetrized difference.

\medskip
We further introduce, for $f:(0,+\infty)\to \R$, twice continuously differentiable with $f''>0$, the \emph{relative
$f$-Fisher information} of $\mu$ with respect to $\nu$,
\begin{equation}\label{eq:ffisher}
  I_f(\mu\|\nu) := \begin{cases}
    \displaystyle \frac{1}{2}\int_{\scrP(\frH)} f''(u)\,\Gamma(u,u)\,\nu(\di P) & \text{if $\mu\ll\nu$ with}\; \displaystyle u=\frac{\di\mu}{\di\nu},\\
    +\infty & \text{otherwise,}
  \end{cases}
\end{equation}
with the carr\'{e} du champ operator $\Gamma$ of \eqref{eq:carre}. Note that $I_f\geq0$ by convexity
of $f$. 

For $f(x)=(x-1)^2$ one has $f''\equiv2$ and $I_f=\int\Gamma(u,u)\,\di\nu$, and for $f=f_\mathtt{KL}$ one has $f''(x)=1/x$ and $I_f$ is the \emph{relative KL-Fisher information}
\begin{equation}\label{eq:klfisher}
  I_\mathtt{KL}(\mu\|\nu)=2\int_{\scrP(\frH)}\Gamma\bigl(\sqrt{u},\sqrt{u}\bigr)\,\nu(\di P)
  =\sum\nolimits_j\gamma_j^2\int_{\scrP(\frH)}
    \Bigl(\Der\sqrt{u}\bigl[\iu[S_j,P]\bigr]\Bigr)^2\,\nu(\di P).
\end{equation}
Finally, for a probability measure $\nu$ on $\scrP(\frH)$ and $v\in L^2(\nu)$ we write
\[
	\mathrm{Var}_\nu(v) := \int v^2\,\di\nu - \biggl(\int v\,\di\nu\biggr)^2,
\]
for the variance of $v$ under $\nu$. For $f(x)=(x-1)^2$ one has $\fdiv(\mu\|\nu)=\mathrm{Var}_\nu(u)$ with $u=\frac{\di\mu}{\di\nu}$.

\medskip
Our main result for this subsection is the following.

\begin{theorem}[Strengthened quantum DPI]\label{thm:contraction}
Let $\rho_0,\sigma_0\in\scrS_+(\frH)$, and let $\mu_t,\nu_t$ solve the Fokker--Planck equation \eqref{eq:fokker-planck} with initial data $\mu_{0,\CB}$ and $\nu_{0,\CB}$, respectively. Assume their densities $u_t,v_t$ with respect to $\nu_\mathtt{FS}$ are sufficiently smooth for $t>0$. Then, for every normalized operator convex $f$ with $f''>0$, the following entropy-dissipation inequality holds:
\begin{equation}
  \fdiv^{\max}(\rho_t\|\sigma_t)+2\int_0^t I_f(\mu_s\|\nu_s)\,\di s \leq\fdiv^{\max}(\rho_0\|\sigma_0)
  \qquad\text{for all }t\geq0 .
\end{equation}
\end{theorem}
\begin{proof}
    Any choice of common basis in $\CB(\rho_0,\sigma_0)$ gives the initial measures $\mu_{0,\CB}$ and $\nu_{0,\CB}$ such that $\fdiv(\mu_{0,\CB}\|\nu_{0,\CB})=\fdiv^{\mathtt{UNR}}(\rho_0\|\sigma_0)$. The evolutions $\mu_t$ and $\nu_t$ starting from $\mu_{0,\CB}$ and $\nu_{0,\CB}$ along the Fokker--Planck equation \eqref{eq:fokker-planck} yield the entropy-dissipation inequality
    \begin{equation}\label{eq:dissipation}
        \fdiv(\mu_t\|\nu_t)+2\int_s^tI_f(\mu_r\|\nu_r)\,\di r \leq \fdiv(\mu_s\|\nu_s),\qquad 0<s\leq t,
    \end{equation}
    with $I_f$ the relative $f$-Fisher information of \eqref{eq:ffisher}. Its derivation is postponed to Appendix~\ref{app:lindbladflow}. 
    
    Now let $\Phi_t$ denote the solution operator to the Lindblad equation, i.e., $\rho_t=\Phi_t(\rho_0)$ satisfies \eqref{eq:lindblad}. Note that since
    \begin{equation}
    \begin{aligned}
    \Lambda(\mu_t)&=\int_{\scrP(\frH)}P\,\mu_t(\di P) 
    =\int_{\scrP(\frH)}\mathbb{E}\bigl[P_t\,\big|\,P_0=P\bigr]\mu_{0,\CB}(\di P)\\
    &=\int_{\scrP(\frH)}\Phi_t(P)\,\mu_{0,\CB}(\di P)
    =\Phi_t\left(\int_{\scrP(\frH)}P\,\mu_{0,\CB}(\di P)\right)
    =\Phi_t(\Lambda(\mu_{0,\CB}))=\Phi_t(\rho_0)=\rho_t,
    \end{aligned}
    \end{equation}
    and similarly for $\Lambda(\nu_t)=\sigma_t$, the time-evolved measures $\mu_t$ and $\nu_t$ realize $\rho_t$ and $\sigma_t$, respectively. Consequently, $\fdiv(\mu_t\|\nu_t)\ge \fdiv^{\mathtt{UNR}}(\rho_t\|\sigma_t)$. Applying \eqref{eq:dissipation} with $s=\varepsilon$, bounding $\fdiv(\mu_\varepsilon\|\nu_\varepsilon)\leq\fdiv(\mu_{0,\CB}\|\nu_{0,\CB})$ by Proposition~\ref{prop:finitedpi} and letting $\varepsilon\downarrow0$, gives
    \[
        \fdiv^{\mathtt{UNR}}(\rho_t\|\sigma_t) + 2\int_0^t I_{f}(\mu_s\|\nu_s)\,\di s \le \fdiv(\mu_{0,\CB}\|\nu_{0,\CB}) = \fdiv^{\mathtt{UNR}}(\rho_0\|\sigma_0).
    \]
    Replacing $\fdiv^{\mathtt{UNR}}$ by $\fdiv^{\max}$ using Theorem~\ref{thm:maintheorem} concludes the proof.
\end{proof}

Theorem~\ref{thm:contraction} requires sufficient regularity of the densities, which follows from hypoellipticity under H\"{o}rmander's condition (cf.\ Remark~\ref{rem:hormander}) and is otherwise assumed.

\subsection{A $\Phi$-entropy inequality for the unraveling}
\label{sec:lsi}

The dissipation estimate \eqref{eq:dissipation} becomes quantitative when the relative $f$-Fisher information bounds the $f$-divergence from below. Such inequalities are called $\Phi$-entropy inequalities. For $f=f_\mathtt{KL}$, they are the logarithmic Sobolev inequalities \cite{logsobolev}, and for $f=f_\mathtt{2}$, the Poincar\'{e} inequalities. This subsection establishes such an inequality for the unraveling, and Section~\ref{sec:decay} turns it into a decay estimate.

Our starting point is the Stratonovich form of the unraveling. Since $P=|\psi\rangle\!\langle\psi|$, equation \eqref{eq:sseStrat} reads on $\scrP(\frH)$ as
\begin{equation}\label{eq:stratonovichP}\tag{\textsf{SSE}$_P$}
  \di P 
  = V_0(P)\,\di t + \sum\nolimits_j\gamma_j V_j(P)\circ\di W_j,
\end{equation}
where both the drift and the noise act on $\scrP(\frH)$ through the vector fields
\begin{equation}\label{eq:jumpfields}
  V_0(P) := -\iu\,[H,P],\qquad V_j(P) := \iu\,[S_j,P],\qquad j=1,\dots,m.
\end{equation}
Thus, \eqref{eq:stratonovichP} exhibits the unraveling as a \emph{Stratonovich} drift $V_0$ perturbed by \emph{Stratonovich} noise along $V_j$, $j=1,\ldots,m$. In particular, as noted after \eqref{eq:sseStrat}, each realization of the unraveling is a unitary evolution driven by a random Hamiltonian. Therefore, the map $\ket{\psi_0}\mapsto\ket{\psi_t}$ is a random unitary $U_t$, and $P_t = U_tP_0U_t^\dagger$. In other words, the unraveling defines a diffusion process on $\scrU(\frH)$, which the conjugation action then transforms into a stochastic process on $\scrP(\frH)$.

Note that each $V_j$ generates the unitary flow $P\mapsto e^{\iu tS_j}Pe^{-\iu tS_j}$ on $\scrP(\frH)$ and, being Hermitian and traceless, $V_j(P)$ is tangent to $\scrP(\frH)$ by construction. These fields are the fundamental vector fields of the $\scrU(\frH)$-action generated by $\iu S_j,-\iu H\in\mathfrak{u}(n)$, so that iterated brackets of the $V_j$ alone generate the fundamental fields of the \emph{jump algebra} $\mathrm{Lie}_\R\{\iu S_j\}$. Each of them is a Killing field for the Fubini--Study metric $\mathbf{g}_\mathtt{FS}$ in the sense of Definition~\ref{def:killing}---their flows are conjugations by one-parameter unitary groups, and $\mathbf{g}_\mathtt{FS}$ is $\scrU(\frH)$-invariant---and hence $\nu_\mathtt{FS}$-divergence free. With some abuse of notation, we interchangeably write $V$ for a field and for the derivation it defines, so that $VF(P):= \Der F(P)[V(P)]$ for a smooth function $F$ on $\scrP(\frH)$, extended componentwise to vector-valued $F$. Thus, $V_j(P)$ is a tangent vector at $P$, while $V_jF$ is a function, and $V_j^2$ is the associated second-order operator. 

We further consider the Lie algebra generated by the jump fields of \eqref{eq:jumpfields}, i.e.,
\begin{equation}\label{eq:dynlie}
  \mathfrak{g} := \mathrm{Lie}_{\R}\{V_1,\dots,V_m\},
\end{equation}
along with $G$, the group of transformations of $\scrP(\frH)$ that it generates. 

\begin{definition}[Bracket generation]\label{def:bracket}
A Lie algebra $\mathfrak{a}$ of vector fields on $\scrP(\frH)$ is said to be \emph{bracket generating} if $\mathfrak{a}$ spans $T_P\scrP(\frH)$ at every $P\in\scrP(\frH)$, or equivalently, since the fields are analytic (Nagano's orbit theorem), if the group it generates acts transitively on $\scrP(\frH)$.
\end{definition}

We use this for $\mathfrak{g}$, where it is a hypoellipticity property. In the Stratonovich form \eqref{eq:stratonovichP}, the generator is already in the \emph{H\"{o}rmander form},
\begin{equation}\label{eq:hormanderform}
  \scrL=V_0+\scrL_{\mathtt{sym}},\qquad
  \scrL_{\mathtt{sym}}:=\tfrac12\sum\nolimits_j\gamma_j^2\,V_j^2,
\end{equation}
with $V_0$ being a derivation. Setting $X_j:=(\gamma_j/\sqrt{2})V_j$, the second summand $\scrL_{\mathtt{sym}}$ can be expressed in divergence form. Indeed, since each $V_j$ is $\nu_\mathtt{FS}$-divergence free as noted after \eqref{eq:jumpfields}, one finds that $X_j^*=-X_j$ in $L^2(\nu_\mathtt{FS})$, and so
\begin{equation}\label{eq:lsymdivform}
  \scrL_{\mathtt{sym}}=\sum\nolimits_jX_j^2=-\sum\nolimits_jX_j^*X_j .
\end{equation}
In particular, $\scrL_{\mathtt{sym}}$ is symmetric and non-positive in $L^2(\nu_\mathtt{FS})$, and the semigroup $e^{t\scrL_{\mathtt{sym}}}$ it generates is also symmetric.

\begin{remark}\label{rem:hormander} A second-order operator written in the form $X_0+\sum_jX_j^2$ is said to satisfy \emph{H\"{o}rmander's condition} if the Lie algebra generated by the fields $X_0,X_1,\ldots,X_m$ spans the tangent space at every point. H\"{o}rmander's theorem then makes that operator hypoelliptic, i.e., it states that every distributional solution of $(X_0+\sum_jX_j^2)u=0$ is smooth. H\"{o}rmander's own statement \cite{hormander1967} is for operators of this form built from first-order fields on an open subset of $\R^n$, but since hypoellipticity is a local property and the bracket condition is invariant under diffeomorphism, the theorem holds also for a H\"{o}rmander system of vector fields on a manifold, which is the form stated in \cite[Thm.~III.1.2]{varopoulos1992}. 

By \eqref{eq:dynlie}, the algebra $\mathfrak{g}$ consists of vector fields on $\scrP(\frH)$, so bracket generation in the sense of Definition~\ref{def:bracket} means it spans every tangent space. This is H\"{o}rmander's condition for $\scrL_{\mathtt{sym}}$, so that by
H\"{o}rmander's theorem \cite[Thm.~III.1.2]{varopoulos1992},
\[
  \mathfrak{g}\ \text{bracket generating}\ \Longrightarrow\ \scrL_{\mathtt{sym}}\ \text{hypoelliptic} .
\]
It is also H\"{o}rmander's condition for the space-time operator. Indeed, since the Lie algebra generated by $\{V_1,\dots,V_m,\partial_t\}$ spans $\R{\times}T\!\scrP(\frH)$, we conclude that $\partial_t-\scrL_{\mathtt{sym}}$ is hypoelliptic. On the other hand, hypoellipticity of $\partial_t-\scrL$ follows once $\{V_1,\dots,V_m,V_0-\partial_t\}$ spans $\R{\times}T\!\scrP(\frH)$. In either case, the transition kernel is smooth for every $t>0$.
\end{remark}

Throughout this section, we write $p_t(P,Q)$ for the transition density of \eqref{eq:fokker-planck} with respect to $\nu_\mathtt{FS}$, so that $\mathbb{P}(P_t\in A| P_0=P)=\int_Ap_t(P,Q)\,\nu_\mathtt{FS}(\di Q)$.

\medskip
The first ingredient in deriving our quantitative estimate is the following assumption. 

\begin{property}[Uniform upper and lower kernel bounds]\label{ass:kernel}
There is some time $t_\star>0$ and constants $0<c_-(t_\star)\leq c_+(t_\star)<\infty$ with
\begin{equation}\label{eq:kernelbounds}\tag{\textsf{K}}
  c_-(t_\star) \leq p_{t_\star}(P,\cdot) \leq c_+(t_\star)
  \qquad\text{with respect to }\nu_\mathtt{FS},\text{ uniformly in }P\in\scrP(\frH).
\end{equation}
\end{property}

\begin{remark}\label{rem:kernelhyp}
Bounds like \eqref{eq:kernelbounds} on a compact manifold can often be obtained from hypoellipticity of the space-time operator $\partial_t-\scrL$ (see  Remark~\ref{rem:hormander}). In this case, the transition kernel is smooth away from $t=0$, and since $\scrP(\frH)$ is compact, it is uniformly bounded above in $P$ for every $t_\star>0$. The strict positivity of $p_{t_\star}$ provided by the lower bound, on the other hand, is less obvious and says that the process can reach everywhere by time $t_\star>0$. In Section~\ref{sec:examples} below, the depolarizing qubit satisfies both bounds for every $t_\star>0$. For the dephasing qubit, hypoellipticity gives the upper bound, while the lower bound is observed numerically but not proved.
\end{remark}

\begin{lemma}[Propagation and monotonicity]\label{lem:kernelmono}
Assume Property~\ref{ass:kernel}. Then, $\nu_\mathtt{FS}$ is the unique invariant measure of
\eqref{eq:fokker-planck}, and bounds of the form \eqref{eq:kernelbounds} hold at every time $t\geq t_\star$. Their constants may be chosen with $c_-$ non-decreasing and $c_+$ non-increasing, so that
\begin{equation}\label{eq:cmonotone}
  c_-(t_\star)\leq c_-(t)\leq c_+(t)\leq c_+(t_\star)\qquad\text{for every }t\geq t_\star ,
\end{equation}
and hence $c_-(t_0)\leq c_-(t)\leq c_+(t)\leq c_+(t_0)$ whenever $t\geq t_0\geq t_\star$. Both these choices are realized simultaneously by
\[
	c_-(t):=\inf_{P,Q}p_t(P,Q)\qquad\text{and}\qquad 
c_+(t):=\sup_{P,Q}p_t(P,Q),\] 
which we adopt unless stated otherwise. The examples of Section~\ref{sec:examples} use the cruder admissible choice $c_\pm(t)=1\pm\sup_{P,Q}|p_t(P,Q)-1|$, for which a closed-form bound is available.
\end{lemma}
\begin{proof}
The monotonicity of the constants follows easily from Chapman--Kolmogorov together with the invariance of $\nu_\mathtt{FS}$. Indeed, for $s,t>0$, we have that
\[
  p_{t+s}(P,Q)=\int_{\scrP(\frH)}p_t(P,R)\,p_s(R,Q)\,\nu_\mathtt{FS}(\di R) \geq c_-(t)\int_{\scrP(\frH)}p_s(R,Q)\,\nu_\mathtt{FS}(\di R) = c_-(t),
\]
where we used the invariance property of $\nu_\mathtt{FS}$. So any lower bound valid at time $t$ is still valid at every later time, and a similar computation shows the same for an upper bound. Applied at $t=t_\star$, this is both the propagation of \eqref{eq:kernelbounds} to all later times and the monotonicity of the constants. Uniqueness of the invariant measure follows from the strict positivity of $p_{t_\star}$.
\end{proof}

The second ingredient is the \emph{Poincar\'{e} inequality} for the invariant measure: there exists a constant $\lambda_{\mathtt{sym}}>0$ such that
\begin{equation}\label{eq:poincarefs}\tag{\textsf{P}$_\lambda$}
  \int_{\scrP(\frH)}\Gamma(v,v)\,\di\nu_\mathtt{FS} \geq \lambda_{\mathtt{sym}}\mathrm{Var}_{\nu_\mathtt{FS}}(v)
  \qquad\text{for every $v\in L^2(\nu_\mathtt{FS})$.}
\end{equation}
The Poincar\'{e} inequality is classical for H\"{o}rmander systems on compact manifolds, and \cite{lugiewiczzegarlinski} obtains it in that generality on the way to a logarithmic Sobolev inequality. That approach uses Sobolev embedding that requires $\dim\geq3$, which fails for $n=2$, where $\scrP(\frH)=\mathbb{CP}^{1}$. In Proposition~\ref{prop:bracket} below, we do not impose this restriction, but we only obtain an implicit constant. Note that in our setting, bracket generation yields the Poincar\'{e} inequality \eqref{eq:poincarefs}.

\begin{proposition}[Geometric criterion for spectral gap]\label{prop:bracket}
If $\mathfrak{g}$ is bracket generating, then $-\scrL_{\mathtt{sym}}$ has a positive spectral gap $\lambda_{\mathtt{sym}}>0$ in $L^2(\nu_\mathtt{FS})$, or, equivalently, \eqref{eq:poincarefs} holds.
\end{proposition}
\begin{proof}
By \eqref{eq:carre}, $\Gamma(F,F)=0$ if and only if $V_jF=0$ for every $j$, i.e., if and only if $F$ is constant along the flows of the $V_j$, which implies that $F$ is constant on the orbits of $G$. The operator $\scrL_{\mathtt{sym}}$ is symmetric in $L^2(\nu_\mathtt{FS})$ by \eqref{eq:lsymdivform}, and since $\mathfrak{g}$ is bracket generating, it satisfies H\"{o}rmander's subelliptic estimate $\|F\|_{H^\epsilon}\leq C\bigl(\|\scrL_{\mathtt{sym}}F\|_{L^2}+\|F\|_{L^2}\bigr)$ for some $\epsilon>0$ \cite{hormander1967}. On the compact manifold $\scrP(\frH)$, the embedding of $H^\epsilon$ into $L^2$ is compact by Rellich's theorem, so its resolvent is compact and its spectrum is discrete. Its kernel consists of the constants, since by the previous paragraph $\Gamma(F,F)=0$ forces $F$ to be constant on the orbits of $G$, which is all of $\scrP(\frH)$ by bracket generation. Hence, its first non-zero eigenvalue, written $\lambda_{\mathtt{sym}}$, is positive, and \eqref{eq:poincarefs} is the associated variational characterization.
\end{proof}

For $f=f_\mathtt{KL}$, bracket generation gives a logarithmic Sobolev inequality (LSI) for $\nu_\mathtt{FS}$, with an explicit constant, as soon as the semigroup is ultracontractive at a single time. 

\begin{proposition}[Logarithmic Sobolev inequality]\label{prop:lsifs}
Let $\mathfrak{g}$ be bracket generating and $\lambda_{\mathtt{sym}}>0$ be the constant of Proposition~\ref{prop:bracket}. Further, let $t_1>0$ be a time at which the transition density of $e^{t_1\scrL_{\mathtt{sym}}}$ with respect to $\nu_\mathtt{FS}$ is at most $c_{\mathtt{sym}}(t_1)<+\infty$. Then,
\begin{equation}\label{eq:lsifs}\tag{\textsf{LSI}$_\lambda$}
  I_\mathtt{KL}(\mu\|\nu_\mathtt{FS})\geq\lambda_{\mathtt{LSI}}\,\KL(\mu\|\nu_\mathtt{FS}),
  \qquad\text{for every $\mu\ll\nu_\mathtt{FS}$},
\end{equation}
with
\begin{equation*}
  \lambda_{\mathtt{LSI}} :=\frac{\lambda_{\mathtt{sym}}}{2\bigl(\lambda_{\mathtt{sym}}t_1+\log c_{\mathtt{sym}}(t_1)+2\bigr)} > 0 .
\end{equation*}
\end{proposition}
\begin{proof}
By \eqref{eq:lsymdivform}, the semigroup $e^{t\scrL_{\mathtt{sym}}}$ is symmetric in $L^2(\nu_\mathtt{FS})$, and by Proposition~\ref{prop:bracket}, it satisfies the $\lambda_{\mathtt{sym}}$-Poincar\'{e} inequality \eqref{eq:poincarefs}, which also makes it ergodic. Moreover, the assumption on the transition density provides the ultracontractivity bound $\|e^{t_1\scrL_{\mathtt{sym}}}\|_{L^1\to L^\infty}\leq c_{\mathtt{sym}}(t_1)$. Together, they satisfy the assumptions of \cite[Theorem~3.9]{gaojungelaracuenteli}, thus providing a modified LSI with the constant $\lambda_{\mathtt{LSI}}$, which is equivalent to \eqref{eq:lsifs} since $\scrL_{\mathtt{sym}}$ generates a diffusive semigroup \cite[Thm.~5.2.1]{logsobolev}.
\end{proof}

\begin{remark}
    Bracket generation makes $\partial_t-\scrL_{\mathtt{sym}}$ hypoelliptic by Remark~\ref{rem:hormander}, so the transition density of $e^{t\scrL_{\mathtt{sym}}}$ is smooth for every $t>0$. Since $\scrP(\frH)$ is compact, the transition density is bounded for all $t>0$. Thus, any $t_1>0$ would work. We note, however, that this is a bound on the kernel of the symmetrized generator $\scrL_{\mathtt{sym}}$, and that \eqref{eq:kernelbounds} does not necessarily provide the constant $c_{\mathtt{sym}}(t_1)$.
\end{remark}

The estimates below depend on the projective distance between the two propagated measures, defined as follows. For probability measures $\mu,\nu$ on $\scrP(\frH)$ that are mutually absolutely continuous, \emph{Hilbert's projective metric} is defined by
\[
  d_\mathtt{H}(\mu,\nu) := \log\frac{\sup\,u}{\inf\,u},\qquad u=\frac{\di\mu}{\di\nu}.
\]

\begin{lemma}[Projective contraction along the flow]\label{lem:projective}
Assume Property~\ref{ass:kernel} and let $t_0\geq t_\star$. Further, let $c_\pm(t_0)$ be the kernel bounds of \eqref{eq:cmonotone} and set
\[
  \varkappa := \tanh\biggl(\frac12\log\frac{c_+(t_0)}{c_-(t_0)}\biggr)<1 .
\]
Let $\mu_t,\nu_t$ solve \eqref{eq:fokker-planck} with probability measures as initial data, and put
\begin{equation}\label{eq:betadef}
  \beta(t) := \max\bigl\{d_\mathtt{H}(\mu_t,\nu_t), d_\mathtt{H}(\nu_t,\nu_\mathtt{FS})\bigr\} .
\end{equation}
Then the following holds:
\begin{enumerate}[label=(\roman*)]
  \item $\beta$ is non-increasing on $(0,\infty)$, and
    $\beta(t_0)\leq2\log\bigl(c_+(t_0)/c_-(t_0)\bigr)< +\infty$;
  \item $\beta(t+t_0)\leq\varkappa\,\beta(t)$ for every $t>0$, and consequently
    \begin{equation}\label{eq:betadecay}
      \beta(t) \leq \beta(t_0)\,\varkappa^{\lfloor t/t_0\rfloor-1}\qquad\text{for every $t\geq t_0$};
    \end{equation}
  \item $\beta(t)\to0$ and $\int_{t_0}^{\infty}\beta(r)\,\di r\leq\beta(t_0)\,t_0/(1-\varkappa)$;
  \item The constants of \eqref{eq:kernelbounds} may be chosen such that $c_\pm(t)\to1$ as $t\to\infty$.
\end{enumerate}
\end{lemma}
\begin{proof}
All four parts of the statement rely on Birkhoff's theorem. On the cone of positive measures, equipped with Hilbert's projective metric $d_\mathtt{H}$, a linear map preserving the cone is Lipschitz with constant $\tanh(\Delta/4)$, where $\Delta$ is the projective diameter of its image \cite{birkhoff1957}. In particular, such a map is always non-expansive, and it is a strict contraction as soon as $\Delta<\infty$. Clearly, the semigroup $e^{s\scrL^\dagger}$ preserves the cone, and fixes $\nu_\mathtt{FS}$.

$(i)$ Non-expansiveness applied to the pairs $(\mu_t,\nu_t)$ and $(\nu_t,\nu_\mathtt{FS})$ shows that both distances in \eqref{eq:betadef} are non-increasing, hence so is their maximum. For the value at $t_0$, \eqref{eq:cmonotone} puts both $\di\mu_{t_0}/\di\nu_\mathtt{FS}$ and $\di\nu_{t_0}/\di\nu_\mathtt{FS}$ in $[c_-(t_0),c_+(t_0)]$, so each of the two distances is at
most $2\log(c_+(t_0)/c_-(t_0))$, independent of the initial data.

$(ii)$ From \eqref{eq:cmonotone}, we know that the kernel of $e^{t_0\scrL^\dagger}$ lies between $c_-(t_0)$ and $c_+(t_0)$, so its projective diameter is at most $2\log(c_+(t_0)/c_-(t_0))$ and its Lipschitz constant at most $\varkappa$. Applying this to the two pairs then gives the first inequality. For the second, write $t=kt_0+\tau$ with $k=\lfloor t/t_0\rfloor\geq1$ and $\tau\in[0,t_0)$, and apply the first inequality $k-1$ times starting from $t_0+\tau$, then $(i)$ at $t_0+\tau\geq t_0$.

$(iii)$ Since $\varkappa<1$, the limit follows directly from \eqref{eq:betadecay}. On the interval $[kt_0,(k+1)t_0)$ of length $t_0$, that bound reads $\beta\leq\beta(t_0)\varkappa^{k-1}$, so using the geometric series over $k\geq1$ gives
\[
	\int_{t_0}^{\infty}\beta(r)\,\di r \leq \sum_{k\ge 1} t_0\beta(t_0) \varkappa^{k-1} = \frac{t_0\beta(t_0)}{(1-\varkappa)}.
\]

$(iv)$ With $\mu_0=\nu_0=\delta_P$, the two measures coincide for all $t$, and hence, $d_\mathtt{H}(\mu_t,\nu_t)=0$ and $\beta(t)=d_\mathtt{H}(p_t(P,\cdot),\nu_\mathtt{FS})$. Since the bound in \eqref{eq:betadecay} depends on the initial data only through $\beta(t_0)$, which by $(i)$ is at most $2\log(c_+(t_0)/c_-(t_0))$, we deduce that 
\[
    \frac{\sup_Qp_t(P,Q)}{\inf_Qp_t(P,Q)}\to1\qquad \text{uniformly in $P$},
\]
and since $p_t(P,\cdot)$ integrates to $1$ against $\nu_\mathtt{FS}$, we easily conclude that both the supremum and the infimum tend to $1$ as $t\to \infty$.
\end{proof}

\begin{lemma}[From spectral gap to a $\Phi$-entropy inequality]\label{lem:perturbation}
Let the Poincar\'e inequality \eqref{eq:poincarefs} hold, let $f$ be normalized operator convex with $f''>0$, and let $g_f$ be as in Proposition~\ref{prop:opconv}. Let $\mu,\nu$ be probability measures on $\scrP(\frH)$ with $\beta :=\max\{d_\mathtt{H}(\mu,\nu),\,d_\mathtt{H}(\nu,\nu_\mathtt{FS})\}<+\infty$. Then,
\begin{equation}\label{eq:gaptophi}
  I_f(\mu\|\nu) \geq \lambda_{\mathtt{sym}}\xi_f(\beta)\,\fdiv(\mu\|\nu),\qquad
  \xi_f(\beta) := e^{-\beta}\,
  \frac{\frac{1}{2}f''(e^{\beta})}{g_f(e^{-\beta})} .
\end{equation}
The factor $\xi_f$ depends only on $f$ and $\beta$, is non-increasing, and satisfies $\xi_f(0)=1$.
\end{lemma}
\begin{proof}
Write $h:=\di\nu/\di\nu_\mathtt{FS}$ and $u:=\di\mu/\di\nu$. Since $\int u\,\di\nu=1$, the infimum of $u$ is at most $1$ and its supremum at least $1$, so $d_\mathtt{H}(\mu,\nu)\leq\beta$ gives $e^{-\beta}\leq u\leq e^{\beta}$ and $d_\mathtt{H}(\nu,\nu_\mathtt{FS})\leq\beta$ gives $\inf h\geq e^{-\beta}\sup h$. We begin from
the Fisher information side of the inequality. Notice that
\begin{equation}\label{eq:fsecond}
  f''(x)=2b+\int_{[0,\infty)}\frac{2(1+s)^2}{(x+s)^3}\,\tau_f(\di s),
\end{equation}
which is a sum of positive non-increasing functions. Thus, $f''$ is non-increasing and $f''\geq f''(e^{\beta})>0$ on the range of $u$. Therefore, we obtain
\[
  I_f(\mu\|\nu)=\frac12\int f''(u)\,\Gamma(u,u)\,\di\nu \geq \frac{1}{2}f''(e^{\beta})\int\Gamma(u,u)\,\di\nu .
\]
Bounding $\nu$ from below by $(\inf h)\,\nu_\mathtt{FS}$ and $\mathrm{Var}_{\nu_\mathtt{FS}}$ by $(\sup h)^{-1}\mathrm{Var}_{\nu}$, and applying \eqref{eq:poincarefs} yields
\[
  \int\Gamma(u,u)\,\di\nu \geq (\inf h)\,\lambda_{\mathtt{sym}}\mathrm{Var}_{\nu_\mathtt{FS}}(u)
  \geq \lambda_{\mathtt{sym}}e^{-\beta}\,\mathrm{Var}_{\nu}(u)
  = \lambda_{\mathtt{sym}}e^{-\beta}\int(u-1)^2\,\di\nu ,
\]
where the last step holds since $\int u\,\di\nu=1$.

For the $f$-divergence, by \eqref{eq:opconvrep} and $u\geq e^{-\beta}$, we simply deduce
\[
  \fdiv(\mu\|\nu) = \int (u-1)^2g_f(u)\,\di\nu \leq g_f(e^{-\beta})\int(u-1)^2\,\di\nu,
\]
Putting the inequalities together yields \eqref{eq:gaptophi} as asserted.
\end{proof}

For the KL divergence $\KL$, a second approach is available using the logarithmic Sobolev inequality \eqref{eq:lsifs} of Proposition~\ref{prop:lsifs}. It depends on less of the flow, since only the second term of \eqref{eq:betadef} enters, and it is sharper than the other when $\beta$ is large.

\begin{proposition}[Transferring a logarithmic Sobolev inequality]\label{prop:lsi}
Let \eqref{eq:lsifs} hold with $\lambda_{\mathtt{LSI}}>0$. If $\nu$ is any probability measure mutually absolutely continuous with $\nu_\mathtt{FS}$, then
\begin{equation}\label{eq:lsitransfer}
  I_\mathtt{KL}(\mu\|\nu) \geq \lambda_{\mathtt{LSI}} e^{-d_\mathtt{H}(\nu,\nu_\mathtt{FS})}\,\KL(\mu\|\nu)
  \qquad\text{for every $\mu\ll\nu$}.
\end{equation}
Note that only the second of the two terms in \eqref{eq:betadef} enters here, and that $\lambda_{\mathtt{sym}}$ does not appear. In particular, any logarithmic Sobolev inequality for $\nu_\mathtt{FS}$ is transferred at the same cost.
\end{proposition}
\begin{proof}
Write $h:=\di\nu/\di\nu_\mathtt{FS}$ and $u:=\di\mu/\di\nu$, and set
\[
  \mathrm{Ent}_\zeta(g) := \int g\log g\,\di\zeta-\biggl(\int g\,\di\zeta\biggr)\log\biggl(\int g\,\di\zeta\biggr),
  \qquad
  \mathcal{E}_\zeta(g):=\frac{1}{2}\int\frac{\Gamma(g,g)}{g}\,\di\zeta ,
\]
for a probability measure $\zeta$. Then, $\KL(\mu\|\zeta)=\mathrm{Ent}_\zeta(\di\mu/\di\zeta)$ and $I_\mathtt{KL}(\mu\|\zeta)=\mathcal{E}_\zeta(\di\mu/\di\zeta)$. Since both functionals are homogeneous of degree one in $g$, \eqref{eq:lsifs} reads $\mathrm{Ent}_{\nu_\mathtt{FS}}(g)\leq\lambda_{\mathtt{LSI}}^{-1}\mathcal{E}_{\nu_\mathtt{FS}}(g)$ for every $g\geq0$. Writing $\mathrm{Ent}_\nu(g)=\inf_{s>0}\int(g\log g-g\log s-g+s)\,\di\nu$ as an infimum of integrals of non-negative integrands, we find that $h\leq\sup h$ gives $\mathrm{Ent}_{\nu}(u)\leq(\sup h)\,\mathrm{Ent}_{\nu_\mathtt{FS}}(u)$, while $h\geq\inf h$ gives $\mathcal{E}_{\nu_\mathtt{FS}}(u)\leq(\inf h)^{-1}\mathcal{E}_{\nu}(u)$. Putting the inequalities together then yields
\begin{align*}
  \KL(\mu\|\nu)=\mathrm{Ent}_{\nu}(u) &\leq(\sup h)\,\mathrm{Ent}_{\nu_\mathtt{FS}}(u) \\
  &\leq\frac{\sup h}{\lambda_{\mathtt{LSI}}}\,\mathcal{E}_{\nu_\mathtt{FS}}(u)
  \leq\frac{\sup h}{\lambda_{\mathtt{LSI}}\,\inf h}\,\mathcal{E}_{\nu}(u)
  =\frac{e^{d_\mathtt{H}(\nu,\nu_\mathtt{FS})}}{\lambda_{\mathtt{LSI}}}\,I_\mathtt{KL}(\mu\|\nu),
\end{align*}
where we used $\sup h/\inf h=e^{d_\mathtt{H}(\nu,\nu_\mathtt{FS})}$.
\end{proof}

\begin{remark}\label{rem:holleystroock}
The usual approach to proving decay is to establish a $\Phi$-entropy inequality for $f$ at $\nu_\mathtt{FS}$ and transfer it to $\nu$. For the Poincar\'{e} and logarithmic Sobolev inequalities, the transfer is the bounded perturbation argument of Holley and Stroock \cite{holleystroock1987}, which needs no curvature, and Proposition~\ref{prop:lsi} is exactly this argument. For a general $f$, however, both steps meet obstacles. First, the theory of $\Phi$-entropy inequalities in \cite{bolleygentil2010,chafai2004,dolbeaultli2018} is developed for $1/f''$ concave, which excludes the extreme rays $f_x$ of Proposition~\ref{prop:opconv}, where $1/f_x''$ is strictly convex. $\KL$ and $\chi^2$, however, do lie in this class. Second, the criteria given there for establishing such an inequality at $\nu_\mathtt{FS}$ assume a lower curvature bound, which is generally unavailable for a sub-Riemannian sum of squares. This second obstacle affects all $f$-divergences, which is why \eqref{eq:lsifs} is taken from results that do not require it. Lemma~\ref{lem:perturbation} avoids both obstacles by passing through the Poincar\'{e} inequality of Proposition~\ref{prop:bracket}, at the cost of the factor $\xi_f(\beta)$.
\end{remark}

\begin{remark}[Hamiltonian drift]\label{rem:drift}
The Hamiltonian enters the following estimates only through the mixing constants. Since $H$ does not appear in \eqref{eq:carre}, the spectral gap $\lambda_{\mathtt{sym}}$ of Proposition~\ref{prop:bracket} is a spectral quantity of $\scrL_{\mathtt{sym}}$ alone, and is independent of $H$. The kernel bounds $c_\pm$ are not, and therefore neither are $\beta$ nor the rate $\lambda_{\mathtt{sym}}\xi_f(\beta)$ of Lemma~\ref{lem:perturbation}. It is the spectral gap, not the decay rate, that is insensitive to the Hamiltonian.
\end{remark}

\subsection{Exponential decay of maximal $f$-divergences}
\label{sec:decay}

Building on Section~\ref{sec:lsi}, we have the following result.

\begin{lemma}[Dissipation along the flow]\label{lem:dissipation}
Let $\mathfrak{g}$ be bracket generating and assume Property~\ref{ass:kernel}. Further, let $t_0\geq t_\star$, and let $f$ be a smooth, normalized, operator-convex function with $f''>0$. Let $\rho_0,\sigma_0\in\scrS_+(\frH)$ and let $\mu_t,\nu_t$ solve
\eqref{eq:fokker-planck} with initial data $\mu_{0,\CB},\nu_{0,\CB}$. Then, with $\lambda_{\mathtt{sym}}>0$ given by Proposition~\ref{prop:bracket}, $\beta(t)$ by Lemma~\ref{lem:projective}, and $\xi_f$ by Lemma~\ref{lem:perturbation}, we obtain
\begin{equation}\label{eq:diffdecay}
  \frac{\di}{\di t}\fdiv(\mu_t\|\nu_t) \leq -2\lambda_{\mathtt{sym}}\xi_f(\beta(t))\fdiv(\mu_t\|\nu_t)
  \qquad\text{for every }t\geq t_0 .
\end{equation}
For the KL divergence, Proposition~\ref{prop:lsi} provides
$\lambda_{\mathtt{LSI}} e^{-d_\mathtt{H}(\nu_t,\nu_\mathtt{FS})}$ in place of $\lambda_{\mathtt{sym}}\xi_f(\beta(t))$ (see Corollary~\ref{cor:lsidecay}).
\end{lemma}
\begin{proof}
For $t\geq t_0$, the pair $(\mu_t,\nu_t)$ satisfies the hypothesis of Lemma~\ref{lem:perturbation} with the $\beta(t)$ of \eqref{eq:betadef}, which is finite by Lemma~\ref{lem:projective}$(i)$. Lemma~\ref{lem:perturbation} applied with $\nu=\nu_t$ then yields the $\Phi$-entropy inequality $I_f(\mu_t\|\nu_t)\geq\lambda_{\mathtt{sym}}\xi_f(\beta(t))\fdiv(\mu_t\|\nu_t)$. For $t\geq t_0$ the densities are smooth and, by \eqref{eq:cmonotone}, bounded on both sides, so $r\mapsto\fdiv(\mu_r\|\nu_r)$ is differentiable and \eqref{eq:dissipation} may be differentiated at $r=t$, giving $\frac{\di}{\di t}\fdiv(\mu_t\|\nu_t)\leq-2I_f(\mu_t\|\nu_t)$. Combining these two inequalities gives the assertion \eqref{eq:diffdecay}.
\end{proof}

\begin{remark}
Notice that the decay rate in \eqref{eq:diffdecay} is not constant. Indeed, by Lemma~\ref{lem:projective}$(iii)$, $\beta(t)\to0$ and, since $\xi_f$ is continuous with $\xi_f(0)=1$, the dissipation eventually runs at the full spectral gap as $t\gg 1$. Integrating the rate instead of freezing it at $t_0$ makes the time at which the full spectral gap takes effect explicit (see Figure~\ref{fig:timescales}). This approach is not limited to this case. Whenever the decay rate improves along the flow, integration converts it into a delay. A similar approach is used for the Hellinger--Kantorovich--Boltzmann gradient flow in \cite{lieromielketsezhu}, where $\beta(t)$ corresponds to the extreme eigenvalues of the covariance.
\end{remark}

\begin{figure}[h]
\centering
\begin{tikzpicture}[>=Stealth, line width=1pt]

  % Main axis line with arrow to infinity
  \draw[->] (0,0) -- (12,0);

  % Tick marks at 0, t0, t_f
  \draw (0,0.2) -- (0,-0.2);
  \draw (3,0.2) -- (3,-0.2);
  \draw (7.5,0.2) -- (7.5,-0.2);

  % Point labels below the axis
  \node[below] at (0,-0.2) {$0$};
  \node[below] at (3,-0.2) {$t_0$};
  \node[below] at (7.5,-0.2) {$t_f$};

  % Infinity label at the end of the arrow
  \node[right] at (12,0) {$+\infty$};

  % Regime labels above the axis, centered in each interval
  \node[above] at (1.5,0.1)  {DPI};
  \node[above] at (5.25,0.1)  {decay rate $2\lambda_\mathtt{sym} \xi_0$};
  \node[above] at (9.75,0.1) {decay rate $2\lambda_\mathtt{sym}$};

\end{tikzpicture}
\caption{\emph{Timescales.} Theorem~\ref{thm:decay} suggests that the $f$-divergence decays at different rates across time regimes. On the interval $[0,t_0]$, the quantum DPI holds without rate, on $[t_0, t_f]$, it decays at the reduced rate $2\lambda_\mathtt{sym} \xi_0$, and after $t_f$, it decays with the full rate $2\lambda_\mathtt{sym}$.}
\label{fig:timescales}
\end{figure}

\begin{theorem}[Exponential decay]\label{thm:decay}
Under the assumptions of Lemma~\ref{lem:dissipation}, set
\begin{equation}\label{eq:delay}
  \xi_0:=\xi_f(\beta(t_0))\qquad \text{and}\qquad t_f:=t_0+\int_{t_0}^{\infty}\bigl[1-\xi_f(\beta(r))\bigr]\di r ,
\end{equation}
which is finite. Then for all $t\geq t_0$,
\begin{equation}\label{eq:refined}
  \fdiv^{\max}(\rho_t\|\sigma_t) \leq
  \exp\biggl(-2\lambda_{\mathtt{sym}}\!\int_{t_0}^{t}\!\xi_f(\beta(r))\,\di r\biggr)\fdiv^{\max}(\rho_0\|\sigma_0),
\end{equation}
and consequently
\begin{equation}\label{eq:twobounds}
  \fdiv^{\max}(\rho_t\|\sigma_t) \leq
  \exp\Bigl(-2\lambda_{\mathtt{sym}}\max\bigl\{\xi_0(t-t_0), t-t_f\bigr\}\Bigr)\fdiv^{\max}(\rho_0\|\sigma_0) .
\end{equation}
\end{theorem}
\begin{proof}
Integrating \eqref{eq:diffdecay} by Gr\"{o}nwall's inequality over $[t_0,t]$ gives
\[
	\fdiv(\mu_t\|\nu_t)\leq\exp\biggl(-2\lambda_{\mathtt{sym}}\int_{t_0}^{t}\xi_f(\beta(r))\di r\biggr)\fdiv(\mu_{t_0}\|\nu_{t_0}),
\]
while on $[0,t_0]$ the classical DPI gives $\fdiv(\mu_{t_0}\|\nu_{t_0})\leq\fdiv(\mu_{0,\CB}\|\nu_{0,\CB})$. Theorem~\ref{thm:maintheorem} at times $t$ and $0$ then converts the two ends into $\fdiv^{\max}$, which is \eqref{eq:refined}.

For \eqref{eq:twobounds}, note that since $\beta$ is non-increasing by Lemma~\ref{lem:projective}$(i)$ and $\xi_f$ is non-increasing, we then have that $r\mapsto\xi_f(\beta(r))$ is non-decreasing. Therefore, bounding the integrand below by $\xi_0$ gives the first term in the maximum, while writing $\xi_f=1-(1-\xi_f)$ and extending the integral to $\infty$ gives the second.

We are left to show that $t_f$ is finite. By Lemma~\ref{lem:perturbation}, the map $\xi_f$ is non-increasing, and it is locally Lipschitz because $f$ is smooth on $(0,\infty)$. With $\xi_f(0)=1$ this gives $0\leq1-\xi_f(\beta)\leq L_\xi\beta$ on $[0,\beta(t_0)]$, where $L_\xi$ is the Lipschitz constant there. Since Lemma~\ref{lem:projective}$(iii)$ bounds
$\int_{t_0}^\infty\beta(r)\,\di r$ by $\beta(t_0)t_0/(1-\varkappa)$, we conclude that
\[
  t_f \leq t_0\biggl(1+\frac{L_\xi \beta(t_0)}{1-\varkappa}\biggr)<+\infty.	\qedhere
\]
\end{proof}

\begin{remark}
The constants $\xi_0$ and $t_f$ in Theorem~\ref{thm:decay} depend on $\rho_0$ and $\sigma_0$ through $\beta$. However, they can be made independent of the initial states. Indeed, by Lemma~\ref{lem:projective}$(i)$--$(ii)$, we find 
\[
	\beta(t)\leq\bar\beta(t):=2\log\bigl(c_+(t_0)/c_-(t_0)\bigr)\varkappa^{\lfloor t/t_0\rfloor-1}\qquad \text{for every $t\geq t_0$.}
\]
Since $\xi_f$ is non-increasing, Theorem~\ref{thm:decay} holds with $\bar\beta$ in place of $\beta$, uniformly over all faithful states $\rho_0,\sigma_0\in\scrS_+(\frH)$.
\end{remark}

For the BS-entropy, the logarithmic Sobolev inequality in Proposition~\ref{prop:lsifs} gives a second estimate of the same form. Its rate involves only the second of the two distances in \eqref{eq:betadef}, and $\lambda_{\mathtt{LSI}}$ in place of $\lambda_{\mathtt{sym}}$.

\begin{corollary}[Decay of $\BS$ via LSI]\label{cor:lsidecay}
Under the assumptions of Lemma~\ref{lem:dissipation}, let \eqref{eq:lsifs} hold with $\lambda_{\mathtt{LSI}}>0$, for instance as given by Proposition~\ref{prop:lsifs}, and set
\[
  \alpha(t):=d_\mathtt{H}(\nu_t,\nu_\mathtt{FS}),\qquad
  t_\mathtt{LSI}:=t_0+\int_{t_0}^{\infty}\bigl[1-e^{-\alpha(r)}\bigr]\di r .
\]
Then for all $t\geq t_0$,
\begin{align*}
  \BS(\rho_t\|\sigma_t) &\leq
  \exp\biggl(-2\lambda_{\mathtt{LSI}}\!\int_{t_0}^{t}\!e^{-\alpha(r)}\,\di r\biggr)\BS(\rho_0\|\sigma_0) \\
  &\leq \exp\Bigl(-2\lambda_{\mathtt{LSI}}\max\bigl\{e^{-\alpha(t_0)}(t-t_0),\, t-t_\mathtt{LSI}\bigr\}\Bigr)\BS(\rho_0\|\sigma_0),
\end{align*}
where $\alpha(t_0)\leq\log\bigl(c_+(t_0)/c_-(t_0)\bigr)$ and $t_\mathtt{LSI}\leq t_0\bigl(1+\alpha(t_0)/(1-\varkappa)\bigr)$.
\end{corollary}
\begin{proof}
For $t\geq t_0$, the measure $\nu_t$ is mutually absolutely continuous with $\nu_\mathtt{FS}$ by \eqref{eq:cmonotone}, so Proposition~\ref{prop:lsi} with $\nu=\nu_t$ gives $I_\mathtt{KL}(\mu_t\|\nu_t)\geq\lambda_{\mathtt{LSI}}e^{-\alpha(t)}\KL(\mu_t\|\nu_t)$. This is the case $f=f_\mathtt{KL}$ of Lemma~\ref{lem:dissipation}, with $\lambda_{\mathtt{LSI}}e^{-\alpha(t)}$ in place of $\lambda_{\mathtt{sym}}\xi_f(\beta(t))$. The proof of Theorem~\ref{thm:decay} then applies in the same way, since $\fdiv^{\max}=\BS$ for $f=f_\mathtt{KL}$, and $\alpha$ is non-increasing by the non-expansiveness used in Lemma~\ref{lem:projective}$(i)$.

For the bounds on the constants, $\di\nu_{t_0}/\di\nu_\mathtt{FS}$ takes values in $[c_-(t_0),c_+(t_0)]$, which gives $\alpha(t_0)\leq\log(c_+(t_0)/c_-(t_0))$. The proof of Lemma~\ref{lem:projective}$(ii)$--$(iii)$ applies to the pair $(\nu_t,\nu_\mathtt{FS})$ and gives $\int_{t_0}^\infty\alpha(r)\,\di r\leq\alpha(t_0)\,t_0/(1-\varkappa)$. Together with $1-e^{-\alpha}\leq\alpha$, this bounds $t_\mathtt{LSI}$.
\end{proof}

\begin{remark}\label{rem:reducible}
When $\mathfrak{g}$ is not bracket generating, the Poincar\'{e} inequality \eqref{eq:poincarefs} may fail, and with it the arguments for Lemma~\ref{lem:dissipation}. Take $n=2$ with $S_1=\sigma_x$ and $H=\sigma_z$, where $\mathfrak{g}$ consists of the fields of $\mathrm{Lie}_\R\{\iu\sigma_x\}=\R\,\iu\sigma_x$ and is one-dimensional. Its orbits are circles on $\scrP(\frH)\cong\mathbb{S}^2$, so a smooth non-constant function constant on them exists. By \eqref{eq:carre}, it satisfies $\Gamma(F,F)=0$, i.e., \eqref{eq:poincarefs} fails. Yet, the associated process is non-degenerate, i.e., brackets of the jump field with the drift field span at every point (see Section~\ref{sec:examples}), and so Remark~\ref{rem:kernelhyp} gives the upper half of Property~\ref{ass:kernel}, while the lower half is assumed as it is throughout. So the approach via a functional inequality fails, but the decay itself may not.
\end{remark}

\subsection{Sharp rates without a functional inequality}
\label{sec:sharprate}

A limitation of Theorem~\ref{thm:decay} is that it requires $\mathfrak{g}$ to be bracket generating, which excludes the degenerate systems of Remark~\ref{rem:reducible}, where the jump fields alone do not span the full space. We now remove this restriction. 

Throughout this subsection, we write $\scrL^*$ for the adjoint of $\scrL$ in $L^2(\nu_\mathtt{FS})$. Written in terms of densities, $\mu_t=u_t\,\nu_\mathtt{FS}$, the Fokker--Planck equation \eqref{eq:fokker-planck} reads $\partial_tu_t=\scrL^*u_t$. We equip $L^2(\nu_\mathtt{FS})$ with the norm $\|\cdot\|_{L^2(\nu_\mathtt{FS})}$, and consider the set of mean-zero functions $\mathbf{1}^\perp:=\{w\in L^2(\nu_\mathtt{FS}):\int w\,\di\nu_\mathtt{FS}=0\}$. Since $\nu_\mathtt{FS}$ is a probability measure, the constant function $\mathbf{1}$ has unit norm and $\|w\|_{L^2(\nu_\mathtt{FS})}\leq\|w\|_\infty$ for every bounded function $w$. The estimates below depend on the semigroup $\mathscr{T}_t:=e^{t\scrL^*}$ only through a constant $M$ and a decay rate $\lambda$,  for which the following inequality holds
\begin{equation}\label{eq:gapdef}
  \bigl\|\mathscr{T}_t w\bigr\|_{L^2(\nu_\mathtt{FS})} \leq M e^{-\lambda t}\,\|w\|_{L^2(\nu_\mathtt{FS})},
  \qquad \text{for every $w\in\mathbf{1}^\perp$,\, $t\geq0$}.
\end{equation}
Such a semigroup $\mathscr{T}_t$ is said to be of \emph{type} $(M,-\lambda)$, following the terminology of \cite{engelnagel}. Note that the subspace $\mathbf{1}^\perp$ is invariant under $\mathscr{T}_t$ since $\nu_\mathtt{FS}$ is invariant. This makes \eqref{eq:gapdef} the relevant estimate, since the difference of two probability densities is mean-zero. And necessarily $M\geq1$, with equality when $\mathscr{T}_t$ is a contraction on $\mathbf{1}^\perp$. Morally, $M>1$ records a `time lag' in the semigroup before it reaches its decay rate.

\medskip
Our first result here shows that the kernel bounds \eqref{eq:kernelbounds} already provide such a pair: the semigroup $\mathscr{T}_t = e^{t\scrL^*}$ is of type $(M,-\lambda)$ with explicit, but not necessarily optimal, constants $M$ and $\lambda$ obtained from $c_\pm(t_0)$. In particular, \eqref{eq:gapdef} is not really an assumption.

\begin{proposition}[Existence of a decay pair]\label{prop:gappair}
Assume Property~\ref{ass:kernel} and let $t_0\geq t_\star$. Further, let $c_\pm(t_0)$ be the kernel bounds of \eqref{eq:cmonotone} and set $\varepsilon(t_0):=\min\bigl\{1-c_-(t_0),\,c_+(t_0)-1\bigr\}\in[0,1)$. If $\varepsilon(t_0)>0$, then $\mathscr{T}_t=e^{t\scrL^*}$ is of type $(M,-\lambda)$ with the explicit pair
\[
  \lambda = -\frac{\log\varepsilon(t_0)}{t_0}>0,
  \qquad M=\frac{1}{\varepsilon(t_0)}\geq1 .
\]
If $\varepsilon(t_0)=0$, then $\mathscr{T}_t$ is of type $(e^{\lambda t_0},-\lambda)$ for \emph{every} $\lambda>0$.
\end{proposition}
\begin{proof}
By Jensen's inequality, we obtain $|e^{t\scrL}F|^2\leq e^{t\scrL}|F|^2$ pointwise. Therefore, integrating against the invariant $\nu_\mathtt{FS}$ shows that $e^{t\scrL}$ is a contraction. Therefore, its adjoint $\mathscr{T}_t$ is also a contraction. This covers the times below $t_\star$.

For $t\geq t_\star$, where Lemma~\ref{lem:kernelmono} provides a transition density $p_t$ with respect to $\nu_\mathtt{FS}$, the adjoint semigroup is given by $\mathscr{T}_t w(Q)=\int p_t(R,Q)w(R)\,\nu_\mathtt{FS}(\di R)$. Note that this kernel integrates to $1$ w.r.t.\ $\nu_\mathtt{FS}$ in either of its arguments. Indeed,
\[
  \int p_t(R,Q)\,\nu_\mathtt{FS}(\di Q)=1
  \quad\text{for every $R$},
  \qquad
  \int p_t(R,Q)\,\nu_\mathtt{FS}(\di R)=1
  \quad\text{for every $Q$},
\]
where the first holds because $p_t(R,\cdot)$ is a probability density, and the second because $\nu_\mathtt{FS}$ is invariant. Hence, $c_-(t_0)\leq1\leq c_+(t_0)$, so that $\varepsilon(t_0)\geq0$, while $\varepsilon(t_0)<1$ since $c_-(t_0)>0$. Subtracting a constant $q$ then gives $\int(p_t-q)\,\di\nu_\mathtt{FS}=1-q$ in either argument, and if $p_t-q$ has one sign, the absolute value may be taken inside the integral, i.e., $\int|p_t-q|\,\di\nu_\mathtt{FS}=|1-q|$.

Let us fix such a $t$ and $q$. For $w\in\mathbf{1}^\perp$, we have $\mathscr{T}_t w(Q)=\int(p_t(R,Q)-q)w(R)\,\nu_\mathtt{FS}(\di R)$. Applying Cauchy--Schwarz in $R$, followed by integration in $Q$ then gives
\[
  \int |\mathscr{T}_t w|^2\di\nu_\mathtt{FS}
  \leq |1-q|\iint\bigl|p_t(R,Q)-q\bigr|\,|w(R)|^2\,\nu_\mathtt{FS}(\di R)\,\nu_\mathtt{FS}(\di Q)
  = |1-q|^2\int|w|^2\di\nu_\mathtt{FS} .
\]
Taking $t=t_0$ and $q=c_-(t_0)$ or $q=c_+(t_0)$, both admissible due to \eqref{eq:cmonotone}, gives the factors $1-c_-(t_0)$ and $c_+(t_0)-1$. Consequently,
\[
  \|\mathscr{T}_{t_0}w\|_{L^2(\nu_\mathtt{FS})}
  \leq \varepsilon(t_0) \|w\|_{L^2(\nu_\mathtt{FS})},
  \qquad w\in\mathbf{1}^\perp .
\]
Since $\mathbf{1}^\perp$ is invariant, we can iterate this. Given $t\geq0$, write $t=nt_0+\tau$ with $n\geq0$ and $\tau\in[0,t_0)$. Applying it $n$ times, with the contraction covering the remaining interval, we obtain $\|\mathscr{T}_tw\|_{L^2(\nu_\mathtt{FS})}\leq\varepsilon(t_0)^n\|w\|_{L^2(\nu_\mathtt{FS})}$. For $\varepsilon(t_0)>0$, we have $\varepsilon(t_0)^n=e^{-\lambda(t-\tau)}\leq e^{\lambda t_0}e^{-\lambda t}$ with $\lambda$ as given above and $e^{\lambda t_0}=M$, which is \eqref{eq:gapdef}.

If instead $\varepsilon(t_0)=0$, the same iteration reads $\mathscr{T}_tw=0$ for every $t\geq t_0$ and every $w\in\mathbf{1}^\perp$, so that \eqref{eq:gapdef} holds with $M=e^{\lambda t_0}$ for every $\lambda>0$, the contraction covering $t<t_0$. This is the case when the process is already in equilibrium at time $t_0$.
\end{proof}

\begin{remark}[The spectral pair]\label{rem:spectralpair}
The pair obtained in Proposition~\ref{prop:gappair} is far from optimal---the optimal rate is, in fact, the spectral gap of $\scrL$ on $\mathbf{1}^\perp$:
\[
  \lambda_{\scrL}:=-\sup\bigl\{\Re z:\ z\in\spec\bigl(\scrL|_{\mathbf{1}^\perp}\bigr)\bigr\}.
\]
The restriction to $\mathbf{1}^\perp$ makes this positive, since the constant is an eigenfunction of $\scrL$ with eigenvalue $0$. Indeed, $\mathscr{T}_t$ is an integral operator on $L^2(\nu_\mathtt{FS})$ whose kernel $p_t$ is bounded by Lemma~\ref{lem:kernelmono}. Since $\nu_\mathtt{FS}$ is finite, the kernel is then square integrable, so $\mathscr{T}_t$ is Hilbert--Schmidt and hence compact for every $t\geq t_\star$, i.e., the semigroup is \emph{eventually compact}. For such a semigroup, the spectral bound equals the growth bound \cite[Cor.~IV.3.12]{engelnagel}, so it decays at precisely the rate its spectrum predicts. That rate is then $\lambda_{\scrL}$, since $\spec(\scrL^*)=\overline{\spec(\scrL)}$, having the same real parts. Hence, \eqref{eq:gapdef} holds with a finite $M$ for every $\lambda<\lambda_{\scrL}$, and at $\lambda_{\scrL}$ itself whenever the supremum is attained by semisimple eigenvalues \cite[Cor.~V.3.2]{engelnagel}.
\end{remark}

We now state the main result of this section.

\begin{theorem}[Decay from the kernel bound]\label{thm:sharpchi}
Assume Property~\ref{ass:kernel}. Let $t_0\geq t_\star$, and let $f$ be a non-affine normalized operator convex function. Further, let $c_\pm$ be the kernel bounds of \eqref{eq:cmonotone}, and let $\mathscr{T}_t$ be of type $(M,-\lambda)$ as in \eqref{eq:gapdef}. Then for all $\rho_0,\sigma_0\in\scrS_+(\frH)$ and $t\geq t_0$,
\begin{equation}\label{eq:sharpchi}
  \fdiv^{\max}(\rho_t\|\sigma_t) \leq \exp\Bigl(-2\lambda\max\bigl\{0,\,t-s_f(t)\bigr\}\Bigr)\,\fdiv^{\max}(\rho_0\|\sigma_0),
\end{equation}
with the time delay
\[
	s_f(t) := t_0+\frac{1}{2\lambda}\log \biggl(M^2\frac{g_f(1) c_+(t_0)}{g_f(\theta_{\max}) c_-(t)}\biggr),\qquad \theta_{\max}:=e^{\dinf^{\max}(\rho_0\|\sigma_0)},
\]
where $g_f$ is as given in Proposition~\ref{prop:opconv}.
\end{theorem}
\begin{proof}
Let $\mu_t,\nu_t$ solve \eqref{eq:fokker-planck} with initial data $\mu_{0,\CB},\nu_{0,\CB}$ supported on a common basis of $(\rho_0,\sigma_0)$. These initial measures are atomic, but for every $t\geq t_\star$ the propagated ones have densities with respect to $\nu_\mathtt{FS}$ by Lemma~\ref{lem:kernelmono}. We write $u_t,v_t$ for these densities, with the difference lying in $\mathbf{1}^\perp$ for every $t\ge t_\star$. As the two are transported by the same semigroup, $u_t-v_t=\mathscr{T}_{t-t_0}(u_{t_0}-v_{t_0})$ for $t\geq t_0$, so \eqref{eq:gapdef} gives
\[
  \|u_t-v_t\|_{L^2(\nu_\mathtt{FS})} \leq M e^{-\lambda(t-t_0)}\|u_{t_0}-v_{t_0}\|_{L^2(\nu_\mathtt{FS})},
  \qquad \text{for all $t\geq t_0$}.
\]

For the left-hand side of \eqref{eq:sharpchi}, both densities are bounded below. Indeed, since $\mu_{0,\CB}$ and $\nu_{0,\CB}$ are probability measures, $u_t=\int p_t(P,\cdot)\,\mu_{0,\CB}(\di P)\geq c_-(t)$ and likewise $v_t\geq c_-(t)$. Therefore, evaluating \eqref{eq:opconvrep} at $x/y$ and multiplying by $y$ gives, for all $x,y\geq c_-(t)$,
\begin{align*}
  y\,f(x/y) &= b\,\frac{(x-y)^2}{y}+\int_{[0,\infty)}\frac{(x-y)^2}{x+sy}\,\tau_f(\di s) \\
  &\leq \frac{(x-y)^2}{c_-(t)}\biggl(b+\int_{[0,\infty)}\frac{1}{1+s}\,\tau_f(\di s)\biggr) = \frac{g_f(1)}{c_-(t)}\,(x-y)^2,
\end{align*}
where we used $y \geq c_-(t)$ in the first term and $x+sy \geq c_-(t)(1+s)$ in the second. Applying this estimate pointwise with $x=u_t$ and $y=v_t$ and integrating, we obtain
\[
  \fdiv(\mu_t\|\nu_t) = \int_{\scrP(\frH)}v_t\,f\Bigl(\frac{u_t}{v_t}\Bigr)\,\di\nu_\mathtt{FS}
  \leq \frac{g_f(1)}{c_-(t)}\,\|u_t-v_t\|_{L^2(\nu_\mathtt{FS})}^2 .
\]

For the right-hand side, we write $p_i:=p_{t_0}(P_i,\cdot)$ for the transition densities starting from the atoms carrying $\mu_{0,\CB}$ and $\nu_{0,\CB}$, with weights $(\rho_i)_i$ and $(\sigma_i)_i$, so that $u_{t_0}=\sum_i\rho_ip_i$ and $v_{t_0}=\sum_i\sigma_ip_i$. Splitting $(\rho_i-\sigma_i)p_i=\sqrt{\sigma_ip_i}\cdot\sigma_i^{-1/2}(\rho_i-\sigma_i)\sqrt{p_i}$, and applying Cauchy-Schwarz pointwise yields
\[
  (u_{t_0}-v_{t_0})^2 \leq v_{t_0}
  \biggl(\sum\nolimits_i\frac{(\rho_i-\sigma_i)^2}{\sigma_i}\,p_i\biggr)
  \leq c_+(t_0)\sum\nolimits_i\frac{(\rho_i-\sigma_i)^2}{\sigma_i}\,p_i,
\]
where we used $v_{t_0}\leq c_+(t_0)$. Integrating with respect to $\nu_\mathtt{FS}$ and setting $\theta_i:=\rho_i/\sigma_i$ gives
\[
  \|u_{t_0}-v_{t_0}\|^2_{L^2(\nu_\mathtt{FS})}
  \leq c_+(t_0)\sum_{i=1}^n\frac{(\rho_i-\sigma_i)^2}{\sigma_i}
  = c_+(t_0)\sum_{i=1}^n\sigma_i(\theta_i-1)^2 .
\]
By \eqref{eq:spectral-geometric} the $\theta_i$'s are the eigenvalues of $\sigma_0^{-1/2}\rho_0\sigma_0^{-1/2}$, whose largest is $\theta_{\max}$ by Definition~\ref{def:dmax}. Since $f(x)=(x-1)^2g_f(x)$ with $g_f$ non-increasing, we have that $(\theta_i-1)^2\leq f(\theta_i)/g_f(\theta_{\max})$. Hence,
\[
  \sum_{i=1}^n\sigma_i(\theta_i-1)^2
  \leq \frac{1}{g_f(\theta_{\max})}\sum_{i=1}^n\sigma_i f(\theta_i)
  = \frac{\fdiv^{\max}(\rho_0\|\sigma_0)}{g_f(\theta_{\max})} ,
\]
where the last step holds due to \eqref{eq:starallf}. 

Combining the three estimates above and collecting the constants into $e^{2\lambda s_f(t)}$ gives
\begin{equation}\label{eq:sharpchimeasure}
  \fdiv(\mu_t\|\nu_t) \leq e^{-2\lambda(t-s_f(t))}\fdiv^{\max}(\rho_0\|\sigma_0),
  \qquad\text{for all $t\geq t_0$},
\end{equation}
Since $(\mu_t,\nu_t)$ realizes $(\rho_t,\sigma_t)$, Theorem~\ref{thm:maintheorem} implies $\fdiv^{\max}(\rho_t\|\sigma_t)\leq\fdiv(\mu_t\|\nu_t)$, and \eqref{eq:sharpchimeasure} therefore yields \eqref{eq:sharpchi} on the range where $t\geq s_f(t)$. On the rest, where the exponent is $0$, it follows from the quantum DPI in Proposition~\ref{prop:quantumdpi}.
\end{proof}

For the Belavkin--Staszewski entropy, comparison with $\chi^2$ gives a decay estimate of logarithmic form, at the same rate as in \eqref{eq:sharpchi}. In the following, we write $g_\mathtt{KL}$ for $g_f$ at $f_\mathtt{KL}$.

\begin{corollary}\label{cor:sharpbs}
Under the hypotheses of Theorem~\ref{thm:sharpchi}, one has for $t\geq t_0$,
\begin{equation}\label{eq:sharpbs}
  \BS(\rho_t\|\sigma_t) \leq \min\Bigl\{\BS(\rho_0\|\sigma_0),\
  \log\bigl(1+e^{-2\lambda(t-s_\mathtt{BS}(t))}\BS(\rho_0\|\sigma_0)\bigr)\Bigr\},
\end{equation}
with the time delay
\[
	s_\mathtt{BS}(t):=t_0+\frac{1}{2\lambda}\log \biggl(M^2\frac{c_+(t_0)}{g_\mathtt{KL}(\theta_{\max}) c_-(t)}\biggr).
\]
Moreover, $\BS(\rho_t\|\sigma_t)\leq\exp\bigl(-2\lambda\max\{0,\,t-s_\mathtt{BS}(t)+\log 2/(2\lambda)\}\bigr)\,\BS(\rho_0\|\sigma_0)$.
\end{corollary}
\begin{proof}
As before, the first argument of the minimum follows directly from the quantum DPI in Proposition~\ref{prop:quantumdpi}, so we only need to prove the second. Let $\mu_t,\nu_t$ be as in the proof of Theorem~\ref{thm:sharpchi} and write $u=\di\mu_t/\di\nu_t$. Since $\log$ is concave and $\mu_t$ is a probability measure,
Jensen's inequality applied with respect to $\mu_t$ gives
\[
  \KL(\mu_t\|\nu_t) = \int\log u\,\di\mu_t \leq \log\int u\,\di\mu_t
  =\log\int u^2\,\di\nu_t=\log\bigl(1+\chi^2(\mu_t\|\nu_t)\bigr),
\]
where the last step holds because $\int(u-1)^2\di\nu_t=\int u^2\di\nu_t-1$. As $(\mu_t,\nu_t)$ realizes $(\rho_t,\sigma_t)$, Theorem~\ref{thm:maintheorem} gives $\BS(\rho_t\|\sigma_t)\leq\KL(\mu_t\|\nu_t)$. Hence, it suffices to bound $\chi^2(\mu_t\|\nu_t)$. This is the content of \eqref{eq:sharpchimeasure} at $f(x)=(x-1)^2$, where $g_f\equiv1$, and which gives
\[
  \chi^2(\mu_t\|\nu_t)\leq e^{-2\lambda(t-s_{\chi^2}(t))}\,D^{\max}_{\chi^2}(\rho_0\|\sigma_0),
  \qquad s_{\chi^2}(t)=t_0+\frac{1}{2\lambda}\log\biggl(M^2\frac{c_+(t_0)}{c_-(t)}\biggr).
\]
Note that it is the estimate for the measures that is needed here, since $D^{\max}_{\chi^2}(\rho_t\|\sigma_t)$, which \eqref{eq:sharpchi} bounds, lies below $\chi^2(\mu_t\|\nu_t)$.

We now bound the initial datum. On the common basis, the ratios are $\theta_i=\rho_i/\sigma_i\le \theta_{\max}$. Recalling $f_\mathtt{KL}(\theta)=\theta\log\theta-\theta+1\geq0$ and using $\sum_i\sigma_i(\theta_i-1)=0$, we write
\[
  \BS(\rho_0\|\sigma_0)=\sum_{i=1}^n\sigma_i\,\theta_i\log\theta_i=\sum_{i=1}^n\sigma_i f_\mathtt{KL}(\theta_i),
  \qquad
  D^{\max}_{\chi^2}(\rho_0\|\sigma_0)=\sum_{i=1}^n\sigma_i\,(\theta_i-1)^2 .
\]
We then obtain $D^{\max}_{\chi^2}(\rho_0\|\sigma_0)\leq\BS(\rho_0\|\sigma_0)/g_\mathtt{KL}(\theta_{\max})$ by following the same arguments as in the proof of Theorem~\ref{thm:sharpchi}. Since $s_\mathtt{BS}(t)=s_{\chi^2}(t)+\log\bigl(1/g_\mathtt{KL}(\theta_{\max})\bigr)/2\lambda$, the two estimates combine to give $\chi^2(\mu_t\|\nu_t)\leq e^{-2\lambda(t-s_\mathtt{BS}(t))}\BS(\rho_0\|\sigma_0)$, which yields \eqref{eq:sharpbs}. 

Finally, the linear bound does not need the comparison with $\chi^2$. Since $f_\mathtt{KL}$ is itself a non-affine normalized operator convex function, Theorem~\ref{thm:sharpchi} applies to it directly, and since $g_\mathtt{KL}(1)=\tfrac12$ by Remark~\ref{ex:opconv}, its delay is $s_{f_\mathtt{KL}}(t)=s_\mathtt{BS}(t)-\log 2/(2\lambda)$. This is shorter than the delay $s_\mathtt{BS}(t)$ obtained by linearizing \eqref{eq:sharpbs} with $\log(1+x)\leq x$, which the logarithmic form improves on only where $e^{-2\lambda(t-s_\mathtt{BS}(t))}\BS(\rho_0\|\sigma_0)$ is large.
\end{proof}

\subsection{Examples}
\label{sec:examples}

Both examples we consider below have $n=2$, allowing us to transfer to a diffusion on the sphere $\mathbb{S}^2$. We write
\[
  r(P):=\bigl(\tr[\sigma_xP],\,\tr[\sigma_yP],\,\tr[\sigma_zP]\bigr)^\top\in\mathbb{S}^2,
\]
for the Bloch vector corresponding to $P\in\scrP(\frH)$. It has unit length on pure states, and $P\mapsto r(P)$ is a diffeomorphism of $\scrP(\frH)$ onto $\mathbb{S}^2$, which we use throughout to identify the two, i.e., a function on $\scrP(\frH)$ is considered as a function on $\mathbb{S}^2$ without change of name, and $\nu_\mathtt{FS}$ becomes the uniform probability measure on $\mathbb{S}^2$.

Indices $j,k\in\{x,y,z\}$ from here on, and let $e_x,e_y,e_z$ denote the standard basis of $\R^3$. Let $\mathcal{R}_j:=(e_j\times r)\cdot\nabla$ be the field generating rotations about $e_j$, so that $\Delta_{\mathbb{S}^2}=\mathcal{R}_x^2+\mathcal{R}_y^2+\mathcal{R}_z^2$. Under the
identification, the jump field \eqref{eq:jumpfields} of $S_j=\sigma_j$ gives $V_j=-2\mathcal{R}_j$, and the drift $V_0=-\iu[H,\cdot\,]$ of $H=\sum_k\omega_k\sigma_k$ gives $2\sum_k\omega_k\mathcal{R}_k$. For jump operators drawn from the Pauli matrices, with rates $\gamma_j$, the generator and its carr\'{e} du champ are therefore
\begin{equation}\label{eq:blochgen}
  \scrL = 2\sum\nolimits_k\omega_k\mathcal{R}_k+2\sum\nolimits_j\gamma_j^2\mathcal{R}_j^2 ,
  \qquad
  \Gamma(F,F)=2\sum\nolimits_j\gamma_j^2\bigl(\mathcal{R}_jF\bigr)^2 .
\end{equation}
We write $\mathcal{Y}_\ell$ for the spherical harmonics of degree $\ell$ with $\dim \mathcal{Y}_\ell = 2\ell +1$. They are the eigenspaces of $\Delta_{\mathbb{S}^2}$ for the eigenvalues $-\ell(\ell+1)$ respectively. Since distinct eigenspaces are orthogonal and $\mathcal{Y}_0$ consists of the constants, we have that $\mathcal{Y}_\ell\subset\mathbf{1}^\perp$ for $\ell\geq1$. Moreover, each $\mathcal{R}_j$ preserves $\mathcal{Y}_\ell$, so does $\scrL$, and its spectrum can then be read off block by block. The block $\ell=1$ is the Lindblad generator of \eqref{eq:lindblad} written in terms of Bloch vectors.

\subsubsection*{The depolarizing qubit}

Take $S_j=\sigma_j$ with $\gamma_j=\gamma$ for every $j$, and $H=0$. Then \eqref{eq:blochgen}
reads
\begin{equation}
  \scrL = 2\gamma^2\Delta_{\mathbb{S}^2},\qquad
  \Gamma(F,F) = 2\gamma^2|\nabla_{\mathbb{S}^2}F|^2,\qquad
  I_\mathtt{KL}(\mu\|\nu) = 4\gamma^2\int|\nabla_{\mathbb{S}^2}\sqrt{u}|^2\,\di\nu,
\end{equation}
which is the diffusion process of the Bloch vector given by
\begin{equation}\label{eq:blochsde}
  \di r = -4\gamma^2 r\,\di t - 2\gamma\sum\nolimits_j (e_j\times r)\,\di W_j ,
\end{equation}
with diffusion tensor $4\gamma^2(\mathbbm{1}-rr^\top)$, and the radial drift is the It\^{o} correction that keeps $|r|=1$.

Since $\mathfrak{g}$ consists of the fields of $\mathrm{Lie}_\R\{\iu\sigma_x,\iu\sigma_y,\iu\sigma_z\}=\mathfrak{su}(2)$, it is bracket generating in the sense of Definition~\ref{def:bracket}, and, therefore, Proposition~\ref{prop:bracket} applies. The first non-zero eigenvalue of $-\Delta_{\mathbb{S}^2}$ is $2$, so Proposition~\ref{prop:bracket} holds with $\lambda_{\mathtt{sym}}=4\gamma^2$. By compactness of $\mathbb{S}^2$ and positivity of the heat kernel, the kernel bound holds here for every $t_\star>0$, which gives Property~\ref{ass:kernel}. The rate is consistent with the exact relaxation $\rho_t-\tfrac{\mathbbm{1}}{2}=e^{-4\gamma^2t}(\rho_0-\tfrac{\mathbbm{1}}{2})$, and Theorem~\ref{thm:decay} gives
\begin{equation}
  D^{\max}_{\chi^2}(\rho_t\|\sigma_t) =\Tr[\rho_t\sigma_t^{-1}\rho_t]-1 \leq e^{-8\gamma^2(t-t_f)} \bigl(\Tr[\rho_0\sigma_0^{-1}\rho_0]-1\bigr) = e^{-8\gamma^2(t-t_f)}D^{\max}_{\chi^2}(\rho_0\|\sigma_0).
\end{equation}
The rate $8\gamma^2$ is exactly $2\lambda_{\mathtt{sym}}$, and it is sharp.

The delay is where the estimate falls short. On the spherical harmonics $\mathcal{Y}_\ell$, the generator $2\gamma^2\Delta_{\mathbb{S}^2}$ has eigenvalue $-2\gamma^2\ell(\ell+1)$, and the corresponding spectral projection has kernel bounded by its diagonal value $2\ell+1$, so that
\[
  \sup\nolimits_{P,Q}\,\bigl|p_{t_0}(P,Q)-1\bigr| \leq \sum\nolimits_{\ell\geq1}(2\ell+1)e^{-2\gamma^2\ell(\ell+1)t_0} .
\]
Since $2\ell+1\leq3\ell$ and $\ell(\ell+1)\geq2\ell$ for $\ell\geq1$, this is bounded by a geometric series,
\[
  \sup\nolimits_{P,Q}\,\bigl|p_{t_0}(P,Q)-1\bigr| \leq 3\sum\nolimits_{\ell\geq1}\ell\,q^{\,\ell} = \frac{3q}{(1-q)^{2}}=:\varepsilon_+(t_0),
  \qquad q:=e^{-4\gamma^2t_0}.
\]
For $\gamma=1$ and $t_0=1$, this gives $\varepsilon_+(1)\approx 0.057$. We may then take $c_\pm(t_0)=1\pm\varepsilon_+(t_0)$, such that the constant $\varepsilon(t_0)$ of Proposition~\ref{prop:gappair} is $\varepsilon_+(t_0)$ itself. With this choice, Lemma~\ref{lem:projective}$(i)$ gives $\beta(t_0)\leq0.23$ and $\varkappa\leq0.058$. For $f(x)=(x-1)^2$, one has $\xi_f(\beta)=e^{-\beta}$, so the bound $t_f\leq t_0\bigl(1+L_\xi\beta(t_0)/(1-\varkappa)\bigr)$ of Theorem~\ref{thm:decay} applies with $L_\xi=1$ and gives $t_f\leq1.25$, i.e., the delay is a little over $t_0$ itself. See \cite{temmechi2} for $\chi^2$ mixing-time estimates of the same kind.

\subsubsection*{Dephasing with a transverse field}

Theorem~\ref{thm:sharpchi} covers systems that Theorem~\ref{thm:decay} cannot reach. Consider the single jump operator $S_1=\sigma_x$ at rate $\gamma>0$ and the Hamiltonian $H=\omega\sigma_z$. Then, \eqref{eq:blochgen} reads
\[
  \scrL=2\gamma^2\mathcal{R}_x^2+2\omega\mathcal{R}_z,\qquad
  \Gamma(F,F)=2\gamma^2\bigl(\mathcal{R}_xF\bigr)^2 .
\]
Since $\mathfrak{g}$ is the one-dimensional algebra of fields of $\R\,\iu\sigma_x$, whose orbits are the circles about the $x$-axis, every function $F$ depending only on $r_x$ gives $\Gamma(F,F)=0$, i.e., the Poincar\'e inequality fails. However, the drift $H$ restores non-degeneracy. Indeed, the rotation fields obey the $\mathfrak{so}(3)$ relations $[\mathcal{R}_x,\mathcal{R}_y]=-\mathcal{R}_z$ and its cyclic permutations. Bracketing the jump field $\mathcal{R}_x$ with the drift field $\mathcal{R}_z$ therefore returns $[\mathcal{R}_z,\mathcal{R}_x]=-\mathcal{R}_y$, and bracketing that with $\mathcal{R}_x$ again returns $[\mathcal{R}_x,-\mathcal{R}_y]=\mathcal{R}_z$, so all three rotations are reached. In particular, the Lie algebra generated by $\{V_1,V_0-\partial_t\}$ spans $T_{(t,r)}(\R{\times}\mathbb{S}^2)$ at every point---the rotations span $T_r\mathbb{S}^2$ and $V_0-\partial_t$ gives the time direction---which is H\"{o}rmander's condition for the space-time operator $\partial_t-\scrL$ making it hypoelliptic (cf.\ Remark~\ref{rem:hormander}). The kernel is smooth for $t>0$, so Remark~\ref{rem:kernelhyp} gives the upper bound in Property~\ref{ass:kernel}. Although the lower bound is not proved here, numerical results show that $\inf_{P,Q}p_{t_0}(P,Q)>0$ for $t_0\gtrsim 7$ (cf.\ Figure~\ref{fig:dephasingdelay}(B)).

On the block $\ell=1$ the spectrum of $\scrL$ is
$\{-2\gamma^2\}\cup\{-\gamma^2\pm\sqrt{\gamma^4-4\omega^2}\}$, giving
\begin{equation}\label{eq:degengap}
  \lambda_{\scrL}=\gamma^2-\sqrt{\max\{0,\gamma^4-4\omega^2\}}
  =\begin{cases}
    2\omega^2/\gamma^2+O(\omega^4), & \omega\ll\gamma^2,\\
    \gamma^2, & |\omega|\geq\gamma^2/2 .
  \end{cases}
\end{equation}
The minimum over $\ell$ is attained at $\ell=1$, since the gaps of the higher blocks increase like $\gamma\sqrt{\omega\ell}$ as shown in Figure~\ref{fig:dephasingspectrum}(A). Notice that the rate vanishes quadratically as $\omega\to0$, which is the failure of bracket generation seen quantitatively. On the other hand, the rate saturates at $\gamma^2$ once $|\omega|\geq\gamma^2/2$, so beyond that point a stronger field gives no further improvement; Figure~\ref{fig:dephasingspectrum}(B) plots \eqref{eq:degengap} against the gap computed block by block. For comparison, the depolarizing qubit above has $\lambda_{\scrL}=4\gamma^2$.

As with $\lambda_{\mathtt{sym}}$ in the previous example, this value is obtained from the spectrum rather than from Proposition~\ref{prop:gappair}. Unlike the previous example, we have no closed form for the constants $c_\pm(t_0)$ here---the values used in Figure~\ref{fig:dephasingdelay}(B) are computed numerically from the transition kernel, as $c_\pm(t_0)=1\pm\sup_{P,Q}|p_{t_0}(P,Q)-1|$. Theorem~\ref{thm:sharpchi} then gives, for every non-affine normalized operator convex $f$ and every $t\geq t_0$,
\[
  \fdiv^{\max}(\rho_t\|\sigma_t) \leq \exp\bigl(-2\lambda\max\{0,\,t-s_f(t)\}\bigr)\,\fdiv^{\max}(\rho_0\|\sigma_0),
\]
for every $\lambda<\lambda_{\scrL}$, and at $\lambda=\lambda_{\scrL}$ itself whenever the eigenvalues attaining it are semisimple. The rate is therefore the one given by \eqref{eq:degengap}, which, by Remark~\ref{rem:ellone}, is the Lindblad relaxation rate. Figure~\ref{fig:dephasingdelay}(A) shows that this rate is indeed attained, with the exact decay running parallel to $2\lambda_{\scrL}$ and losing only by an offset.

\begin{figure}[t]
    \centering
    \begin{subfigure}{0.4\textwidth}
    \includegraphics[width=\textwidth]{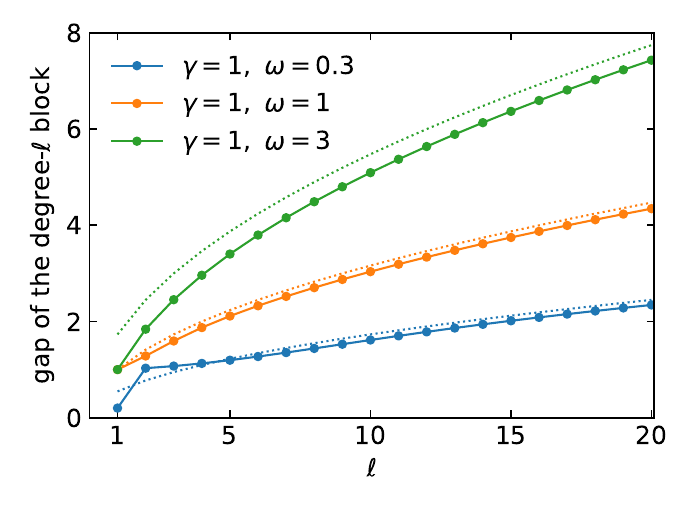}
    \vspace{-2em}
    \caption{}
    \end{subfigure}
    \begin{subfigure}{0.4\textwidth}
	\includegraphics[width=\textwidth]{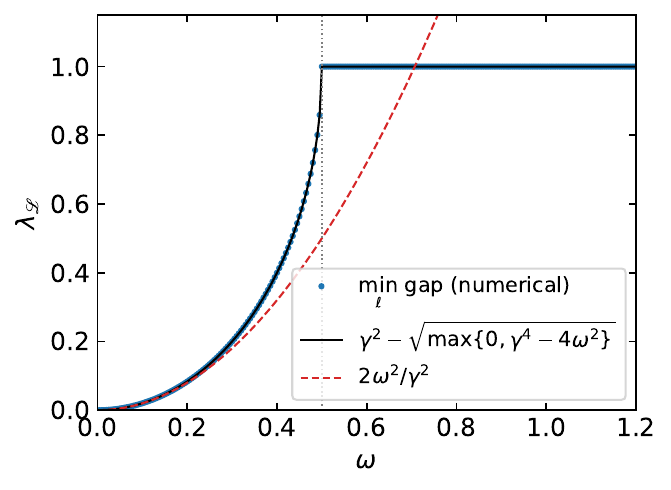}
	    \vspace{-2em}
	\caption{}
    \end{subfigure}
    \caption{\emph{The spectrum of $\scrL$ for transverse field.} (A) The gap on the harmonics of degree $\ell$. The minimum is attained at $\ell=1$, and the higher blocks grow like $\gamma\sqrt{\omega\ell}$ (dotted). (B) The resulting $\lambda_{\scrL}$ against $\omega$ at $\gamma=1$, computed block by block, against the closed form \eqref{eq:degengap} and its small-$\omega$ behavior $2\omega^2/\gamma^2$. Past the dotted line $\omega=\gamma^2/2$, the rate saturates at $\gamma^2$.}
    \label{fig:dephasingspectrum}
\end{figure}

\begin{figure}[t]
    \centering
    \begin{subfigure}{0.4\textwidth}
    \includegraphics[width=\textwidth]{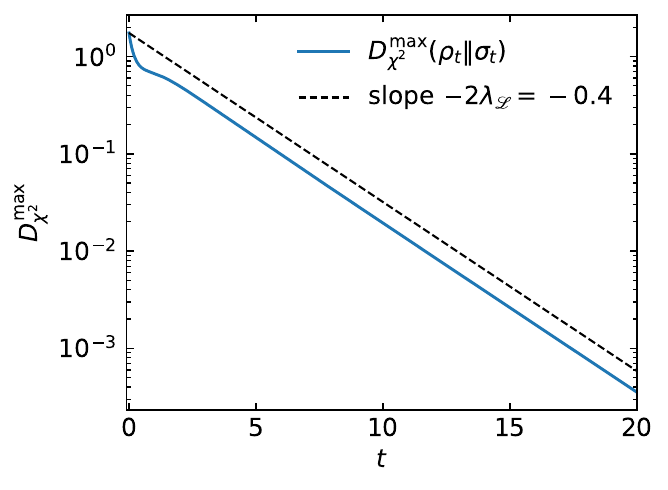}
    \vspace{-2em}
    \caption{}
    \end{subfigure}
    \begin{subfigure}{0.4\textwidth}
    \includegraphics[width=\textwidth]{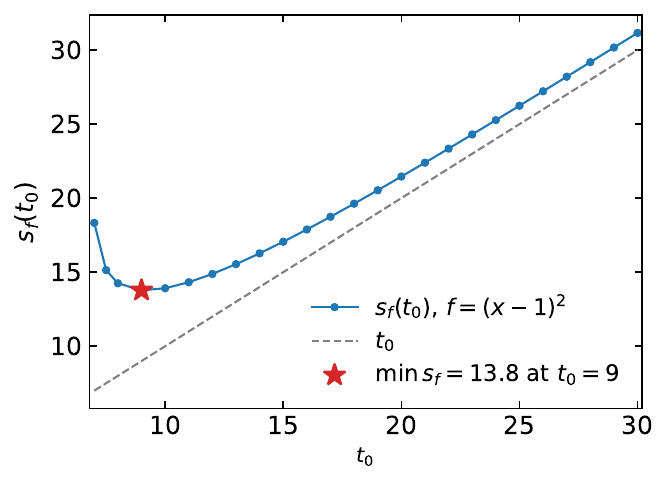}
    \vspace{-2em}
    \caption{}
    \end{subfigure}
    \caption{\emph{Decay and delay at $\gamma=1$, $\omega=0.3$, so that $\lambda_{\scrL}=0.2$.} (A) The values of $D^{\max}_{\chi^2}(\rho_t\|\sigma_t)$, where $\rho_t$ and $\sigma_t$ are computed exactly with initial data $r(\rho_0)=(0.55,0.10,0.30)$ and $r(\sigma_0)=(-0.20,0.45,-0.35)$, against $2\lambda_{\scrL}$. The two are asymptotically parallel, so the rate is attained, and what the estimate loses is only the offset. The bend near $t=1$ comes from the two faster Bloch modes, which decay at $1.8$ and $2$ versus $0.2$ for the slow one. 
    (B) The delay $s_f(t)$ of Theorem~\ref{thm:sharpchi}, evaluated at $t=t_0$, for $f=f_\mathtt{2}$ and $\lambda=\lambda_{\scrL}$, computed from the kernel bounds $c_\pm(t_0)=1\pm\sup_{P,Q}|p_{t_0}(P,Q)-1|$ with $M=1.25$, computed numerically. Its minimum is $13.8$, at $t_0=9$, against $\lambda_{\scrL}^{-1}=5$. The curve begins at $t_0\approx7$, and $c_-(t_0)$ is not positive before then.
	}
    \label{fig:dephasingdelay}
\end{figure}

We note that the delay cannot be short here. Indeed, let $w\in\mathcal{Y}_1$ be an eigenfunction of $\scrL^*$ for an eigenvalue with real part $-\lambda_{\scrL}$. Since $\mathcal{Y}_1\subset \mathbf{1}^\perp$, we can express $\mathscr{T}_tw(Q)=\int\bigl(p_t(R,Q)-1\bigr)w(R)\,\nu_\mathtt{FS}(\di R)$. Using Cauchy--Schwarz, this yields
\[
  e^{-\lambda_{\scrL}t}\|w\|_{L^2(\nu_\mathtt{FS})}
  =\bigl\|\mathscr{T}_tw\bigr\|_{L^2(\nu_\mathtt{FS})}
  \leq\Bigl(\sup\nolimits_{P,Q}\bigl|p_t(P,Q)-1\bigr|\Bigr)\|w\|_{L^2(\nu_\mathtt{FS})} .
\]
Since $p_{t_0}(P,\cdot)$ integrates to $1$, its infimum is at most $1$ and its supremum at least $1$. Hence,
\[
  \sup\nolimits_{P,Q}\bigl|p_{t_0}(P,Q)-1\bigr|
  =\max\bigl\{\sup\nolimits_{P,Q}p_{t_0}(P,Q)-1,\ 1-\inf\nolimits_{P,Q}p_{t_0}(P,Q)\bigr\} .
\]
Moreover, since constants admissible in \eqref{eq:kernelbounds} satisfy $c_-(t_0)\leq\inf p_{t_0}$ and $c_+(t_0)\geq\sup p_{t_0}$, each entry of that maximum is at most the corresponding entry of $\max\{c_+(t_0)-1,\,1-c_-(t_0)\}$, which is therefore at least $e^{-\lambda_{\scrL}t_0}$. So either $c_+(t_0)$ or $c_-(t_0)$ differs from $1$ by $e^{-\lambda_{\scrL}t_0}$ or more, and the two cannot both be within $\delta$ of $1$ unless $t_0\geq\lambda_{\scrL}^{-1}\log(1/\delta)$. Since $s_f(t)\geq t_0$ in Theorem~\ref{thm:sharpchi}, the delay is at least that, regardless of how $t_0$ is chosen. Figure~\ref{fig:dephasingdelay}(B) depicts this numerically---the delay is minimised at $t_0=9$, where it is still $13.8$, against a relaxation time $\lambda_{\scrL}^{-1}=5$.

\begin{remark}[The states see only $\ell=1$]\label{rem:ellone}
Since a qubit state is determined by its Bloch vector, and the block $\ell=1$ of $\scrL$ is the Lindblad generator written there, the decay of $\fdiv^{\max}(\rho_t\|\sigma_t)$ is then governed entirely by that block, while the higher blocks affect only the measures. In this example, the minimum over $\ell$ is attained at $\ell=1$, so the rate of Theorem~\ref{thm:sharpchi} is exactly the Lindblad relaxation rate, and the passage through the unraveling incurs no loss.
\end{remark}

\section{Conclusion}
\label{sec:discussion}
In this work, we showed that stochastic unravelings of the Lindblad equation can be used to strengthen the quantum data processing inequality for maximal $f$-divergences, and to turn that strengthening into decay rates for a class of Lindblad evolutions, by transporting the problem to the space of pure states, where the analysis is classical. Two results make the transport possible. First, defining an $f$-divergence as an infimum of classical $f$-divergences, we showed that the infimum is attained on a common basis and equals the maximal $f$-divergence---a bound on measures therefore descends to a bound on $\fdiv^{\max}$. Although the common basis coincides with the reverse test constructed in \cite{explicitmax}, we obtain here that the infimum is attained over measures on a continuum. 

Second, this minimizer consists of pure states and is thus admissible initial data for the stochastic Schr\"{o}dinger equation, so that it may be propagated along the associated Fokker--Planck equation. Doing so gives the strengthened inequality of Theorem~\ref{thm:contraction}, where the loss of $\fdiv^{\max}$ becomes explicit via the relative $f$-Fisher information of the unraveling for every normalized operator convex $f$ with $f''>0$ and every pair of faithful states. This yields two decay estimates that differ in their assumptions and the rates they reach. The first combines the dissipation with a $\Phi$-entropy inequality for the unraveling and reaches twice the spectral gap of the symmetrized generator after an explicit delay, at the cost of a bracket-generating condition on the jump fields. The second relies on kernel bounds alone and gives any rate below twice the spectral gap of the full generator, so a Hamiltonian may supply a rate that the jump operators alone do not.

What our approach costs is its intrinsicness. The dissipation in Theorem~\ref{thm:contraction} is a functional of the propagated measures rather than of $\rho_t$ and $\sigma_t$, and the delays derived for both decay estimates are functionals of the transition kernel of the unraveling. It is unclear whether a remainder expressible in the states alone is available for this class of evolutions.

Several questions remain open. First, Section~\ref{sec:dissipation} assumes the two-sided kernel bound of Property~\ref{ass:kernel} rather than proving it. Second, the rate in Theorem~\ref{thm:decay} reflects the spectral gap of the symmetrized generator and is sharp for the depolarizing qubit. However, its delay $t_f$ is not optimal, since it relies on kernel bounds via Birkhoff's contraction, and the estimate for $[0,t_0]$ neglects dissipation entirely. The same holds for the delay $s_f$ in Theorem~\ref{thm:sharpchi}. A direct argument avoiding both would be interesting. Moreover, recall that we assume the jump operators are Hermitian, which allows classical-noise unraveling. A general Lindbladian would require a different unraveling, and we do not pursue that here. The Hilbert space is likewise finite-dimensional, and we use the compactness of $\scrP(\frH)$ at nearly every step. We see no structural obstruction there, but removing it would require tools from operator theory and infinite-dimensional analysis that this work does not develop.

\section*{ACKNOWLEDGEMENTS}
We thank Raul dos Santos, Pim Kersbergen, Jasper Postema, Jasper van de Kraats and Nil Dinç for fruitful discussions. This research is financially supported by the Dutch Ministry of Economic Affairs and Climate Policy (EZK), as part of the QDNL program,  the Horizon Europe programme HORIZON-CL4-2021-DIGITAL-EMERGING-01-30 via the project 101070144 (EuRyQa), and by the Netherlands Organisation for Scientific Research (NWO) under Grant No.\ 680.92.18.05 and NGF.1582.22.009.

\appendix

\section{Common basis construction and proofs of supporting lemmas}
\label{app:commonbasisproofs}

\subsection{Common basis construction}

This appendix collects the construction underlying Theorem~\ref{thm:commonbasis}, followed by proofs of the remaining lemmas in Section~\ref{sec:results_marcos}.

\begin{lemma}[Spectral structure of $\sigma^{-1}\rho$]\label{lem:lams}
    Let $\rho\in\scrS(\frH)$ and let $\sigma\in\scrS_+(\frH)$ be faithful. Then $\sigma^{-1}\rho$ is diagonalizable and its $n$ eigenvalues, counted with multiplicity, are exactly the $\theta_i$ of \eqref{eq:spectral-geometric}. Equivalently, there is a basis $\{v_i\}_{i=1}^n$ of $\frH$ with
    \begin{equation}\label{eq:pencil}
        \rho\ket{v_i} = \theta_i\,\sigma\ket{v_i},\qquad i=1,\dots,n.
    \end{equation}
\end{lemma}
\begin{proof}
    By \eqref{eq:threeforms}, $\sigma^{-1}\rho=\sigma^{-1/2}\eta\,\sigma^{1/2}$ is similar to $\eta$,
    which is Hermitian and positive semi-definite, hence diagonalizable with real non-negative
    eigenvalues; so the two have the same eigenvalues with the same multiplicities. Algebraic and geometric multiplicities therefore agree, and \eqref{eq:pencil} is a restatement of the eigenvalue equation, with
    $\ket{v_i}=\sigma^{-1/2}\ket{w_i}$.
\end{proof}

Before defining the pure states that will compose the common basis, we propose a scheme to ensure that the chosen representatives of the eigenspaces of $\sigma^{-1}\rho$ also form a dual basis for the common basis.

\begin{lemma}[Gram--Schmidt]\label{lem:gramschmidt}
    Let $\{w_{i,j}\}_j$ be a normalized basis of an eigenspace of $\sigma^{-1}\rho$ with eigenvalue $\theta_i$. Then a normalized basis $\{u_{i,j}\}_j$ of the same eigenspace can be constructed for which $\braket{u_{i,j}|\sigma|u_{i,l}}=0$ for $j\neq l$.
\end{lemma}
\begin{proof}
     Let $\{w_{i,j}\}_{j=1}^k$ be a $k$-dimensional, normalized basis of the eigenspace of $\sigma^{-1}\rho$ for a given eigenvalue $\theta_i$. Define $u_{i,j}$ recursively as
    \begin{equation}
    \begin{aligned}
            &\ket{u_{i,1}} = \ket{w_{i,1}},\\
            &\ket{\hat u_{i,l}} = \ket{w_{i,l}}  - \sum_{j=1}^{l-1}\frac{\braket{u_{i,j}|\sigma|w_{i,l}}}{\braket{u_{i,j}|\sigma|u_{i,j}}}\ket{u_{i,j}}, \quad \quad \ket{u_{i,l}}= \frac{\ket{\hat u_{i,l}}}{\| \ket{\hat u_{i,l}} \|} \quad \forall \, l \in \{2,...,k\},
        \end{aligned}
        \end{equation}
with $\|\cdot\|$ the norm induced by the inner product. Then $\{u_{i,j}\}_{j=1} ^k$ is linearly independent, and each element is an eigenvector of $\sigma^{-1}\rho$ with eigenvalue $\theta_i$. Therefore, $\text{span}(\ket{u_{i,1}},...,\ket{u_{i,k}})=\text{span}(\ket{w_{i,1}},...,\ket{w_{i,k}})$. We show that by construction
\begin{equation}
        \braket{u_{i,1}|\sigma|u_{i,2}} = \frac{1}{\| \ket{\hat u_{i,2}} \|}\left( \braket{u_{i,1}|\sigma|w_{i,2}}  -\frac{\braket{u_{i,1}|\sigma|w_{i,2}}}{\braket{u_{i,1}|\sigma|u_{i,1}}}\braket{u_{i,1}|\sigma|u_{i,1}}\right)=0.
    \end{equation}

\noindent The remaining cases $\braket{u_{i,j}|\sigma|u_{i,l}}=0$ for $j\neq l$ follow similarly by induction on the lower indices.
\end{proof}

\begin{lemma}[Linear independence]
\label{lem:basis}
Take $\{u_i\}_{i=1}^n$ to be the eigenvectors of $\sigma^{-1}\rho$ where the eigenvectors for eigenvalues with algebraic multiplicity $k<n$ are obtained according to Lemma~\ref{lem:gramschmidt}. Define the pure states $\{\psi_i\}_{i=1}^n\subset \scrP(\frH)$ in the following way
\begin{align}
    |\psi_i\rangle := \frac{\sigma |u_i\rangle}{\| \sigma |u_i\rangle\|} \quad \forall \, i \in \{1,...,n\}.
\end{align}
Then $\{u_i\}_{i=1}^n$ is a basis of $\frH$, hence so is $\{\psi_i\}_{i=1}^n$ because $\sigma$ is invertible, and $\braket{\psi_j|u_i} = 0$ for all $i \neq j \in \{1,...,n\}$, so that $\ket{\psi_i^\perp}=\bigl(\|\sigma\ket{u_i}\|/\braket{u_i|\sigma|u_i}\bigr)\ket{u_i}$ is the dual basis of $\{\psi_i\}_{i=1}^n$.
\end{lemma}
\begin{proof}
Fix some $i, k$, with $i\neq k$. Using the definition of $\psi_k$
\begin{equation}
    \braket{\psi_k|u_i} = \frac{ \braket{u_k|\sigma |u_i}}{\| \sigma \ket{u_k}\|}.
\end{equation}
First consider the case $\theta_i = \theta_k$. By the construction in Lemma~\ref{lem:gramschmidt} it follows that $\braket{u_k|\sigma|u_i}=0$. Assume now $\theta_i\neq \theta_k$. From \eqref{eq:pencil}, $\braket{u_k|\rho|u_i}=\theta_i\braket{u_k|\sigma|u_i}$, while conjugating the same relation for $u_k$ and using that $\rho,\sigma$ are Hermitian and $\theta_k\in\R$ gives $\braket{u_k|\rho|u_i}=\theta_k\braket{u_k|\sigma|u_i}$. Hence, $(\theta_i-\theta_k)\braket{u_k|\sigma|u_i}=0$, so $\braket{u_k|\sigma|u_i}=\braket{\psi_k|u_i}=0$. To show that $\{u_j\}_{j=1}^n$ is linearly independent, we argue by contradiction. Fix some $j \in \{1,...,n\}$ and assume that there are scalars $\alpha_i\in \mathbb{C}$, with at least one of them non-zero, such that $\ket{u_j} = \sum_{i=1, i\neq j}^n \alpha_i \ket{u_i}$. Choose $k \in \{1,...,n\}$ such that $\alpha_k \neq 0$. Then,
\begin{equation}
        0 = \braket{\psi_k | u_j} = \sum_{i=1, i\neq j}^n \alpha_i \braket{\psi_k|u_i} = \alpha_k \braket{\psi_k|u_k} = \alpha_k \frac{\braket{u_k | \sigma|u_k}}{\| \sigma \ket{u_k}\|} \neq 0,
\end{equation}
which is a contradiction, so $\{u_j\}_{j=1}^n$ is linearly independent, and applying the invertible $\sigma$ to it shows the same for $\{\psi_i\}_{i=1}^n$.
\end{proof}

We now continue constructing a common basis and thus prove Theorem~\ref{thm:commonbasis}.

\begin{proof}[Second proof of Theorem~\ref{thm:commonbasis}]
Let $\rho\in\scrS(\frH)$ and let $\sigma\in\scrS_+(\frH)$ be faithful, with $n=\dim\frH$. By Lemma~\ref{lem:lams}, there are $n$ eigenpairs $(\theta_i,v_i)$ of $\sigma^{-1}\rho$ satisfying \eqref{eq:pencil}. To account for multiplicity, use the scheme in Lemma \ref{lem:gramschmidt} to obtain a new set of $n$ vectors $\{u_i\}_i$, and define the pure states $\ket{\psi_i}$ according to Lemma~\ref{lem:basis}.

Now consider the ansatzes
\begin{equation}
\begin{aligned}
    &\hat{\sigma} =\sum_{i=1} ^n \sigma_i |\psi_i\rangle\langle\psi_i|\,,\qquad \,\sigma_i:=\frac{\| \sigma |u_i\rangle\|^2}{\braket{u_i|\sigma|u_i}}\,,\\
    &\hat{\rho} =\sum_{i=1} ^n \rho_i |\psi_i\rangle\langle\psi_i|\,,\qquad \rho_i:=\frac{\| \sigma |u_i\rangle\|^2}{\braket{u_i|\sigma|u_i}^2}\braket{u_i|\rho|u_i} = \sigma_i \frac{\braket{u_i|\rho|u_i}}{\braket{u_i|\sigma|u_i}} = \theta_i\sigma_i.
\end{aligned}
\end{equation}
    We claim that $\sigma = \hat{\sigma}$, $\rho = \hat{\rho}$. According to Lemma~\ref{lem:basis}, $\{u_i\}_{i=1}^n$ forms a basis. Hence, it suffices to show that $\sigma \ket{u_k} = \hat{\sigma}\ket{u_k}$ and $\rho \ket{u_k} = \hat{\rho}\ket{u_k}$ for all $k \in \{1,...,n\}$. Using Lemma~\ref{lem:basis}, we find
\begin{align*}
        \hat{\sigma}|u_k\rangle &= \sum_{i=1} ^n \frac{\| \sigma |u_i\rangle\|^2}{\braket{u_i|\sigma|u_i}} |\psi_i\rangle\langle\psi_i|u_k\rangle = \frac{\| \sigma |u_k\rangle\|^2}{\braket{u_k|\sigma|u_k}} |\psi_k\rangle\langle\psi_k|u_k\rangle \\
        &= \frac{\| \sigma |u_k\rangle\|^2}{\braket{u_k|\sigma|u_k}} \frac{\sigma|u_k\rangle}{\| \sigma |u_k\rangle\|} \frac{\braket{u_k|\sigma|u_k}}{\| \sigma |u_k\rangle\|} = \sigma |u_k\rangle.
    \end{align*}
    Similarly, using $\rho\ket{u_k}=\theta_k\,\sigma\ket{u_k}$, we arrive at
\begin{equation}
        \hat{\rho}|u_k\rangle = \sum_{i=1} ^n \frac{\| \sigma |u_i\rangle\|^2}{\braket{u_i|\sigma|u_i}^2}\braket{u_i|\rho|u_i} |\psi_i\rangle\langle\psi_i|u_k\rangle =\frac{\braket{u_k|\rho|u_k}}{\braket{u_k|\sigma|u_k}} \sigma |u_k\rangle=\rho|u_k\rangle.
    \end{equation}
    Clearly, $(\rho_i)_i, (\sigma_i)_i\in \R^n$ are probability vectors, as $\rho_i,\sigma_i\ge 0$ for every $i\in\{1,\ldots,n\}$ and
\begin{equation}
        1 = \Tr[\rho]  = \Tr\left[\sum_{i=1} ^n \rho_i\ket{\psi_i}\bra{\psi_i} \right]=\sum_{i=1} ^n \rho_i,
    \end{equation}
    and similarly for $(\sigma_i)_i$. Thus, $\{\psi_i\}_{i=1}^n$ is a common basis of $\rho$ and $\sigma$.
\end{proof}

\subsubsection*{Agreement with the spectral construction}

By the proof of Lemma~\ref{lem:lams} the eigenvectors of $\sigma^{-1}\rho$ are $\ket{v_i}=\sigma^{-1/2}\ket{w_i}$, so Lemma~\ref{lem:basis} produces $\ket{\psi_i}\propto\sigma\sigma^{-1/2}\ket{w_i}=\sigma^{1/2}\ket{w_i}$, which is the basis of \eqref{eq:cbspectral}. Writing $\mathcal{N}_i:=\|\sigma^{-1/2}\ket{w_i}\|$ for the normalization, so that $\ket{u_i}=\sigma^{-1/2}\ket{w_i}/\mathcal{N}_i$, one has $\braket{u_i|\sigma|u_i}=\mathcal{N}_i^{-2}$ and $\|\sigma\ket{u_i}\|^2=\mathcal{N}_i^{-2}\braket{w_i|\sigma|w_i}$, whence the ansatz above returns $\sigma_i=\braket{w_i|\sigma|w_i}$ and $\rho_i=\theta_i\sigma_i$, again the weights of \eqref{eq:cbspectral}. Finally, $\braket{u_i|\sigma|u_j}=(\mathcal{N}_i\mathcal{N}_j)^{-1}\braket{w_i|w_j}=0$ for $i\neq j$, so the $\sigma$-orthogonality of Lemma~\ref{lem:gramschmidt} holds automatically for any orthonormal choice within a degenerate eigenspace.

\subsection{Proof of Lemma~\ref{lem:cbstructure}}

	$(i)$ Suppose a probability measure with support on the pure states $\{\phi_i\}_{i=1}^m$ realizes $\rho$, so that with coefficients $\rho_i>0$
\begin{equation}
        \rho = \sum_{i=1}^m \rho_i \ket{\phi_i}\bra{\phi_i}.
    \end{equation}
    Every vector in the range of $\rho$ is then a combination of the $\ket{\phi_i}$, so $m\geq\rank\rho$, and the eigen-decomposition of $\rho$ attains this bound. The same holds for a measure $\mu$ that is not finitely supported: if $\braket{v|\rho|v}=0$, then $\braket{v|P|v}=0$ for $\mu$-almost every $P$, so the range of $\rho$ lies in the span of the pure states in the support of $\mu$, which therefore contains $\rank\rho$ linearly independent ones. We are left to show that $\mu_\CB$ does too. By \eqref{eq:cbspectral} its weights are $\rho_i=\theta_i\sigma_i$ with $\sigma_i>0$, so the number of non-zero ones is $\#\{i:\theta_i>0\}$, which is the rank of $\eta$ and hence of $\rho$, since the two are congruent.

$(ii)$ Let $\{\psi_i\}_{i=1}^n$ and $\{\phi_i\}_{i=1}^n$ be two common bases of $(\rho,\sigma)$, with weights $(\rho_i^\psi,\sigma_i^\psi)$ and $(\rho_i^\phi,\sigma_i^\phi)$. By the correspondence following the proof of Theorem~\ref{thm:commonbasis}, the vectors
\[
  \ket{e_i}:=(\sigma_i^\psi)^{1/2}\,\sigma^{-1/2}\ket{\psi_i}
  \qquad\text{and}\qquad
  \ket{f_i}:=(\sigma_i^\phi)^{1/2}\,\sigma^{-1/2}\ket{\phi_i}
\]
form two orthonormal eigenbases of $\eta=\sigma^{-1/2}\rho\,\sigma^{-1/2}$, with eigenvalues $\rho_i^\psi/\sigma_i^\psi$ and $\rho_i^\phi/\sigma_i^\phi$, respectively. Both lists of eigenvalues are the $\theta_i$, which are pairwise distinct. Thus, every eigenspace of $\eta$ is one-dimensional. The matching eigenvalues define a permutation $\pi$ such that $\ket{e_i}=\beta_i\ket{f_{\pi(i)}}$ for some $\beta_i\in\mathbb{C}$, $|\beta_i|=1$. Applying $\sigma^{1/2}$ then yields $\ket{\psi_i}=c_i\ket{\phi_{\pi(i)}}$ with $c_i=\beta_i\bigl(\sigma_{\pi(i)}^\phi/\sigma_i^\psi\bigr)^{1/2}$. Since $\ket{\psi_i}$ and $\ket{\phi_{\pi(i)}}$ are both unit vectors, we have $|c_i|=1$. Setting $\alpha_i:=c_i$ then proves the claim.

\subsection{Proof of Lemma~\ref{lem:inverse}}
\label{app:inverse}
   Let $\hat\sigma^{-1} := \sum_{i=1}^n \frac{1}{\sigma_i}|\psi_i ^\perp\rangle\langle\psi_i^\perp|$. It suffices to show that $\sigma \hat\sigma^{-1}=\hat\sigma^{-1} \sigma = I$. Since $\{\psi_i^\perp\}_{i=1}^n$ is the dual basis of $\{\psi_i\}_{i=1}^n$, it follows that $\braket{\psi_i|\psi_j^\perp}=\delta_{ij}$. Then,
    \begin{align}
        \sigma \hat\sigma^{-1} = \sum_{i=1}^n \sum_{j=1}^n \frac{\sigma_i}{\sigma_j}\ket{\psi_i}\braket{\psi_i|\psi_j^\perp}\bra{\psi_j^\perp} = \sum_{i=1}^n \ket{\psi_i}\bra{\psi_i^\perp}.
    \end{align}
    In particular,
    \begin{align}
        \sigma \hat\sigma^{-1}\ket{\psi_i} = \sum_{j=1}^n \ket{\psi_j}\braket{\psi_j^\perp|\psi_i} = \ket{\psi_i}\qquad \text{for $i \in \{1,...,n\}$.}
    \end{align}
    Since $\{\psi_i\}_{i=1}^n$ is a basis, $\sigma \hat\sigma^{-1}=I$. A similar argument holds for $\hat\sigma^{-1} \sigma =I$ using the fact that $\{\psi_i^\perp\}_{i=1}^n$ is also a basis.

\subsection{Lipschitz contraction of $\Lambda$}\label{app:proof:continuitylambda}
\begin{lemma}[Lipschitz contraction of $\Lambda$]
\label{lem:continuitylambda}
    The map $\Lambda: (\mathcal{P}(\mathscr{P}(\frH)), d_{\mathtt{W}})\rightarrow (\scrS(\frH), d_{\mathtt{TR}})$ is Lipschitz continuous when $\mathscr{P}(\frH)$ is equipped with the Fubini--Study metric $d_{\mathtt{FS}}$. Here $d_{\mathtt{W}}$ denotes the 1-Wasserstein metric \cite{Santambrogio2015OptimalTF}.
\end{lemma}
\begin{proof}
Consider the map $\Lambda: (\mathcal{P}(\mathscr{P}(\frH)), d_{\mathtt{W}})\rightarrow (\scrS(\frH), d_{\mathtt{TR}})$ given by
\begin{equation}
    \Lambda(\mu) = \int_{\scrP(\frH)} P \,\mu(\di P).
\end{equation}
Equip $\scrP(\frH)$ with the Fubini--Study metric $d_{\mathtt{FS}}$ and let $\mu, \nu \in \mathcal{P}(\scrP(\frH))$. For any coupling $\pi$ of $\mu$ and $\nu$, we have that
\begin{align*}
     d_{\mathtt{TR}}(\Lambda(\mu),\Lambda(\nu)) &\le \iint_{\scrP(\frH){\times}\scrP(\frH)}  d_{\mathtt{TR}}(P,Q)\,\pi(\di P,\di Q) \\
     &= \iint_{\scrP(\frH){\times}\scrP(\frH)} \sqrt{1 - \tr[PQ]}\,\pi(\di P,\di Q)  \\
     &= \iint_{\scrP(\frH){\times}\scrP(\frH)} \sin(d_{\mathtt{FS}}(P,Q))\,\pi(\di P,\di Q) 
     \le \iint_{\scrP(\frH){\times}\scrP(\frH)} d_{\mathtt{FS}}(P,Q)\,\pi(\di P,\di Q).
\end{align*}
Therefore, infimizing over all couplings of $\mu$ and $\nu$ yields
\[
    d_{\mathtt{TR}}(\Lambda(\mu),\Lambda(\nu)) \le d_{\mathtt{W}}(\mu,\nu).
\]
Therefore, $\Lambda$ is a Lipschitz contraction. 
\end{proof}

\subsection{Proof of Lemma~\ref{lem:discretization}}\label{app:proof:discretization}
Let $\mu, \nu\in\calP(\scrP(\frH))$ be two probability measures realizing $\rho\in\scrS(\frH)$ and a faithful $\sigma \in \scrS_+(\frH)$ respectively. To obtain a discrete approximation of the measures, we construct a sequence of partitions shrinking in diameter. Consider an open cover of $\scrP(\frH)$ consisting of balls of the form $\{B(P, 1/k), P \in \scrP(\frH)\}$ for $k\in\mathbb{N}$. By compactness of $\scrP(\frH)$ (see Proposition~\ref{prop:fubinistudy}), the open cover has a finite subcover, denoted by $\{B_{k,j}\}_{1\leq j \leq J_k}$. For each $k\in\mathbb{N}$, we let $\mathcal{A}_k := \{A_{k,j}\}_{1\leq j \leq J_k}$ be a family of sets, where
\begin{equation}
        A_{k,1} = B_{k,1},\quad 
        A_{k,j} = B_{k,j} \backslash{\bigcup_{i=1}^{j-1} A_{k,i}}, \quad \text{for } 1<j\leq J_k.
\end{equation}
By construction, $\mathcal{A}_k$ is a covering of $\scrP(\frH)$ by mutually disjoint $d_\mathtt{FS}$-Borel sets of diameter at most $2/k$. Note that if $A_{k,j}=\varnothing$, $J_k$ can be decreased. Let $P_{k,j}\in\scrP(\frH)$ be arbitrary states in $A_{k,j}$ for all $1\leq j\leq J_k$, $k\in\mathbb{N}$. Define sequences of discrete measures $(\mu_k)_{k\in \mathbb{N}}$ and $(\nu_k)_{k \in \mathbb{N}}\subset\calP(\scrP(\frH))$ by
\begin{equation}
        \mu_k := \sum_{j=1}^{J_k}  \mu(A_{k,j}) \delta_{P_{k,j}}, \qquad 
        \nu_k := \sum_{j=1}^{J_k}  \nu(A_{k,j}) \delta_{P_{k,j}},
\end{equation}
and correspondingly, the sequence of states $(\rho_k:=\Lambda(\mu_k))_{k\in \mathbb{N}}$ and $(\sigma_k:=\Lambda(\nu_k))_{k \in \mathbb{N}}$.

\medskip
Now we prove $(i)$ by showing $d_\mathtt{TR}(\rho_k,\rho) \to 0$ and $d_\mathtt{TR}(\sigma_k,\sigma) \to 0$. First, we equip $\mathcal{P}(\scrP(\frH))$ with the 1-Wasserstein metric $d_{\mathtt{W}}$. Consider the sequence of transport maps $T_k: \scrP(\frH) \rightarrow \scrP(\frH) {\times} \scrP(\frH)$ given by $T_k(P) = (P, P_{k,j})$ for all $P \in A_{k,j}$ and $k \in \mathbb{N}$. Then, define the coupling $\pi_k = (T_k)_\# \mu$ as the push-forward measure of $\mu$ under $T_k$. Then,
\begin{equation}
    d_{\mathtt{W}}(\mu, \mu_k) \leq \iint_{\scrP(\frH){\times}\scrP(\frH)} d_{\mathtt{FS}}(P,Q)\,\pi_k(\di P\,\di Q) = \sum_{j=1}^{J_k}\int_{A_{k,j}} d_{\mathtt{FS}}(P, P_{k,j}) \mu(\di P).
\end{equation}
By construction, for all $P \in A_{k,j}$ we have $d_{\mathtt{FS}}(P, P_{k,j}) \leq 2/k$, for every $j \leq J_k$. Thus, $d_{\mathtt{W}}(\mu, \mu_k) \leq 2/k$. Combining this with Lemma~\ref{lem:continuitylambda} we find, $d_{\mathtt{TR}}(\rho, \rho_k) = d_{\mathtt{TR}}(\Lambda(\mu), \Lambda(\mu_k)) \leq d_{\mathtt{W}}(\mu, \mu_k)\leq 2/k$. Taking the limit $k \to \infty$, it follows that $\rho_k \to \rho$ and, by the same argument, $\sigma_k \to \sigma$. By lower semi-continuity of the maximal $f$-divergence $\fdiv^{\max}$ it directly follows
\begin{equation}
    \liminf_{k\to \infty} \fdiv^{\max}(\rho_k \| \sigma_k) \geq \fdiv^{\max}(\rho \| \sigma).
\end{equation}
Next we prove $(ii)$. Consider the sequence of maps $\sfK_k: \scrP(\frH){\times}\calF\to [0,1]$ given by
\begin{equation}
    \sfK_k(Q,A) := \sum_{j=1}^{J_k}  \mathbbm{1}_{A_{k,j}}(Q)\, \delta_{P_{k,j}}(A).
\end{equation}
We will argue that these maps are Markov kernels (see Definition~\ref{def:markovkernel}). Fix some $Q\in \scrP(\frH)$, then $\sfK_k(Q,\cdot) = \delta_{P_{k,j}}(\cdot)$ for exactly one $j$ such that $Q \in A_{k,j}$ because  $\mathcal{A}_k$ is a partition. It directly follows that $\sfK_k(Q,\scrP(\frH)) = 1$. Hence, for all $Q\in \scrP(\frH)$, $\sfK_k(Q,\cdot)$ is a probability measure. Now fix a pair $(P,C) \in \scrP(\frH)\times\calF$. As  $\mathcal{A}_k$ is a partition, there is exactly one $j$ for which $P\in A_{k,j}$. If for this $j$ we have $P_{k,j} \in C$, then $\sfK_k(P,C) = 1$. If $P_{k,j} \not\in C$, then $\sfK_k(P,C) = 0$. Both of these outcomes are $\mathcal{B} _{[0,1]}$-measurable preimages. Therefore, $\sfK_k(\cdot,C)$ is measurable and hence $\sfK_k$ is a Markov kernel.
Using the previously introduced notation
\begin{equation}
    \begin{aligned}
        \sfK_k{\circ} \mu &= \int_{\scrP(\frH)} \sfK_k(Q,\cdot) \mu(\di Q) = \int_{\scrP(\frH)}\sum_{j=1} ^{J_k} \mathbbm{1}_{A_{k,j}}(Q) \delta_{P_{k,j}} \mu(\di Q)= \sum_{j=1}^{J_k}  \mu(A_{k,j}) \delta_{P_{k,j}} = \mu_k.
    \end{aligned}
\end{equation}
Similarly, $\nu_k = \sfK_k {\circ} \nu$. By the classical DPI (Proposition~\ref{prop:finitedpi})
\begin{equation}
    \fdiv(\mu \| \nu) \ge \fdiv(\sfK_k\circ\mu\|\sfK_k\circ\nu)= \fdiv(\mu_k \| \nu_k), \qquad \forall \, k \in \mathbb{N}.
\end{equation}
Therefore, it also follows that $\fdiv(\mu \| \nu) \geq \limsup_{k \to \infty}\fdiv(\mu_k \| \nu_k)$, proving $(ii)$.

\medskip
Finally, we prove $(iii)$. Let $\mathbb{C}^{J_k}$ be a $J_k$-dimensional Hilbert space. Let $\{\phi_{k,j}\}_{1\leq j \leq J_k}$ be an orthonormal basis. Define the sequence $(\hat\rho_k)_{k\in \mathbb{N}}$ by $\hat\rho_k = \sum_{j=1}^{J_k}  \mu(A_{k,j} )\ket{\phi_{k,j}}\bra{\phi_{k,j}} $. Define $(\hat\sigma_k)_{k\in \mathbb{N}}$ similarly. Because $\hat{\rho}_k$ commutes with $\hat{\sigma}_k$, and every maximal $f$-divergence reduces to the classical $\fdiv$ of the common eigenvalues on commuting states, we find $\fdiv(\mu_k\|\nu_k)=\fdiv^{\max}(\hat{\rho}_k\|\hat{\sigma}_k)$. Now, define the sequence of CPTP maps $(\Xi_k:\mathbb{C}^{J_k}\to \frH)_{k\in \mathbb{N}}$ by the Kraus operators $K_{k,j} = \ket{\psi_{k,j}}\bra{\phi_{k,j}} $, where $\ket{\psi_{k,j}}$ is a unit vector representing $P_{k,j}$. Then
\begin{equation}
\begin{aligned}
        &\sum_{j=1}^{J_k}  K_{k, j}^\dagger K_{k,j} = \sum_{j=1}^{J_k} \ket{\phi_{k,j}}\bra{\phi_{k,j}} = I_{J_k},\\
        \Xi_k(\hat\rho_k) = &\sum_{j=1}^{J_k}  K_{k,j}\hat{\rho}_kK_{k,j}^\dagger=\sum_{j=1}^{J_k}  \mu(A_{k,j} )P_{k,j}  = \rho_k,
\end{aligned}
\end{equation}
where $I_{J_k}$, is the $J_k$ dimensional identity. Because $\fdiv^{\max}$ satisfies the quantum DPI (cf.\ Proposition~\ref{prop:quantumdpi}), it follows that $\fdiv(\mu_k\|\nu_k)=\fdiv^{\max} (\hat\rho_k \| \hat\sigma_k) \geq \fdiv^{\max}(\rho_k \| \sigma_k) $ for every $k \in \mathbb{N}$.

\section{Motivating the infimum: a large deviation principle}
\label{sec:largedev}
Quantum versions of Stein's lemma and Sanov's theorem identify the Umegaki relative entropy as the optimal error exponent for discriminating two quantum states by measurements \cite{steinslemma1,quantumstein2,quantumsanov1}. It is natural to ask whether the Belavkin--Staszewski relative entropy plays a similar role when pure states are sampled directly from ensembles that realize the states. The result below is a first step in this direction, but of a different nature: it concerns the empirical measure of pure states sampled from an ensemble realizing $\sigma$, not outcomes of measurements performed on $\sigma$. We do not establish a hypothesis-testing interpretation here, and the relation to the quantum Stein setting is left open.

The KL divergence plays a central role in hypothesis testing and large deviations theory. Sanov's theorem establishes it as the exponential decay rate of the probability that the empirical distribution of a system deviates from the true, underlying distribution \cite{Sanov}. Below we present a corollary of our results related to large deviation theory, for which the following two propositions are necessary.

\begin{proposition}[Sanov's theorem]
\label{prop:sanov}
    Let $(\frX, \mathcal{F})$ be a measurable space, where $\frX$ is a Polish space. Let $\nu \in \mathcal{P}(\frX)$. Consider a sequence of i.i.d. random variables $X_1,X_2,...$ on $\frX$ with law $\nu$. Let $\nu_N$ be the empirical measure of the sample $X_1,...,X_N$, given by $\nu_N = \frac{1}{N}\sum_{i=1}^N \delta_{X_i}$. The following Large Deviation Principle (LDP) holds
    \begin{gather}
     \liminf_{N\rightarrow \infty}\frac{1}{N}\log \mathbb{P}(\nu_N\in O)\ge -\inf_{\mu \in O} \KL(\mu\Vert \nu),\\
     \limsup_{N\rightarrow \infty}\frac{1}{N}\log \mathbb{P}(\nu_N\in C)\le -\inf_{\mu \in C}\KL(\mu\Vert \nu),
\end{gather}
    where $O$ denotes an open set in $\mathcal{P}(\frX)$, equipped with the weak topology, and $C$ denotes a closed one. Note that $\nu_N$ is a random variable by virtue of the sampling.
\end{proposition}

The Contraction Principle formalizes how a continuous function extends an LDP on one space to another \cite{Hollander_2008}. We state this result for the empirical measure in Proposition~\ref{prop:sanov}.

\begin{proposition}[Contraction Principle]
\label{prop:contractionprinciple}
    Let $\frX_1,\frX_2$ be Polish spaces. Let $\nu_N$ be a sequence of $\mathcal{P}(\frX_1)$-valued random variables that satisfies an LDP with rate $N$ and rate function $I$. If $\varGamma: \mathcal{P}(\frX_1) \to \frX_2$ is a continuous map, then for any Borel sets $O\subset \frX_2$ open and $C\subset\frX_2$ closed, the sequence $\varGamma(\nu_N)$ satisfies the LDP
        \begin{gather}
     \liminf_{N\rightarrow \infty}\frac{1}{N}\log \mathbb{P}(\varGamma(\nu_N)\in O)\ge - \inf_{y \in O}\inf_{\mu \in \varGamma^{-1}(y)} I(\mu),\\
     \limsup_{N\rightarrow \infty}\frac{1}{N}\log \mathbb{P}(\varGamma(\nu_N)\in C)\le - \inf_{y \in C}\inf_{\mu \in \varGamma^{-1}(y)} I(\mu).
\end{gather}
    Moreover, $J(y) = \inf_{\mu \in \varGamma^{-1}(y)} I(\mu)$ is lower semi-continuous.
\end{proposition}

\begin{corollary}[Large deviations Belavkin--Staszewski]\label{cor:largedeviations}
Equip $\calP(\scrP(\frH))$ with $d_\mathtt{W}$ and $\scrS(\frH)$ with $d_\mathtt{TR}$. Let $\rho,\sigma\in\scrS_+(\frH)$ be two faithful states. Let $\nu \in \calP(\scrP(\frH))$ be a probability measure realizing $\sigma$ and consider the empirical measure $\nu_N = \frac{1}{N}\sum_{i=1}^N \delta_{P_i}$, where $P_1,\dots,P_N$ are drawn independently from $\nu$. Suppose there is a $\mu^* \in\calP(\scrP(\frH)) $ realizing $\rho$ such that $\KL(\mu^* \Vert \nu) < \infty$. Then,
\begin{enumerate}[label=(\roman*)]    
\item the following large deviation principle holds
\begin{equation}
    \forall\epsilon>0,\quad\lim_{N \to \infty} \frac{1}{N}\log \mathbb{P}(\Lambda(\nu_N) \in B_\epsilon(\rho)) = -\inf_{\varrho \in B_\epsilon (\rho)} J(\varrho),\qquad J(\varrho) := \inf_{\mu \in \Lambda^{-1}(\varrho)} \KL(\mu\|\nu),
\end{equation}
where $B_\epsilon(\rho) = \{\varrho \in \scrS(\frH): d_{\mathtt{TR}}(\varrho , \rho) < \epsilon\}$;
\item the limit of $(i)$ gives
\begin{equation}
    \lim_{\epsilon \to 0} \lim_{N \to \infty} \frac{1}{N}\log \mathbb{P}(\Lambda(\nu_N)\in B_\epsilon(\rho))  \leq - \BS(\rho \Vert \sigma),
\end{equation}
with equality when $\nu = \nu_{\CB}$;
\item when the support of $\nu$ is a set of $n=\dim\frH$ pure states $\{Q_i\}_{i=1}^n\subset\scrP(\frH)$ whose representatives form a basis of $\frH$, the rate function reduces to
\begin{equation}
    J(\varrho) = \begin{cases}
        \BS(\varrho \Vert \sigma) & \text{when }\{Q_i\}_{i=1}^n \in \CB(\varrho, \sigma) ,\\
        + \infty & \text{else}.
    \end{cases}
\end{equation}
\end{enumerate}
\end{corollary}

\begin{proof}
First, note that the space of probability measures on the pure states $\calP(\scrP(\frH))$ is Polish when equipped with the topology of weak convergence. Since $(\scrP(\frH), d_\mathtt{FS})$ is a compact space, the $d_{\mathtt{W}}$ metric metricizes the weak topology \cite{Santambrogio2015OptimalTF}. By Sanov's theorem Proposition~\ref{prop:sanov}, $\nu_N$ satisfies the LDP with lower semi-continuous rate function $I(\mu)=\KL(\mu \Vert \nu)$. By Lemma~\ref{lem:continuitylambda}, $\Lambda$ is continuous. Therefore, by the contraction principle of Proposition~\ref{prop:contractionprinciple} (with $\varGamma=\Lambda$), $\Lambda(\nu_N)$ satisfies the LDP with lower semi-continuous rate function $J(\varrho) = \inf_{\mu \in \Lambda^{-1}(\varrho)} \KL(\mu \Vert\nu)$.
Fix some $\epsilon >0$. Let $B_\epsilon(\rho)$ be the open ball and $\overline{B}_\epsilon(\rho)$ be its closure. It follows directly that
\begin{align*}
     -\inf_{\varrho \in B_\epsilon(\rho)} J(\varrho) &\leq \liminf_{N\rightarrow \infty}\frac{1}{N}\log \mathbb{P}(\Lambda(\nu_N)\in B_\epsilon(\rho)) \\
     &\leq \limsup_{N \to \infty}\frac{1}{N}\log \mathbb{P}(\Lambda(\nu_N)\in \overline{B}_\epsilon(\rho)) \leq -\inf_{\varrho \in \overline{B}_\epsilon(\rho)} J(\varrho).
\end{align*}
To establish the rate function in $(i)$, we show that $a(\epsilon):=\inf_{\varrho \in B_\epsilon(\rho)} J(\varrho) = \inf_{\varrho \in \overline{B}_\epsilon(\rho)} J(\varrho)=:b(\epsilon)$. Note that it suffices to show $a(\epsilon) \le b(\epsilon)$, since the reverse inequality is trivial.

\medskip
Since $\Lambda:\mathcal{P}(\mathscr{P}(\frH))\rightarrow\scrS(\frH)$ is a continuous map from a compact space to a compact space, the preimage $\Lambda^{-1}(\rho)$ is also compact \cite{prokhorov}. Since lower semi-continuous functions always achieve their infimum on compact sets, there exists a measure $\mu_{\rho}\in\Lambda^{-1}(\rho)$ such that
\begin{equation}
a(\epsilon),b(\epsilon)\leq \KL(\mu_{\rho} \| \nu)=J(\rho)=\inf_{\mu\in\Lambda^{-1}(\rho)
}\KL(\mu\|\nu) \leq \KL(\mu^*\|\nu)< +\infty.
\end{equation}
Let $\varrho \in \overline{B}_{\epsilon}(\rho)$. If $J(\varrho) = +\infty$, then clearly $a(\epsilon) \le J(\varrho)$ .
Assume $J(\varrho) < +\infty$. Then, again by lower semi-continuity and compactness, there exists $\mu_{\varrho} \in \calP(\scrP(\frH))$ such that $J(\varrho) = \KL(\mu_{\varrho} \| \nu) < +\infty$. Let $t \in (0, 1)$ and define $\mu_t = (1-t)\mu_{\varrho} + t\mu_{\rho}$.
Then linearity of $\Lambda$ gives
\begin{equation}
    d_{\mathtt{TR}}(\Lambda(\mu_t), \rho) = d_{\mathtt{TR}}((1-t)\varrho + t\rho, \rho) = (1-t) d_{\mathtt{TR}}(\varrho, \rho) < \epsilon.
\end{equation}
Thus, $\Lambda(\mu_t) \in B_{\epsilon}(\rho)$, $t \in (0, 1)$, which together with convexity of the relative entropy implies
\begin{equation}
    a(\epsilon) = \inf_{\varrho' \in B_{\epsilon}(\rho)} J(\varrho') \le J(\Lambda(\mu_t)) \le \KL(\mu_t \| \nu)\leq (1-t)J(\varrho) + tJ(\rho).
\end{equation}
Taking the limit as $t \to 0$, we recover $a(\epsilon) \le J(\varrho)$. This holds for any $\varrho \in \overline{B}_{\epsilon}(\rho)$, thus
\begin{equation}
a(\epsilon) \le \inf_{\varrho \in \overline{B}_{\epsilon}(\rho)} J(\varrho) = b(\epsilon).
\end{equation}
This proves statement $(i)$ of the corollary,
\begin{equation}
\lim_{N \to \infty} \frac{1}{N}\log \mathbb{P}(\Lambda(\nu_N) \in B_\epsilon(\rho)) = -\inf_{\varrho \in B_\epsilon (\rho)} J(\varrho) .
\end{equation}
We continue with proving the limit of statement $(ii)$. Clearly, we have that
\begin{equation}
    a(\epsilon) := \inf_{\varrho \in B_\epsilon(\rho)} J(\varrho) \leq J(\rho).
\end{equation}
As $a(\epsilon)$ is monotonically increasing as $\epsilon$ decreases, $\lim_{\epsilon\rightarrow0}a(\epsilon)$ exists. Let $\epsilon_k\rightarrow 0$ be a decreasing sequence of positive numbers. We find a sequence $\varrho_k\in B_{\epsilon_k}(\rho)$ with $\varrho_k\rightarrow\rho$ in trace distance and
\begin{equation}
    a(\epsilon_k)>J(\varrho_k)-\frac{1}{k}\quad\Longrightarrow\quad \lim_{k\rightarrow\infty}a(\epsilon_k)\geq\liminf_{k\rightarrow\infty}J(\varrho_k)=J(\rho),
\end{equation}
where the implication follows by lower semi-continuity of $J$. Thus, $\lim_{\epsilon\rightarrow0}a(\epsilon)=J(\rho)$, allowing us to conclude that
\begin{equation}
    \lim_{\epsilon \to 0} \lim_{N \to \infty} \frac{1}{N}\log \mathbb{P}(\Lambda(\nu_N)\in B_\epsilon(\rho)) = \lim_{\epsilon \to 0} \left( - \inf_{\varrho \in B_\epsilon(\rho)} J(\varrho) \right) = - \lim_{\epsilon \to 0} a(\epsilon) = -J(\rho).
\end{equation}
By Theorem \ref{thm:maintheorem}, $J(\rho) \ge  \BS(\rho \| \sigma)$, where equality holds when $\nu$ is the optimal measure $\nu_{\CB}$ that achieves the minimum.

\medskip
As for $(iii)$, assume $\mathrm{supp}(\nu)$ is a finite set $\{Q_i\}_{i=1}^n\subset\scrP(\frH)$ whose representatives form a basis of $\frH$. Then, for any $\mu\in\mathcal{P}(\mathscr{P}(\frH))$, $\KL(\mu \Vert\nu)< \infty \Leftrightarrow \mathrm{supp}(\mu)\subseteq\{Q_i\}_{i=1}^n $. Therefore, for such $\mu$, it holds that $\{Q_i\}_{i=1}^n $ is a common basis of $\sigma=\Lambda(\nu)$ and $\varrho=\Lambda(\mu)$. Hence,
\begin{equation}
    J(\varrho) = \begin{cases}
        \BS(\varrho \Vert \sigma) & \text{when }\{Q_i\}_{i=1}^n \in \CB(\varrho, \sigma) ,\\
        + \infty & \text{else}.
    \end{cases} \qedhere
\end{equation}
\end{proof}

\section{Temporal derivative under SSE evolution}
\label{app:lindbladflow}

We compute the time derivative of $\fdiv(\mu_t\|\nu_t)$ along the Fokker--Planck equation \eqref{eq:fokker-planck}, i.e., the estimate \eqref{eq:dissipation} used in the proof of Theorem~\ref{thm:contraction}. Fix $t>0$, where that theorem provides smooth densities $u_t,v_t$ of $\mu_t,\nu_t$ with respect to $\nu_\mathtt{FS}$. We define the following functions:
\[
  u_t^\varepsilon := u_t+\varepsilon,\qquad v_t^\varepsilon := v_t+\varepsilon,\qquad z_t^\varepsilon := u_t^\varepsilon/v_t^\varepsilon,\qquad \varepsilon>0 .
\]
By construction, $z_t^\varepsilon$ is inside a compact subinterval of $(0,\infty)$. Indeed, since the densities are continuous on the compact $\scrP(\frH)$ and hence bounded, we have $\varepsilon/(\varepsilon+\sup v_t)\leq z_t^\varepsilon\leq(\varepsilon+\sup u_t)/\varepsilon$. On such an interval $f$ is real analytic, and may be differentiated freely below. Note that neither $u_t^\varepsilon$ nor $v_t^\varepsilon$ is a probability density, but each solves the same equation as the density it regularizes, since $\scrL^*$ annihilates constants. In the following, we express $\fdiv$ in terms of the perspective function $\Upsilon_f(a,b):=b\,f(a/b)$, which gives $\fdiv(\mu\|\nu)=\int \Upsilon_f(u,v)\,\di\nu_\mathtt{FS}$.

Since $\dot u_t^\varepsilon=\scrL^*u_t^\varepsilon$ and $\dot v_t^\varepsilon=\scrL^*v_t^\varepsilon$, and $\scrL^*$ is the adjoint of $\scrL$ in $L^2(\nu_\mathtt{FS})$, we deduce
\begin{equation}
\begin{aligned}
  \frac{\di}{\di t}\int v_t^\varepsilon f(z_t^\varepsilon)\,\di\nu_\mathtt{FS}
  &=\int\Bigl[\dot v_t^\varepsilon f(z_t^\varepsilon)+f'(z_t^\varepsilon)\bigl(\dot u_t^\varepsilon-z_t^\varepsilon\dot v_t^\varepsilon\bigr)\Bigr]\di\nu_\mathtt{FS} \\
  &=\int\bigl[f(z_t^\varepsilon)-z_t^\varepsilon f'(z_t^\varepsilon)\bigr]\scrL^*v_t^\varepsilon \,\di\nu_\mathtt{FS} + \int f'(z_t^\varepsilon)\,\scrL^*u_t^\varepsilon\,\di\nu_\mathtt{FS} \\
  &=\int\Bigl\{\scrL\bigl(f(z_t^\varepsilon)\bigr)-\scrL\bigl(z_t^\varepsilon f'(z_t^\varepsilon)\bigr)+z_t^\varepsilon\,\scrL\bigl(f'(z_t^\varepsilon)\bigr)\Bigr\}\,v_t^\varepsilon \,\di\nu_\mathtt{FS}.
\end{aligned}
\end{equation}
Using the diffusion property $\scrL(g(u))=g'(u)\scrL u+g''(u)\Gamma(u,u)$, the coefficients of $\scrL z_t^\varepsilon$ cancel and those of
$\Gamma(z_t^\varepsilon,z_t^\varepsilon)$ combine to $f''-2f''-z_t^\varepsilon f'''+z_t^\varepsilon f'''=-f''$, giving
\begin{equation}\label{eq:dissipationeps}
  \frac{\di}{\di t}\int v_t^\varepsilon f(z_t^\varepsilon)\,\di\nu_\mathtt{FS}
  =-\int f''(z_t^\varepsilon)\,\Gamma(z_t^\varepsilon,z_t^\varepsilon)\,v_t^\varepsilon\,\di\nu_\mathtt{FS} \leq 0 .
\end{equation}
Integrating over an interval $[s,t]\subset(0,\infty)$, we end up with 
\[
	\int \Upsilon_f(u_t^\varepsilon,v_t^\varepsilon)\,\di\nu_\mathtt{FS} + \int_s^t\!\! \int f''(z_r^\varepsilon)\,\Gamma(z_r^\varepsilon,z_r^\varepsilon)\,v_r^\varepsilon\,\di\nu_\mathtt{FS}\,\di r \le \int \Upsilon_f(u_s^\varepsilon,v_s^\varepsilon)\,\di\nu_\mathtt{FS} .
\]
We now pass to the limit $\varepsilon\to 0$, treating the two ends differently. On the right-hand side, $\Upsilon_f$ is positively homogeneous and jointly convex, hence subadditive, so that $\Upsilon_f(u+\varepsilon,v+\varepsilon)\leq \Upsilon_f(u,v)+\varepsilon f(1) = \Upsilon_f(u,v)$ pointwise. Hence, the right-hand side is simply bounded from above by $\fdiv(\mu_s\|\nu_s)$ for every $\varepsilon>0$, with no limit to take. On the left-hand side, Fatou's lemma handles both terms. The first because $\Upsilon_f$ is lower semi-continuous, while the second because its integrand is non-negative and converges pointwise wherever $v_r>0$, which is where $\nu_r$ carries its mass. This gives \eqref{eq:dissipation} in the proof of Theorem~\ref{thm:contraction}.

For $f=f_\mathtt{KL}$ one has $f''(u)=1/u$ and $\Gamma(u,u)/u=4\,\Gamma(\sqrt u,\sqrt u)$, so that $I_f$ is the relative KL-Fisher information \eqref{eq:klfisher} and the dissipation is $2I_\mathtt{KL}(\mu_t\|\nu_t)$.

\bibliographystyle{abbrv}
\bibliography{bibliography}

\end{document}